%% file: main.tex
\documentclass{article}[11pt]
\usepackage[utf8]{inputenc}

\title{\vspace{-.7in}Principal Inertia Components and Applications}
\author{Flavio P. Calmon, Ali Makhdoumi, Muriel M\'edard, \\
Mayank Varia, Mark Christiansen, Ken R. Duffy
\thanks{F. P. Calmon is with the IBM T.J. Watson Research Center, Yorktown Heights, NY. Email: \href{mailto:fdcalmon@us.ibm.com}{{fdcalmon@us.ibm.com}}. A. Makhdoumi and M. M\'edard are with the Massachusetts Institute of Technology. Email: \{\href{mailto:makhdoum@mit.edu}{{makhdoum}}, \href{mailto:medard@mit.edu}{{medard}}\}@mit.edu. M. Varia is with Boston University. Email: \href{mailto:varia@bu.edu}{{varia@bu.edu}}. M. Christiansen is with the Automobile Association, Ireland. Email:  \href{mailto:markchristiansen4224@gmail.com}{markchristiansen4224@gmail.com}. K. R. Duffy is with the Hamilton Institute at Maynooth University. Email: \href{mailto:ken.duffy@nuim.ie}{ken.duffy@nuim.ie}. This paper was presented in part at the 51st Annual Allerton Conference  on Communication, Control, and Computing (2013), the 2014 IEEE Info. Theory Workshop, and the 2015 International Symposium on Info. Theory.
}}
\date{}

\include{preamble}

\begin{document}

\maketitle
\vspace{-.3in}
\abstract{
We explore properties and applications of the Principal Inertia Components (PICs) between two discrete random variables $X$ and $Y$. The PICs lie in the intersection of information and estimation theory, and provide a fine-grained decomposition of the dependence between $X$ and $Y$. Moreover, the PICs describe which functions of $X$ can or cannot be reliably inferred (in terms of MMSE) given an observation of $Y$. We demonstrate that the PICs play an important role in information theory, and they can be used to characterize  information-theoretic limits of certain estimation problems. In privacy settings, we prove that the PICs are related to  fundamental limits of perfect privacy.
}

\setcounter{tocdepth}{2}

\vspace{-.1in}
{
\small
\tableofcontents
}

\section{Introduction}
\input{intro}


\section{Principal Inertia Components}
\label{chap:PICs}
\input{PIC}

\section[Applications of the PICs to Information Theory]{Applications of the Principal Inertia Components to Information Theory}
\label{chap:PIC_IT}

\input{PIC_IT}

\section{Application to Estimation: Bounds on Error Probability}
\label{sec:boundEP}

\input{estimation-to-inference}

\input{PIC_Priv}

\input{finalremarks}

\section*{Acknowledgments}

The authors gratefully acknowledge Stefano Tessaro (University of California Santa Barbara), Nadia Fawaz (LinkedIn) and Yury Polyanskiy (Massachusetts Institute of Technology) for helpful and insightful discussions and feedback on the results contained in this paper. We also thank the anonymous reviewers and the Associate Editor for many helpful comments and suggestions.

\begin{appendices}
\input{appPIC}
\end{appendices}

\newpage
\bibliography{references}  
\bibliographystyle{IEEEtran}

\end{document}

%% file: preamble.tex
\usepackage[margin=1in,letterpaper]{geometry}
\usepackage[nospace,noadjust]{cite}
\usepackage[titletoc,title]{appendix}

\usepackage[
            CJKbookmarks=true,
            bookmarksnumbered=true,
            bookmarksopen=true,
            colorlinks=true,
            citecolor=red,
            linkcolor=blue,
            anchorcolor=red,
            urlcolor=blue
            ]{hyperref}
\usepackage[usenames,dvipsnames]{xcolor}
\usepackage{dsfont}
\usepackage{amsbsy}
\usepackage{amssymb}
\usepackage{amsmath}%
\usepackage{amsfonts}%
\usepackage{amsthm}
\usepackage{breqn}
\usepackage{graphicx}
\usepackage{colonequals}
\usepackage{psfrag}
\usepackage{subfigure}
\usepackage{mathrsfs}
\usepackage{hyperref}
\usepackage{mdframed}
\usepackage{enumerate}
\usepackage{setspace}
\usepackage[font={small}]{caption}

\allowdisplaybreaks

\newtheorem{thm}{Theorem}
\newtheorem{lem}{Lemma}

\newtheorem{prop}{Proposition}
\newtheorem{conj}{Conjecture}
\newtheorem{cor}{Corollary}

\makeatletter
\newtheorem*{rep@theorem}{\rep@title}
\newcommand{\newreptheorem}[2]{%
\newenvironment{rep#1}[1]{%
 \def\rep@title{#2 \ref{##1}}%
 \begin{rep@theorem}}%
 {\end{rep@theorem}}}
\makeatother
\newreptheorem{conj}{Conjecture}

\theoremstyle{definition}
 
\newtheorem{defn}{Definition} 
 
\newtheorem{remark}{Remark} 
\newtheorem{example}{Example} 

\newcommand{\etal}{\textit{et al.}~}

\newcommand{\calX}{\mathcal{X}}
\newcommand{\calA}{\mathcal{A}}
\newcommand{\calB}{\mathcal{B}}

\newcommand{\calY}{\mathcal{Y}}

\newcommand{\calJ}{\mathcal{J}}
\newcommand{\calC}{\mathcal{C}}
\newcommand{\calS}{\mathcal{S}}
\newcommand{\calU}{\mathcal{U}}

\newcommand{\calI}{\mathcal{I}}
\newcommand{\calF}{\mathcal{F}}

\newcommand{\calH}{\mathcal{H}}
\newcommand{\calP}{\mathcal{P}}
\newcommand{\calL}{\mathcal{L}}

\newcommand{\bA}{\mathbf{A}}

\newcommand{\bD}{\mathbf{D}}

\newcommand{\bB}{\mathbf{B}}
\newcommand{\bI}{\mathbf{I}}

\newcommand{\bx}{\mathbf{x}}

\newcommand{\bz}{\mathbf{z}}
\newcommand{\bX}{\mathbf{X}}

\renewcommand{\tilde}{\widetilde}

\newcommand{\Adv}{\mathsf{Adv}}

\newcommand{\Xh}{\hat{X}}

\DeclareMathOperator*{\argmin}{\arg\!\min}
\DeclareMathOperator*{\argmax}{\arg\!\max}

\newcommand{\Reals}{\mathbb{R}}

\newcommand{\normX}[2]{\| #1 \|_{#2}}
\newcommand{\normEuc}[1]{\| #1 \|_2}

\newcommand{\ones}{\mathbf{1}}

\newcommand{\defined}{\triangleq}
\newcommand{\Tr}[1]{\mathrm{ tr}\left(#1 \right)}
\newcommand{\diag}[1]{\mathrm{diag}\left( #1 \right)}
\newcommand{\pxy}{p_{X,Y}}
\newcommand{\px}{p_X}
\newcommand{\py}{p_Y}
\newcommand{\pxp}{p_{\hat{X}}}

\newcommand{\pygx}{p_{Y|X}}

\newcommand{\pxhgx}{p_{\Xh|X}}
\newcommand{\qx}{q_X}

\newcommand{\ExpVal}[2]{\mathbb{E}\left[ #2 \right]}
\newcommand{\EV}[2]{\mathbb{E}_{#1}\left[ #2 \right]}

\newcommand{\Pygx}{\mathbf{P}_{Y|X}}
\newcommand{\Pxxp}{\bP_{X,\Xh}}

\newcommand{\Px}{\mathbf{p}_X}

\newcommand{\Pxp}{\mathbf{p}_{\Xh}}
\newcommand{\Py}{\mathbf{p}_Y}
\newcommand{\At}{\tilde{\mathbf{A}}}
\newcommand{\Bt}{\tilde{\mathbf{B}}}
\newcommand{\Ut}{\tilde{\mathbf{U}}}
\newcommand{\Vt}{\mathbf{\tilde{V}}}
\newcommand{\Yt}{\tilde{Y}}

\newcommand{\lambdat}{\tilde{\lambda}}
\newcommand{\Sigmat}{\tilde{\mathbf{\Sigma}}}
\newcommand{\by}{\mathbf{y}}

\newcommand{\blambda}{\pmb{\lambda}}
\newcommand{\blambdat}{\tilde{\pmb{\lambda}}}
\newcommand{\bLambda}{\pmb{\Lambda}}
\newcommand{\Emat}{\mathbf{F}}
\newcommand{\bu}{\mathbf{u}}
\newcommand{\bv}{\mathbf{v}}

\newcommand{\tu}{\tilde{u}}
\newcommand{\tv}{\tilde{v}}
\newcommand{\olU}{\overline{\mathbf{U}}}

\newcommand{\sto}{\mbox{\normalfont s.t.}}
\newcommand{\KFnorm}[2]{\| #1 \|_{#2}}

\newcommand{\ttheta}{\tilde{\theta}}

\newcommand{\EE}[1]{\ExpVal{}{#1}}
\newcommand{\whB}{\widehat{B}}

\newcommand{\xb}{\mathbf{x}}
\newcommand{\yb}{\mathbf{y}}
\newcommand{\fb}{\mathbf{f}}
\newcommand{\gb}{\mathbf{g}}
\newcommand{\bP}{\mathbf{P}}
\newcommand{\eye}{\mathbf{I}}
\newcommand{\bQ}{\mathbf{Q}}
\newcommand{\bU}{\mathbf{U}}
\newcommand{\bSigma}{\mathbf{\Sigma}}

\newcommand{\bV}{\mathbf{V}}
\newcommand{\bT}{\mathbf{T}}
\newcommand{\bbH}{\mathbf{H}}

\newcommand{\suchthat}{\,\mid\,}

\newcommand{\mmse}{\mathsf{mmse}}

\newcommand{\bw}{\mathbf{w}}

\newcommand{\vs}{v^*(p_{S,X})}
\newcommand{\dps}{\delta(p_{S,X})}
\newcommand{\bp}{\mathbf{p}}
\newcommand{\bq}{\mathbf{q}}

\newcommand\independent{\protect\mathpalette{\protect\independenT}{\perp}}
\def\independenT#1#2{\mathrel{\rlap{$#1#2$}\mkern2mu{#1#2}}}
\newcommand{\KC}{\calJ}

\newcommand{\bg}{\mathbf{g}}
\newcommand{\Dx}{\mathbf{D}_X}
\newcommand{\Dy}{\mathbf{D}_Y}

\definecolor{light-gray}{gray}{.90}
\definecolor{aliceblue}{rgb}{0.94, 0.97, 1.0}
\definecolor{airforceblue}{rgb}{0.36, 0.54, 0.66}

\newmdenv[%
  backgroundcolor=light-gray, %
  linecolor=black,
  linewidth =1pt,%
  skipabove = 10pt,%
  skipbelow = 10pt
]{comment}

\newmdenv[%
  backgroundcolor=aliceblue, %
  linewidth = 2pt,%
  skipabove = 10pt,%
  skipbelow = 10pt,
  pstrickssetting={linestyle=dashed,},
  linecolor=airforceblue,
  middlelinewidth=2pt
]{TODO}

%% file: intro.tex

    



There is a fundamental limit to how much we can learn from data. The problem of determining which functions of a hidden variable can or cannot be estimated from a noisy observation is at the heart of estimation, statistical learning theory \cite{abu-mostafa_learning_2012}, and numerous other applications of interest. For example, one of the main goals of prediction is to determine a function of a hidden variable that can be reliably inferred from the output of a system.  

Privacy and security applications are concerned with the inverse problem: guaranteeing that a certain set of functions of a hidden variable \textit{cannot} be reliably estimated given the output of a system. Examples of such functions are the identity of an individual whose information is contained in a supposedly anonymous dataset  \cite{Sweeney-2002}, sensitive information of a user who joined a database  \cite{Dwork-McSherry-2006,dwork_differential_2006}, the political preference of a set of users who disclosed their movie ratings \cite{salamatian2013hide,salamatian2014managing,bhamidipati_priview:_2015}, among others. On the one hand, estimation methods attempt to extract as much information as possible from data. On the other hand, privacy-assuring systems 
seek to minimize the information about a secret variable that can be reliably estimated from 
disclosed data. The relationship between privacy and estimation is similar to the one noted by Shannon between cryptography and communication \cite{shannon_communication_1949}: they are connected fields, but with different goals.
As illustrated in Fig. \ref{fig:PrivEst}, estimation and privacy are concerned with the same fundamental problem, and can be simultaneously studied through an information-theoretic lens.

In this paper, we  discuss information-theoretic tools to address challenges in privacy, security and estimation. By studying fundamental models that are common to these fields, we derive information-theoretic metrics and associated results that simultaneously (i) delineate the fundamental limits of estimation and (ii) characterize the security properties of privacy-assuring systems. 

We focus on the question that is central to privacy and estimation (illustrated in Fig. \ref{fig:PrivEst}): How well can a random variable $S$, that is correlated with a hidden variable $X$, be estimated given an observation of $Y$? The information-theoretic metrics presented here  seek to quantify properties of the random mapping from $X$ to $Y$ that can  be translated into bounds on the  error of estimating $S$ given an observation of $Y$. These bounds, which are often at the heart of information-theoretic converse proofs \cite{cover_elements_2006}, provide universal, algorithm-independent guarantees on what can (or cannot) be learned from $Y$. With a characterization of these bounds in hand, we study properties of random mappings that seek to achieve privacy in terms of how well an adversary can estimate a secret $S$ given the output of the mapping $Y$.

The results in this paper are situated at the intersection of estimation, privacy and security. We derive a set of general sharp bounds on how well certain classes of functions of a hidden variable can(not) be estimated from a noisy observation.  The bounds are expressed in terms of different information metrics of the joint distribution of the hidden and observed variables, and provide converse (negative) results: If an information metric is small, then not only the hidden variable cannot be reliably estimated, but also any non-trivial function of the hidden variable cannot be inferred with probability of error or mean-squared error smaller than a certain threshold. 

These results are applicable to both estimation and privacy. For estimation and statistical learning theory, they shed light on the fundamental limits of learning from noisy data, and can help guide the design of practical learning algorithms. In particular, the converse bounds can be used to derive minimax lower bounds (the same way Fano-style inequalities are used \cite{duchi2013distance}). Furthermore, as illustrated in this paper, the proposed bounds are  useful for creating security and privacy metrics, for  characterizing  the inherent trade-off between privacy and utility in statistical data disclosure problems and for studying the fundamental limits of perfect privacy. The tools used to derive the converse bounds are based on a set of statistics known as the Principal Inertia Components (PICs).

\begin{figure}[!tb]
  \begin{center}
    \includegraphics[scale=0.4]{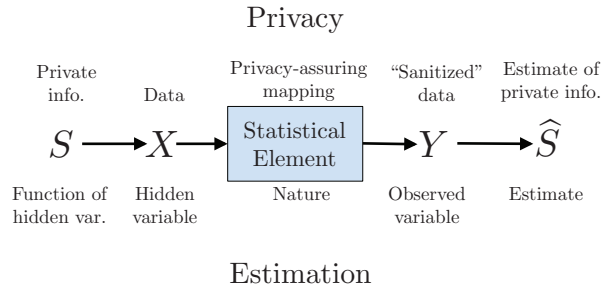}
  \end{center}
  \caption{Problem central to both estimation and privacy.}
  \label{fig:PrivEst}
\end{figure}

\subsection{Principal Inertia Components}
\label{sec:generalPICs}
The PICs provide a fine-grained decomposition of the dependence between two random variables. Well-studied statistical methods for estimating the PICs \cite{greenacre_theory_1984,breiman_estimating_1985} can lead to results  on the (im)possibility of estimating a large classes of functions by using bounds based on the PICs and standard statistical tests. We  show how PICs  can be used to characterize the information-theoretic limits of certain estimation problems. The PICs generalize other measures that are used in information theory, such as maximal correlation \cite{renyi1959measures} and $\chi^2$-dependence \cite{csiszar2004information}. The largest and smallest PIC play an important role in estimation and privacy (discussed in Sections \ref{sec:boundEP} and \ref{chap:PIC_Priv}). We also study properties of the sum of the $k$ largest principal inertia components. Below we list a few key properties of the PICs studied in this paper.

\begin{enumerate}
\item \textbf{Overview of the PICs:} We present an overview of the PICs and their different interpretations, summarized in Theorem \ref{thm:PIC_Charac}.  For two discrete random variables $X$ and $Y$, we denote the $k$ largest PICs by $\lambda_1(X;Y),\lambda_2(X;Y),\dots,\lambda_k(X;Y).$

\item \textbf{Sum of the PICs:} We propose a  measure of dependence termed  $k$-correlation which is defined as the sum of the $k$ largest PICs, i.e., 
 $ \calJ_k(X;Y) \defined \sum_{i=1}^k \lambda_i(X;Y)$.
This metric satisfies two key
properties: (i) convexity in $\pygx$ (Theorem \ref{thm:convex});
  (ii) Data Processing Inequality (Theorem \ref{lem:dataProc}). The latter is also satisfied by
      $\lambda_1(X;Y),\dots,\lambda_d(X;Y)$ individually, where  $d=\min\{|\calX|,|\calY|\}-1$. Both maximal correlation and the $\chi^2$-dependence between $X$ and $Y$ are special cases of $k$-correlation, with $\calJ_1(X;Y)=\rho_m(X;Y)^2$ and $\calJ_d(X;Y)=\chi^2(X;Y)$ (cf. notation in Section \ref{sec:notation}).

\item \textbf{Largest PIC}
 The largest PIC satisfies $\lambda_1(X;Y) = \rho_m(X;Y)^2$, where $\rho_m(X;Y)$ is the
 \textit{maximal correlation} between $X$ and $Y$, defined as \cite{renyi_measures_1959}
\begin{align}
\label{eq:maxcorrdef}
  \rho_m(X;Y) &\defined \max_{\substack{\EE{f(X)}=\EE{g(Y)}=0\\ \EE{f(X)^2}=\EE{g(Y)^2}=1}}\ExpVal{}{f(X)g(Y)}.
\end{align}
We show that both the probability of error and the minimum mean-squared error (MMSE) of estimating any function of a hidden variable $X$ given an observation $Y$ are closely related to the largest PIC.

By making use of the fact that the PICs satisfy the Data Processing Inequality (DPI), we are able to derive a family of bounds for the smallest average error of estimating $X$ having observed $Y$ $P_e(X|Y)$ (cf. \eqref{eq:PeDef} and notation in Section \ref{sec:notation}) in terms of the marginal distribution of $X$, $\px$, and $\lambda_1(X;Y),\dots,\lambda_d(X;Y)$, described in Theorem \ref{thm:Bound}. This result sheds light on the relationship of $P_e(X|Y)$ with   the PICs.

One immediate consequence of Theorem \ref{thm:Bound} is a useful scaling law for $P_e(X|Y)$ in terms of the largest PIC, the maximal correlation. Let $X=1$ be the most likely outcome for $X$. Corollary \ref{cor:coolBounds} proves that the advantage an adversary (who has access to $Y$) has of guessing $X$ over  guessing the most likely outcome  $X=1$ satisfies \vspace{-.1in}
\begin{equation*} 
  \Adv(X|Y) \defined \left|1-\px(1)-P_e(X|Y)\right|\leq O\left(\sqrt{\lambda_1(X;Y)}\right).
\end{equation*}

\item \textbf{Smallest PIC}
We show that the smallest PIC determines when perfect privacy, defined in Section \ref{chap:PIC_Priv}, can be achieved with non-trivial utility in the model depicted in Fig. \ref{fig:PrivEst}. More specifically, perfect privacy can be achieved with non-trivial utility if and only if the smallest PIC is 0 (Theorem \ref{thm:ntrivial}).

\end{enumerate}

\subsection{Organization of the Paper}

This paper is organized as follows. The rest of this section introduces notation and discusses related work. In Section \ref{chap:PICs}, we present the PICs and their multiple characterizations (Theorem \ref{thm:PIC_Charac}). We also introduce the definition of $k$-correlation, and demonstrate several properties of both $k$-correlation and, more broadly, the PICs, including convexity and the DPI. In Section \ref{chap:PIC_IT}, we apply the PICs to problems in information theory. In Section \ref{sec:boundEP}, we derive bounds on error probability and other estimation-theoretic results based on the PICs. Finally, in Section \ref{chap:PIC_Priv}, we demonstrate how the PICs play an important role in privacy and can be used for determining privacy-assuring mappings. We first summarize the main results obtained by applying the PICs to information theory, estimation theory and privacy.

\subsubsection*{Applications to Information Theory}

We present several distinct applications of the PICs to information theory. In Section \ref{sec:functions}, we demonstrate that the PICs correspond to the singular values of certain channel transformation matrices, and there effect on input distributions to the channel bear an interpretation similar to that of filter coefficients in a linear filter \cite{oppenheim_discrete-time_2009}. 
This is illustrated through an example in binary additive noise channels, where we argue that the binary symmetric channel is akin to a low-pass filter. We show how the PICs, and particularly the largest PIC, can be used to derive  bounds on information metrics between one-bit functions of a hidden variable $X$ and a correlated observation $Y$. We apply these results to the ``one-bit function conjecture'' \cite{kumar_which_2013}
We  do not solve this conjecture here. Nevertheless, we present further evidence for its validity, and introduce another conjecture based on our results.

The new conjecture (cf. Conjecture \ref{conj}) generalizes the ``one-bit function conjecture''. It states that, given a symmetric distribution $\pxy$, if we generate a new distribution $q_{X,Y}$ by making all the PICs of $\pxy$ equal to the largest one, then the new distribution is more informative about bits of $X$. By more informative, we mean that, for any 1-bit function $b$, $I(b(X);Y)$ is larger under $q_{X,Y}$ than  under $\pxy$. Indeed, from an estimation-theoretic perspective, increasing the PICs imply that any function of $X$ can be estimated with smaller MMSE when considering $q_{X,Y}$ than $\pxy$. Furthermore, in this case, we show that $q_{X,Y}$ is a $q$-ary symmetric channel. This conjecture, if proven, would imply as a corollary the original one-bit function conjecture.

We do show that our results on the PICs can be used to resolve the one-bit function conjecture in a specific setting in Section \ref{sec:estimators}. Instead of considering the mutual information between $b(X)$ and $Y$, we study the mutual information between $b(X)$ and a one-bit estimator $\hat{b}(Y)$. We show in Theorem \ref{thm:Estimators} that, when $\hat{b}(Y)$ is an unbiased estimator, the information that $\hat{b}(Y)$ carries about $b(X)$ can be upper-bounded for a range of dependence metrics (e.g. mutual information). This result also leads to bounds on estimation error probability. 

\subsubsection*{Applications to Estimation Theory}
In Section \ref{sec:boundEP}, we derive converse bounds on estimation error based on the PICs. In particular, we provide lower bounds on (i) the probability of correctly guessing  a hidden variable $X$ given an observation $Y$ and (ii) on the MMSE of estimating $X$ given $Y$. These results are stated in terms of the PICs between $X$ and $Y$, and provide  algorithm-independent bounds on estimation. We also extend these bounds to the functional setting, and show that the advantage over a random guess of correctly estimating a function of $X$ given an observation of $Y$ is upper-bounded by the largest PIC between $X$ and $Y$.
More specifically, we propose a family of lower bounds for the   error probability of estimating $X$ given $Y$ based on the PICs of $\pxy$ and the marginal distribution of $X$ in Theorems \ref{thm:Bound} and \ref{thm:AdvPeM}. We also extend these bounds for the probability of correctly estimating a function of the hidden variable $X$ given an observation of $Y$. 

These results are based on a more general framework for deriving bounds on error probability, discussed in Section \ref{sec:convexbounds}. 
At the heart of this framework are rate-distortion (test-channel) formulations that allow  bounds  on information metrics to be translated into bounds on estimation. These formulations, in turn, are based on convex programs that minimize the average estimation error over all possible distributions that satisfy a bound on a given information metric. The solution of such convex programs are called the error-rate functions. We study extremal properties of error-rate function and, by revisiting a result by Ahlswede \cite{ahlswede_extremal_1990}, we show how to extend the error-rate function to quantify not only the smallest average error of estimating a hidden variable, but also of estimating any function of a hidden variable. 

\subsubsection*{Applications to Privacy}

When referring to privacy in this paper, we consider the setting studied by Calmon and Fawaz in \cite{du2012privacy}. Using  Fig. \ref{fig:PrivEst} as  reference, we study the problem of disclosing data  $X$ to a third-party in order to derive some utility based on $X$. At the same time, some information correlated with $X$, denoted by $S$, is supposed to remain private. The engineering goal is to create a random mapping, called the privacy-assuring mapping, that transforms $X$ into a new data $Y$ that achieves a certain target utility, while minimizing the information revealed about $S$. For example, $X$ can represent  movie ratings that a user intends to disclose to a third-party in order to receive movie recommendations \cite{salamatian2013hide,salamatian2014managing,zhang_priview,bhamidipati_priview:_2015}. At the same time, the user may want to keep  her political preference $S$ secret. We allow the user to distort movie ratings in her data $X$ in order to generate a new data $Y$. The goal would then be  to find privacy-assuring mappings that minimize the number of distorted entries in $Y$ given a privacy constraint (e.g. the third-party cannot guess $S$ with significant advantage over a random guess). In general, $X$ is not restricted to be the data of an individual user, and can also represent multidimensional data derived from different sources. 

We present  necessary and sufficient conditions for achieving perfect privacy while disclosing a non-trivial amount of useful information when both $S$ and $X$ have finite support $\calS$ and $\calX$, respectively. We prove that the smallest PIC of $p_{S,X}$ plays a central role for achieving perfect privacy (i.e. $I(S;Y)=0$): If $|\calX|\leq |\calS|$, then perfect privacy is achievable with $I(X;Y)>0$ if and only if the smallest PIC of $p_{S,X}$ is 0. Since $I(S;Y)=0$  if and only if $S\independent Y$, this fundamental result holds for any privacy metric where statistical independence implies perfect privacy. We also provide an explicit lower bound for the amount of useful information that can be released  while guaranteeing perfect privacy, and demonstrate how to construct $p_{Y|X}$ in order to achieve  this bound.

In addition, we derive general bounds for the minimum amount of disclosed private information $I(S;Y)$ given that, on average, at least $t$ bits of useful information are revealed, i.e. $I(X;Y)\geq t$. These bounds are sharp, and delimit the achievable privacy-utility region for the considered setting. Adopting an analysis related to the information bottleneck \cite{tishby2000information} and for characterizing linear contraction coefficients in strong DPIs  in \cite{ahlswede_spreading_1976,anantharam_maximal_2013}, we determine the smallest achievable ratio between disclosed private and useful information, i.e.  $\inf_{p_{Y|X}}I(S;Y)/I(X;Y)$. We prove that this value is upper-bounded by the smallest PIC, and is zero if and only if the smallest PIC is zero. In this case, we present an explicit construction of a privacy-assuring mapping that discloses a non-trivial amount of useful information while guaranteeing perfect privacy. We also show that when the data is composed by multiple i.i.d. samples $(S^n,X^n)$, the smallest PIC decreases exponentially in $n$. Consequently, as the number of samples $n$ increases, we can achieve a more favorable trade-off between disclosing useful and private information. Finally, we motivate potential future applications of the PICs as a design driver for privacy assuring mappings in our final remarks in Section \ref{sec:finalremarks}.

\subsection{Notation}
\label{sec:notation}



Capital letters (e.g. $X$ and $Y$) are used to denote random variables, and calligraphic letters (e.g. $\calX$ and $\calY$) denote sets. The  exceptions are (i) $\calI$, which will be used in Section \ref{sec:boundEP} to denote a non-specified  measure of dependence, and (ii) $T$, which will denote the conditional expectation operator (defined below). The support set of  random variables $X$ and $Y$ are denoted by $\calX$ and $\calY$, respectively. If $X$ and $Y$ have finite support sets $ |\calX|< \infty$ and $|\calY|<\infty $, then we denote the joint probability mass function (pmf) of $X$ and $Y$ as $\pxy$, the conditional pmf of $Y$ given $X$ as $\pygx$, and the marginal distributions of $X$ and $Y$ as $\px$ and $\py$, respectively. We denote the fact that  $X$ is distributed according to $\px$ by $X\sim\px$.  When $p_{X,Y,Z}(x,y,z) =\px(x)\pygx(y|x)p_{Z|Y}(z|y)$ (i.e. $X,Y,Z$ form a Markov chain), we write $X\rightarrow Y \rightarrow Z$. We denote independence of two random variables $X$ and $Y$ by $X\independent Y$.

For positive integers $j,k,n$, $j\leq k$, we define $[n]\defined  \{1,\dots,n\}$ and $[j,k]\defined \{j,j+1,\dots,k\}$. For any $x \in \mathbb{R}$, $[x]^+$ is defined as $x$ if $x \ge 0$ and $0$ otherwise. Matrices are denoted in bold capital letters (e.g. $\bX$) and vectors in bold lower-case letters (e.g. $\bx$). The $(i,j)$-th entry of a matrix $\bX$ is given by $[\bX]_{i,j}$. Furthermore, for $\bx\in \Reals^n$, we let $\bx=(x_1,\dots,x_n)$. We denote by $\ones$ the vector with all entries equal to 1, and the dimension of $\ones$ will be clear from the context. The singular values of a  matrix $\bX \in \Reals^{m\times n}$ are denoted by $\sigma_1(\bX),\dots,\sigma_m(\bX)$. For a matrix $\bX$, we denote its $k$-th Ky Fan norm \cite[Eq. (7.4.8.1)]{horn_matrix_2012} by $\KFnorm{\bX}{k} \defined \sum_{i=1}^k \sigma_i(\bX)$.  


For a random variable $X$ with discrete support  and $X\sim \px$, the entropy of $X$ is given by \[H(X)\defined-\EE{\log \left( \px(X)\right)}.\] If $Y$ has a discrete support set and $X,Y\sim \pxy$,  the mutual information between $X$ and $Y$ is \[I(X;Y)\defined \EE{\log\left(\frac{\pxy(X,Y)}{\px(X)\py(Y)} \right)} .\] The basis of the logarithm will be clear from the context. The $\chi^2$-information between $X$ and $Y$ is \[\chi^2(X;Y)\defined \EE{\left(\frac{\pxy(X,Y)}{\px(X)\py(Y)} \right)}-1 .\] 
We denote the binary entropy function $h_b:[0,1]\to\Reals$ as 
\begin{equation}
\label{eq:defhb}
h_b(x)\defined-x\log  x-(1-x)\log (1-x),
\end{equation}
where, as usual, $0\log 0 \defined 0$. 

Let $X$ and $Y$ be discrete random variables with  finite support sets $\calX = [m]$ and $\calY = [n]$, respectively. Then we define the joint distribution matrix  $\bP$  as an $m\times n$ matrix with $[\bP]_{i,j}\defined p_{X,Y}(i,j)$. We denote by $\Px$ (respectively, $\Py$) the vector with $i$-th entry equal to $p_X(i)$ (resp. $p_Y(i)$). $\mathbf{D}_X=\diag{\Px}$ and $\mathbf{D}_Y=\diag{\Py}$ are matrices with diagonal entries equal to $\Px$ and $\Py$, respectively, and all other entries equal to 0. The matrix $\Pygx\in \Reals^{m\times n}$ is defined as $[\Pygx]_{i,j}\defined p_{Y|X}(j|i)$. Note that  $\bP=\bD_X\Pygx$.

For any real-valued random variable $X$, we denote the $L_p$-norm of $X$ as \[\|X\|_p\defined \left(\EE{|X|^p}\right)^{1/p}. \]
The set of all functions that when composed with a random variable $X$ with distribution $\px$ result in an  $L_2$-norm smaller than 1 is given by
\begin{equation}
    \calL_2(\px)\defined \left\{f:\calX\to \Reals\mid  \|f(X)\|_2\leq1
  \right\}.
\end{equation}

The operators $T_X:\calL_2(\py)\to\calL_2(\px)$ and $T_Y:\calL_2(\px)\to\calL_2(\py)$  denote conditional expectation, where  
\begin{equation}
(T_X g)(x)= \EE{g(Y)|X=x}\mbox{ and }(T_Y f)(y)=\EE{f(X)|Y=y}, \label{eq:Tdefn}
\end{equation}
 respectively. 
 Observe that $T_X$ and $T_Y$ are adjoint operators.

For $X$ and $Y$ with discrete support sets, we denote  by $P_e(X|Y)$ the smallest average probability of error of estimating $X$ given an observation of $Y$, defined as
 \begin{equation}
 \label{eq:PeDef}
        P_e(X|Y) = \min_{X\to Y\to \Xh} \Pr\{X\neq \Xh\},
 \end{equation}
where the minimum is taken over all distributions $p_{\Xh|Y}$ such that $X\to Y\to \Xh$. The advantage of correctly estimating $X$ given an observation of $Y$ over a random guess is defined as:
 \begin{equation}
        \Adv(X|Y) = 1-P_e(X|Y)-\max_{x\in \calX} p_X(x).
        \label{eq:advGuess}
 \end{equation}

The MMSE of estimating $X$ from an observation of $Y$ is given by 
\begin{equation*}
  \mmse(X|Y)\defined\min_{X\rightarrow Y\rightarrow \hat{X}}
  \EE{(X-\hat{X})^2},  
\end{equation*}
where the minimum is  taken over all distributions $p_{\Xh|Y}$ such that $X\to Y\to \Xh$. Note that, from Jensen's inequality, it is sufficient to consider $\Xh$ a deterministic mapping of $Y$. For any $X\to Y\to g(Y)$ with $\|g(Y)\|_2=\alpha$ and $\|X\|_2=\sigma$
\begin{align*}
  \EE{(X-g(Y))^2 } 
  &\geq \sigma^2 +\alpha^2 -2\alpha \|\EE{X|Y}\|_2,
\end{align*}    
with equality if and only if $g(Y)$ is proportional to $ \EE{X|Y}$. Minimizing the right-hand side over all $\alpha$, we find that the MMSE estimator of $X$ from $Y$ is
$g(y)=\EE{X|Y=y}$, and
\begin{equation}
  \label{eq:mmse}  
  \mmse(X|Y) = \|X\|_2^2 - \|\EE{X|Y}\|_2^2.
\end{equation}

For a given joint distribution $\pxy$ and corresponding joint distribution matrix $\bP$, the set of all vectors contained in the unit cube in $\Reals^n$ that satisfy $\normX{\bP\bx}{1}=a $ is given by 
\begin{equation} 
\label{eq:Cdef} 
\calC^n(a,\bP)\defined \{\xb\in \Reals^n|0\leq x_i \leq 1, \normX{\bP\bx}{1}=a\}. 
\end{equation} 
%
We represent the set of all $m\times n$ probability distribution matrices  by $\calP_{m,n}$.

For $x^n\in \{-1,1\}^n$ and $\calS\subseteq [n]$, 
\begin{equation}
\label{eq:defchi}
\chi_\calS(x^n)\defined\prod_{i\in \calS}x_i
\end{equation}
 (we consider $\chi_\emptyset(x) = 1$). For $y^n\in \{-1,1\}^n$, $a^n=x^n\oplus y^n$ is the vector resulting from the entrywise product of $x^n$ and $y^n$, i.e. $a_i = x_iy_i$, $i\in [n]$.

Given two probability distributions $p_{X}$ and $q_X$ and $f(t)$ a smooth convex function defined for $t>0$ with $f(1)=0$, the $f$-divergence is defined as \cite{_information_2004}
    \begin{equation}
        D_f(p_X||q_X) \defined \sum_x q_{X}(x) f \left( \frac{p_X(x)}{q_X(x)}
        \right).
    \end{equation}
The $f$-information is given by
    \begin{equation}
    \label{eq:defn_finf}
        I_f(X;Y)\defined D_f(p_{X,Y}||p_Xp_Y).
    \end{equation}
When $f(x)=x\log(x)$, then $I_f(X;Y)=I(X;Y)$. A study of information metrics related to $f$-information  was given in \cite{polyanskiy2010arimoto} in the context of channel coding converses.

\subsection{Related Work}
\label{sec:PIC_Related}
\vspace{-.05in}

The joint distribution matrix $\bP$ can be viewed as a contingency table and decomposed  using standard techniques from correspondence analysis \cite{greenacre_correspondence_2007,greenacre_theory_1984}. For an overview of correspondence analysis, we refer the reader to \cite{greenacre_geometric_1987}. The term ``principal inertia components'', used here, is borrowed from the correspondence analysis literature \cite{greenacre_theory_1984}. However, the study of the PICs of the joint distribution of two random variables or, equivalently, the spectrum of the conditional expectation operator, predates correspondence analysis, and goes back to the work of Hirschfeld \cite{hirschfeld_connection_1935}, Gebelein \cite{gebelein_statistische_1941}, Sarmanov \cite{sarmanov1962maximum}  and R\'enyi \cite{renyi_measures_1959}, having also appeared in the work of Witsenhausen \cite{witsenhausen_sequences_1975} and Ahlswede and G\'acs \cite{ahlswede_spreading_1976}. The PICs are also related to strong DPIs and contraction coefficients, being recently investigated by Anantharam \etal \cite{anantharam_maximal_2013}, Polyanskiy \cite{polyanskiy_hypothesis_2012}, Raginsky \cite{raginsky_logarithmic_2013}, Calmon \etal \cite{calmon_bounds_2013}, Makur and Zheng \cite{makur_bounds_2015}, among others. Recently, Liu \etal \cite{liu_brascamp-lieb_2016} provided a unified perspective on several functional inequalities used in the study of strong DPIs and hypercontractivity. The PICs also play a role in Euclidean Information Theory \cite{huang_euclidean_2015}, since they related to $\chi^2$-divergence and, consequently, to local approximations of mutual information and related  measures.

The largest principal inertia component is equal to $\rho_m(X;Y)^2$, where $\rho_m(X;Y)$ is the \textit{maximal correlation} between $X$ and $Y$. Maximal correlation has been widely studied in the information theory and statistics literature (e.g \cite{sarmanov1962maximum,renyi_measures_1959}). Ahslwede and  G\'acs studied maximal correlation in the context of contraction coefficients in strong data processing inequalities \cite{ahlswede_spreading_1976}, and more recently Anantharam \etal presented in \cite{anantharam_maximal_2013} an overview of different characterizations of maximal correlation, as well as its application in information theory. Estimating the maximal correlation is also the goal of the Alternating Conditional Expectation (ACE) algorithm introduced by Breiman and Friedman \cite{breiman_estimating_1985},  further analyzed by Buja \cite{buja_remarks_1990}, and recently investigated in \cite{makur2015efficient}. 

The DPI for the PICs was shown by Kang and Ulukus in \cite[Theorem 2]{kang_new_2011} in a different setting than the one considered here. Kang and Ulukus made use of the decomposition of the joint distribution matrix  to derive outer bounds for the rate-distortion region achievable in certain distributed source and channel coding problems.

Lower bounds on the average estimation error can be found using Fano-style inequalities. Recently, Guntuboyina \etal (\cite{guntuboyina_lower_2011,guntuboyina_sharp_2013}) presented a family of sharp bounds for the minmax risk in estimation problems involving general $f$-divergences. These bounds generalize Fano's inequality and, under certain assumptions,  can be extended in order to lower bound $P_e(X|Y)$.

Most information-theoretic approaches for estimating or communicating functions of a random variable are concerned with properties of specific functions  given i.i.d. samples of the hidden variable $X$, such as in the functional compression literature \cite{doshi_functional_2010,orlitsky_coding_2001}. These results are rate-based and asymptotic, and do not immediately extend to the case where the function $f(X)$ can be an arbitrary member of a class of functions, and only a single observation is available.

More recently,  Kumar and Courtade \cite{kumar_which_2013} investigated Boolean functions in an information-theoretic context. In particular, they analyzed which is the most informative (in terms of mutual information) 1-bit function for the case where $X$ is composed by $n$ i.i.d. Bernoulli(1/2) random variables, and $Y$ is the result of passing $X$ through a discrete memoryless binary symmetric channel. Even in this simple case, determining the most informative function seems to be non-trivial. Further investigations of this problem was done in \cite{kindler_remarks_2015,anantharam_hypercontractivity_2013,ordentlich_improved_2015,chandar_most_2014}. In particular, \cite{kindler_remarks_2015} studies a related problem in a continuous setting by considering that $X$ and $Y$ are Gaussian random vectors. Recently, Samorodnitsky \cite{samorodnitsky_most_2015} presented a proof of the conjecture in the high noise regime.

Information-theoretic formulations for privacy  have appeared in \cite{reed_information_1973,rebollo-monedero_t-closeness-like_2010,sankar_utility-privacy_2013,tandon_discriminatory_2013,evfimievski_limiting_2003}. For an overview, we refer the reader to \cite{du2012privacy,sankar_utility-privacy_2013} and the references therein. The privacy against statistical inference framework considered here was further studied in \cite{salamatian2013hide, makhdoumi2013privacy, salamatian2014managing}. The results presented in this paper are closely connected to the study of hypercontractivity coefficients and strong data processing results, such as in \cite{ahlswede_spreading_1976,anantharam_maximal_2013,polyanskiy_dissipation_2014,polyanskiy_hypothesis_2012,raginsky_logarithmic_2013}. 
PIC-based analysis were used in the context of security in \cite{li_maximal_2014,calmon2014allerton}. Extremal properties of privacy were also investigated in \cite{chakraborty_protecting_2013,asoodeh_notes_2014}, and in particular \cite{asoodeh_information_2015} builds upon some of the results introduced here. For more details on designing privacy-assuring mappings and applications with real-world data, we refer the reader to \cite{du2012privacy,salamatian2013hide,salamatian2014managing,Makhdoumi2014funnel,zhang_priview,bhamidipati_priview:_2015}.

We note that the privacy against statistical inference setting is related to differential privacy \cite{dwork_differential_2006,Dwork-McSherry-2006}. In the classic differential privacy setting, the output of a statistical query over a database is masked against small perturbations of the data contained in the database.  Assuming this centralized statistical database setting, the private variable $S$ can represent an individual user's entry to the database, and the variable  $X$ the output of a query over the database. Unlike in differential privacy, here we consider an additional distortion constraint, which can be chosen according to the application at hand. In the privacy funnel setting \cite{Makhdoumi2014funnel}, the distortion constraint is given in terms of the mutual information between $X$ and the perturbed query output $Y$. Connections between differential privacy and the privacy setting depicted in Fig. \ref{fig:PrivEst} as well as connections between differential privacy and PICs are studied in \cite{du2012privacy, ISITForgot}.

%% file: PIC.tex




We introduce in this section the Principal Inertia Components (PICs) of the joint distribution of two random variables $X$ and $Y$. The PICs provide a fine-grained decomposition of the statistical dependence between $X$ and $Y$, and are dependence measures that lie in the intersection of information and estimation theory. The PICs possess several desirable information-theoretic properties (e.g. satisfy the DPI, convexity, tensorization, etc.), and describe which functions of $X$ can or cannot be reliably inferred (in terms of MMSE) given an observation of $Y$. The latter interpretation is discussed in more detail in Section \ref{sec:boundEP}. 

\subsection{A Geometric Interpretation of the PICs}
\label{sec:PIC_Geo}
We give an intuitive geometric interpretation of the PICs before presenting their formal definition in the next section. Let  $X$ and $Y$ be related through a conditional distribution (channel), denoted by $p_{Y|X}$. For each $y\in \calY$, $p_{X|Y}(\cdot|y)$ will be a vector on the $|\calX|$-dimensional simplex, and the position of these vectors on the simplex will determine the nature of the relationship between $X$ and $Y$ (Fig. \ref{fig:triangle}).  If $p_{X|Y}$ is fixed, what can be learned about $X$ given an observation of $Y$, or the degree of accuracy of what can be inferred about $X$ \textit{a posteriori}, will  then depend on the marginal distribution $p_Y$. The value $p_Y(y)$, in turn, ponderates the corresponding vector $p_{X|Y}(\cdot|y)$ akin to a mass. As a simple example, if $|\calX|=|\calY|$ and the vectors $p_{X|Y}(\cdot|y)$ are located on distinct corners of the simplex, then $X$ can be perfectly learned from $Y$. As another example, assume that the vectors $p_{X|Y}(\cdot|y)$ can be grouped into two clusters located near opposite corners of the simplex. If the sum of the masses induced by $p_Y$ for each cluster is approximately $1/2$, then one may expect to reliably infer on the order of 1 unbiased bit of $X$ from an observation of $Y$.

\begin{figure}[!tb]
  \begin{center}
    \includegraphics[scale=0.4]{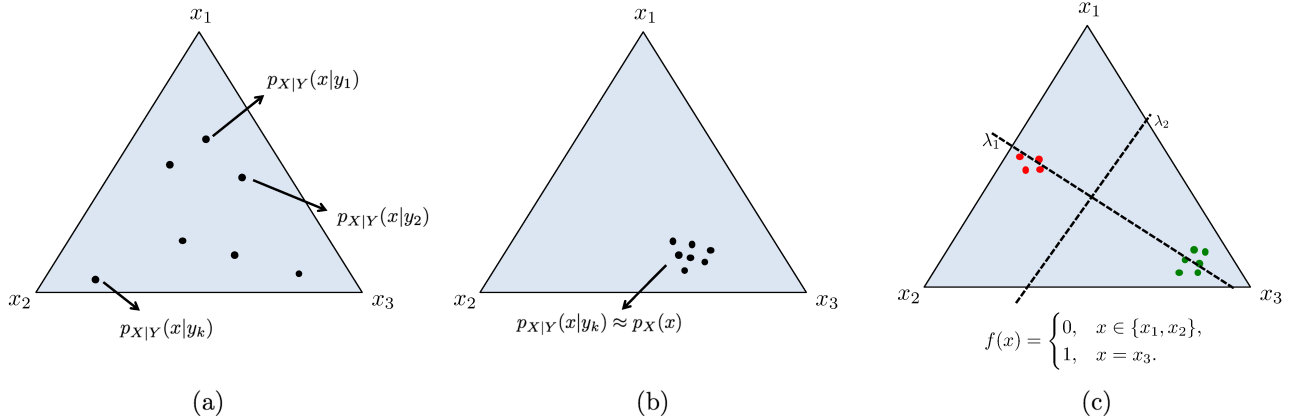}
  \end{center}
  \caption{Geometric interpretation of the PICs for $\calX = \{x_1,x_2,x_3\}$ and $X$ uniformly distributed. In (a), each point on the simplex corresponds to a posterior distribution $p_{X|Y}(\cdot|y)$ induced on $\calX$ by an observation of $Y=y_k$. If all the posterior distribution points are close together (b), then $X$ and $Y$ are approximately independent. If these points are far apart (c), then there may exist a function of $X$ that can be approximately reliably estimated given an observation of $Y$ (in this case, a binary function). The PICs can be intuitively understood as a measure of inertia of the posterior distribution vectors on the simplex.}
  \label{fig:triangle}
\end{figure}

The above discussion naturally leads to considering the use of techniques borrowed from classical mechanics.  For a given inertial frame of reference, the mechanical properties of a collection of distributed point masses can be characterized by the moments of inertia of the system. The moments of inertia measure how the weight of the point masses is distributed around the center of mass. An analogous metric exists for the distribution of the vectors $p_{X|Y}$ and masses $p_Y$ in the simplex, and it is the subject of study of a branch of applied statistics called \textit{correspondence analysis} (\cite{greenacre_theory_1984,greenacre_geometric_1987}). In correspondence analysis, the joint distribution  $p_{X,Y}$  is decomposed in terms of the PICs, which, in some sense, are analogous to the moments of inertia of a collection of point masses. 
For more related literature, we refer the reader back to Section \ref{sec:PIC_Related}.

\subsection{Definition and Characterizations of the PICs}
\label{sec:charac}
We start with the definition of principal inertia components. In this paper we focus on the discrete case, since two of our main goals are  (i) derive lower bounds on average estimation error probability and (ii) apply these results to privacy, where private data is often categorical. In addition, tools from correspondence analysis \cite{greenacre_correspondence_2007} can be used for estimating the PICs in the discrete setting. Nevertheless, the definition below is not limited to discrete random variables, and can be directly extended to general probability measures under compactness of the operator $T_XT_Y$ (cf. \cite[Section 3]{witsenhausen_sequences_1975}).

\begin{defn}
\label{def:PIC_4}
Let $X$ and $Y$ be random variables with support sets $\calX$ and $\calY$, respectively, and joint distribution $p_{X,Y}$. In addition, let $f_0:\calX\to \Reals$ and $g_0:\calY\to \Reals$ be the constant functions $f_0(x)=1$ and $g_0(y)=1$. For $k\in\mathbb{Z}_+$, we (recursively) define
\begin{align} 
\lambda_k(X;Y) = \max\Big\{\EE{f(X)g(Y)}^2\Big|& f\in \calL_2(p_X), g\in
      \calL_2(p_Y),\EE{f(X)f_j(X)}=0,\nonumber \\
      &\EE{g(Y)g_j(Y)}=0, j\in\{0,\dots,k-1\} \Big\},    
\end{align}
where
\begin{align}
             (f_k,g_k) \defined \argmax\Big\{\EE{f(X)g(Y)}^2\Big|& f\in \calL_2(p_X), g\in
      \calL_2(p_Y),\EE{f(X)f_j(X)}=0, \nonumber \\
      &\EE{g(Y)g_j(Y)}=0, j\in\{0,\dots,k-1\} \Big\}.  \label{eq:fkmaxcorr} 
\end{align}
The values $\lambda_k(X;Y)$ are
  called the  \textit{principal inertia components} (PICs) of $\pxy$. The functions $f_k$ and $g_k$ are called the \textit{principal functions} of $X$ and $Y$.
\end{defn}

Observe that the PICs satisfy $\lambda_k(X;Y)\leq 1$, since $f_k\in \calL_2(\px)$ $g_k\in \calL_2(p_Y)$ and 
\begin{equation*}
  \EE{f(X)g(Y)}\leq \|f(X)\|_2\|g(Y)\|_2\leq1.
\end{equation*}
Thus, from Definition \ref{def:PIC_4}, $\lambda_{k+1}(X;Y)\leq \lambda_{k}(X;Y)\leq 1$. When both random variables $X$ and $Y$ have a finite support set, we have the following definition.
\begin{defn}
\label{def:Q}
  For $\calX=[m]$ and $\calY=[n]$, let  $\bP\in \Reals^{m\times n}$ be a matrix with
  entries $[\bP]_{i,j}=p_{X,Y}(i,j)$, and $\bD_X\in \Reals^{m\times m}$ and
  $\bD_Y\in \Reals^{n\times n}$ be diagonal matrices with diagonal entries
  $[\bD_X]_{i,i}=p_X(i)$ and $[\bD_Y]_{j,j}=p_Y(j)$, respectively, where $i\in
  [m]$ and $j\in [n]$. We define
  \begin{equation}
    \label{eq:Qdefn}
    \bQ \defined \bD_X^{-1/2}\bP\bD_Y^{-1/2}.
  \end{equation}
  We denote the singular value decomposition of $\bQ$ by $\bQ=\bU\bSigma\bV^T$.
\end{defn}

The next theorem provides four equivalent characterizations of the PICs.

\begin{thm}
\label{thm:PIC_Charac}
The following characterizations of the PICs are equivalent:
\begin{enumerate}[(1)]
\item The characterization given in Definition \ref{def:PIC_4} where, for $f_k$ and $g_k$ given in \eqref{eq:fkmaxcorr}, $g_k(Y)=\frac{\EE{f_k(X)|Y}}{\|\EE{f_k(X)|Y}\|_2}$ and $f_k(X)=\frac{\EE{g_k(Y)|X}}{\|\EE{g_k(Y)|X}\|_2}$. 

\item \cite[Section 3]{witsenhausen_sequences_1975} Consider the conditional expectation operator $T_Y:\calL_2(p_X)\to
    \calL_2(p_Y)$, defined in \eqref{eq:Tdefn}.
    Then $$\left(1,\sqrt{\lambda_1(X;Y)},\sqrt{\lambda_2(X;Y)},\dots\right)$$ are the singular values of $T_Y$.

\item For any $k\in\mathbb{Z}_+$, 
    \begin{align}
    \label{eq:mmsePIC}
      1- \lambda_k(X;Y) = \min \Big\{\mmse(f(X)|Y)\Big|& f\in
        \calL_2(p_X),\|f(X)\|_2=1,\EE{f(X)h_j(X)}=0, j\in\{0,\dots,k-1\} \Big\},    
      \end{align}
      where 
     \begin{align}
       \label{eq:fkmmse}
      h_k \defined \argmin\Big\{\mmse(f(X)|Y)\Big|& f\in
        \calL_2(p_X),\|f(X)\|_2=1,\EE{f(X)h_j(X)}=0, j\in\{0,\dots,k-1\} \Big\}.    
      \end{align}
      If $\lambda_k(X;Y)$ is unique, then $h_k=f_k$ given in \eqref{eq:fkmaxcorr}.

\end{enumerate}
Finally, if both $\calX$ and $\calY$ are defined over finite supports, the following characterization is also equivalent.
\begin{enumerate}[(1)]
\setcounter{enumi}{3}
\item $\sqrt{\lambda_k(X;Y)}$ is the $(k+1)$-st largest singular value of $\bQ$. The principal functions $f_k$ and $g_k$ in \eqref{eq:fkmaxcorr} correspond to the columns of the matrices $\Dx^{-1/2}\bU$ and $\Dy^{-1/2}\bV$, respectively, where $\bQ = \bU\bSigma\bV$. 
\end{enumerate}
\end{thm}
\begin{proof}
We will prove that $(1)\iff(2)$,  $(1)\iff (3)$, finally and $(1)\iff (4)$.
\begin{itemize}

\item $(1)\iff (2)$. First observe that for $f\in \calL_2(\px)$ and $g\in \calL_2(\py)$
\begin{align*}
\EE{f(X)g(Y)}= \EE{g(Y)\EE{f(X)|Y}} \leq \|g(Y)\|_2\|\EE{f(X)|Y}\|_2\leq \|\EE{f(X)|Y}\|_2,
\end{align*}
where the first inequality follows from the Cauchy-Schwarz inequality, with equality if and only if $g(Y)=\frac{\EE{f(X)|Y}}{\|\EE{f(X)|Y}\|_2}$.
The equivalence  then follows by noting that 
    \begin{align}
      \sqrt{\lambda_1(X;Y)}&= \max_{\substack{\EE{f(X)}=\EE{g(X)}=0\\ \|f(X)\|_2=\|g(Y)\|_2=1}} \EE{f(X)g(Y)} \nonumber\\
      &= \max_{\substack{\EE{f(X)}=\EE{g(X)}=0\\ \|f(X)\|_2=\|g(Y)\|_2=1}} \EE{\EE{g(Y)f(X)|Y} } \nonumber\\
      &= \max_{\substack{\EE{f(X)}=0\\ \|f(X)\|_2=1}} \|\EE{f(X)|Y} \|_2, \label{eq:varPIC}
    \end{align}
    where the last equality follows by setting $g(Y)=\frac{\EE{f(X)|Y}}{\|\EE{(X)|Y}\|_2}$. Inverting the roles of $f$ and $g$, we find $f(X)=\frac{\EE{g(Y)|X}}{\|\EE{(Y)|X}\|_2}$. Since this last expression is the second largest singular value of the conditional expectation operator $T_Y$ (the largest being 1), the result follows for $\lambda_1(X;Y)$. The equivalent result for the other PICs follows by adding orthogonality constraints and the min-max properties of singular values (cf. Rayleigh-Ritz Theorem \cite[Theorem 4.2.2]{horn_matrix_2012}). 

\item $(1)\iff (3)$. The result follows from $\lambda_k(X;Y)=\|\EE{f_k(X)|Y}\|_2^2$ in \eqref{eq:varPIC} and by noting that the MMSE can be written as \eqref{eq:mmse}. Consequently, maximizing $\|\EE{f(X)|Y)}\|$ is equivalent to minimizing the MMSE in \eqref{eq:mmsePIC}.

\item $(1)\iff(4)$. 
  Let $f\in \calL_2(p_X)$ and $g\in \calL_2(p_Y)$.  Define the column-vectors $\fb\defined (f(1),\dots,f(m))^T$ and $\gb\defined (g(1),\dots,g(n))^T$. Then \[\EE{f(X)g(Y)}=\fb^T\bP\gb\] and \[\fb^T\bD_X\fb=\gb^T\bD_Y\gb=1.\]  For $\bQ=\bU\bSigma\bV^T$ given in Definition \ref{def:Q}, put $\bu\defined \bU^T\Dx^{1/2}\fb$ and $\bv\defined \bV\Dy^{1/2}\gb$. Then $\|\bu\|_2=\|\bv\|_2=1$, and \[\EE{f(X)g(Y)}=\bu^T\bSigma\bv.\] The result then follows directly from the variational characterization of singular values \cite[Theorem 7.3.8]{horn_matrix_2012}.
      
  Assuming unique PICs, note that the column-vectors $(\fb_0,\fb_1,\dots,\fb_d)$ corresponding to the functions $(f_0,f_1,\dots,f_d)$ are the first $d+1$  columns of $\Dx^{-1/2}\bU$, and the column-vectors $(\bg_0,\bg_1,\dots,\bg_d)$ corresponding to the functions $(g_0,g_1,\dots,g_d)$ are the first $d+1$ of  $\Dy^{-1/2}\bV$. In addition, let $\bz_k\in \Reals^n$ be the column vector with entries $\EE{f_k(X)|Y=j}$. Then 
  \begin{align*} \bz_k= \fb^T\bP\Dy^{-1} = \fb_k^T \Dx^{1/2}\bU\bSigma\bV^T\Dy^{-1/2}=\sqrt{\lambda_k(X;Y)}\bg_k, 
  \end{align*} 
  so $\lambda_k(X;Y)=\|\EE{f_k(X)|Y}\|_2^2$ and once again we find $g_k(Y)=\frac{\EE{f_k(X)|Y}}{\|\EE{f_k(X)|Y}\|_2}$.
\end{itemize}
\end{proof}

The previous theorem provides different operational characterization of the PICs.  Characterization (1), presented in Definition \ref{def:PIC_4} implies  that the principal functions of $X$ and $Y$ are the solution to the following problem: Consider two parties, namely Alice and Bob, where Alice has access to an observation of $X$ and Bob has access to an observation $Y$. Alice and Bob's goal is to produce zero-mean, unit variance functions $f(X)$ and $g(Y)$, respectively, that maximizes the correlation $\EE{f(X)g(Y)}$ without any additional information beyond their respective observations of $X$ and $Y$. The optimal choice of functions is  $f_1$ and $g_1$, given in the theorem. Moreover,
\begin{equation*}
  \lambda_1(X;Y)=\rho_m(X;Y)^2.
\end{equation*}
Characterization (3) above proves that the PICs are the solution to another related question: Given a noisy observation $Y$ of a hidden variable $X$, what is the unit-variance, zero-mean function of $X$ that can be estimated with the smallest mean-squared error? It follows directly from \eqref{eq:mmsePIC} that the function is $f_1(X)$, and the minimum MMSE is $1-\lambda_1(X;Y)$. Indeed, since they are orthonormal, the principal functions form a basis for the zero-mean functions in $\calL_2(\px)$ (we revisit this point in the Section \ref{sec:finalremarks}). Characterization (4) lends itself to the geometric interpretation discussed in Section \ref{sec:PIC_Geo}.

The next result states the well-known tensorization property the PICs between sequences of independent random variables (e.g. \cite{witsenhausen_sequences_1975,anantharam_maximal_2013,beigi_duality_2015}). We present a proof of the discrete case here for the sake of completeness.
\begin{lem}
\label{lem:tensor}
Let $(X_1,Y_1)\independent (X_2,Y_2)$, $d_1=\min\{|\calX_1|,|\calY_1|\}-1<\infty$ and $d_2=\min\{|\calX_2|,|\calY_2|\}-1<\infty$. Then the PICs of $p_{(X_1,X_2),(Y_1,Y_2)}$ are $\lambda_i(X_1,Y_1)\lambda_j(X_2,Y_2)$ for $(i,j)\in [0,d_1]\times[0,d_2]$, where $\lambda_0(X_1,Y_1)=\lambda_0(X_2,Y_2)=1$. Furthermore, denoting the principal functions $(X_1,Y_1)$ by $f_i$ and of $(X_2,Y_2)$ by $\tilde{f}_j$, then the principal functions of $ p_{(X_1,X_2),(Y_1,Y_2)} $ are of the form $(x_1,x_2)\mapsto f_i(x_1) \tilde{f}_j(x_2)$. In particular
\begin{equation*}
  \lambda_1((X_1,X_2);(Y_1,Y_2))= \max\{ \lambda_1(X_1;Y_1),\lambda_1(X_2;Y_2)\}.
\end{equation*}
\end{lem}
\begin{proof}
Let $[\bQ_1]_{i,j}=\frac{p_{X_1,Y_1}(i,j)}{\sqrt{p_{X_1}(i)p_{Y_1}(j)}}$ and $[\bQ_2]_{i,j}=\frac{p_{X_2,Y_2}(i,j)}{\sqrt{p_{X_2}(i)p_{Y_2}(j)}}$. Denoting by $\bQ$ the decomposition in Definition \ref{def:PIC_4} of $ p_{(X_1,X_2),(Y_1,Y_2)} $ then, from the independence assumption, $\bQ=\bQ_1\otimes \bQ_2$, where $\otimes$ is the Kronecker product. The result  follows directly from the fact that the singular values of the Kronecker product of two matrices are the Kronecker product of the singular values (and equivalently for the singular vectors) \cite[Theorem 4.2.15]{topics_horn}.
\end{proof}

\subsection{$k$-correlation}
\label{sec:measureMoments}

In this section we  introduce the $k$-correlation $\calJ_k(X;Y)$ between two random variables, which is equivalent to the sum of the $k$ largest PICs. We prove that  $k$-correlation is convex in $\pygx$ and satisfies the DPI. 

\begin{defn}
We define the \textit{$k$-correlation} between $X$ and $Y$ as 
  \begin{align}
    \KC_k(X;Y)&\defined \sum_{i=1}^k \lambda_i(X;Y).
\end{align}
For finite $\calX$ and $\calY$, the $k$-correlation is given by
  \begin{align}
    \KC_k(X;Y)& \defined \KFnorm{\bQ\bQ^T}{k}-1.
  \end{align}
\end{defn}
Note that 
\begin{equation*}
    \KC_1(X;Y)=\rho_m(X;Y)^2,
\end{equation*}
and for finite $\calX$ and $\calY$, $d=\min\{|\calX|,|\calY|\}-1$,
\begin{align*}
  \KC_d(X;Y) = \ExpVal{X,Y}{\frac{\pxy(X,Y)}{\px(X)\py(Y)}}-1=\chi^2(X;Y).
\end{align*}


We demonstrate next that $k$-correlation and, consequently, maximal correlation, is convex in  $\pygx$ for a fixed $\px$ and satisfies a form of the DPI, i.e. if $X\to Y \to Z$, then $\calJ_k(X;Y)\leq \calJ_k(X;Z)$. These results hold for both discrete and continuous random variables (under appropriate compactness assumptions),



\begin{thm}
\label{thm:convex}
    For a fixed $\px$, $\KC_k(X;Y)$ is convex in $\pygx$.
\end{thm}
\begin{proof}
First note that $\|\EE{f(X)|Y}\|_2^2$ is convex $\pxy$, since for any $U\to (X,Y)$
\begin{align*}
\EV{Y}{\left(\EV{X|Y}{f(X)|Y}\right)^2} &= \EV{Y}{\left(\EV{U|Y}{\EV{X|Y,U}{f(X)|Y,U}}\right)^2}\\
&\leq \EV{Y}{\EV{U|Y}{\left(\EV{X|Y,U}{f(X)|Y,U}\right)^2}} \\
&= \EV{U}{\EV{Y|U}{\left(\EV{X|Y,U}{f(X)|Y,U}\right)^2}},
\end{align*}
where the inequality follows from Jensen's inequality. Consequently, for any $\{f_1,\dots,f_k\}\subseteq \calL_2(\px)$, $\sum_{i=1}^k \|\EE{f_i(X)|Y}\|_2^2$ is convex in $\pxy$ and thus, for a fixed $\px$, convex in $\pygx$. From Theorem \ref{thm:PIC_Charac} and the Poincar\'e separation theorem \cite[Corollary 4.3.16]{horn_matrix_2012}
\begin{align*}
\sum_{i=1}^k\lambda_i(X;Y) = \max_{\substack{\{f_i\}_{i=1}^k\subseteq \calL_2(\px)\\f_i\perp f_j, i\neq j\\\EE{f_i}=0}} \sum_{i=1}^k \|\EE{f_i(X)|Y}\|_2^2.
\end{align*}
Since the pointwise supremum of convex functions is convex \cite[Sec 3.2.3]{boyd_convex_2004}, it follows that for fixed $\px$ $\calJ_k(X;Y)$ is convex in $\pygx$.
\end{proof}


The following lemma will be used to prove  that the PICs satisfy the DPI.

\begin{lem}[DPI for MMSE]
\label{lem:DPI_MMSE}
For $X\to Y \to Z$ and any $f\in \calL_2(p_X)$, $\EE{f(X)}=0$,
\begin{equation}
\label{eq:MMSE_DPI}
  \|\EE{f(X)|Z}\|_2^2\leq \lambda_1(Y;Z)\|\EE{f(X)|Y}\|_2^2.
\end{equation}
Consequently, $\mmse(f(X)|Y)\leq \mmse(f(X)|Z)$.
\end{lem}
\begin{proof}
The proof is in Appendix \ref{app:proofsPIC}.
\end{proof}

Lemma \ref{lem:DPI_MMSE} leads to the following theorem.

\begin{thm}[DPI for the PICs]
  \label{lem:dataProc}
Assume that $X\rightarrow Y \rightarrow Z$. Then $\lambda_k(X;Z)\leq \lambda_1(Y;Z)\lambda_k(X;Y)$ for all $k$.
\end{thm}

\begin{proof}
  A direct consequence of Theorem \ref{thm:PIC_Charac} is that for any two random variables $X, Y$
  \begin{equation*}
    \lambda_k(X;Y)= \min_{\{f_i\}_{i=1}^k\subseteq \calL_2(\px)} \max_{\substack{f\in\calL_2(\px)\\\EE{f(X)f_i(X)}=0}} \|\EE{f(X)|Y}\|_2^2,
  \end{equation*}
and equivalently for $ \lambda_k(X;Z)$. The result then follows directly from  \eqref{eq:MMSE_DPI}.
\end{proof}

The next corollary is a direct consequence of the previous theorem. 
\begin{cor}
    For $X \rightarrow Y \rightarrow Z$ forming a Markov chain, $\KC_k(X;Z)\leq \lambda_1(Y;Z)\KC_k(X;Y)$.
\end{cor}

\begin{remark}
  The data processing result in Theorem \ref{lem:dataProc} and the previous corollary was proved by Kang and Ulukus in \cite[Theorem 2]{kang_new_2011} and applied to problems in distributed source and channel coding, even though they do not make the explicit connection with maximal correlation and PICs. A weaker form of Theorem \ref{lem:dataProc} can be derived using a clustering result presented in \cite[Sec. 7.5.4]{greenacre_theory_1984} and originally due to Deniau \etal \cite{deniau_effet}. We use a different proof technique from the one in \cite[Sec. 7.5.4]{greenacre_theory_1984} and \cite[Theorem 2]{kang_new_2011} to show result stated in the theorem, and present the proof here for completeness. Finally, a related data processing result was stated in \cite{polyanskiy_hypothesis_2012}.
\end{remark}


In the next three sections of the paper, we demonstrate the fundamental role of PICs in problems in information theory, estimation theory, and privacy.


%% file: PIC_IT.tex

In this section, we present results that connect the PICs with other information-theoretic metrics. As seen in Section \ref{chap:PICs}, the distribution of the vectors $p_{Y|X}$ in the simplex or, equivalently, the PICs of the joint distribution of $X$ and $Y$, are inherently connected to how an observation of $Y$ is statistically related to $X$. In this section, we explore this connection within an information theoretic framework. We show that, under certain assumptions, the PICs play an important part in estimating a one-bit function of $X$, namely $b(X)$ where $b:\calX\rightarrow \{0,1\}$, given an observation of $Y$: they can be understood  as the singular values (or filter coefficients) in the linear transformation  of $p_{b(X)|X}$ into $p_{b(X)|Y}$ determined by the channel transition matrix. Alternatively, the PICs can bear an interpretation as the transform of the distribution of the noise in certain additive-noise channels, in particular when $X$ and $Y$ are binary strings. We also show that maximizing the PICs is equivalent to maximizing the first-order term of the Taylor series expansion of certain convex dependence measures between $b(X)$ and $Y$. We conjecture that, for symmetric distributions of $X$ and $Y$ and a given upper bound on the value of the largest PIC, $I(b(X);Y)$ is maximized when all the principal inertia components have the same value as the largest principal inertia component. For uniformly distributed $X$ and $Y$, this is equivalent to $Y$ being the result of passing $X$ through a $q$-ary  symmetric channel. This conjecture, if proven, would imply the conjecture made by Kumar and Courtade in \cite{kumar_which_2013}.

Finally, we study the Markov chain $B\rightarrow X \rightarrow Y \rightarrow \whB$, where $B$ and $\whB$ are binary random variables, and the role of the principal inertia components in characterizing the relation between $B$ and $\whB$. We show that  this relation is linked to solving a non-linear maximization problem, which, in turn, can be solved when $\whB$ is an unbiased estimate of $B$ (i.e. $\EE{B}=\mathbb{E}[\whB])$, the joint distribution of $X$ and $Y$ is symmetric and $\Pr\{B=\whB=0\}\geq \EE{B}^2$. We illustrate this result for the setting where $X$ is a binary string and $Y$ is the result of sending $X$ through a memoryless binary symmetric channel. We note that this is a similar setting to the one considered by Anantharam \etal in \cite{anantharam_hypercontractivity_2013}.

The rest of the section is organized as follows. Section \ref{sec:conf} introduces the notion of conforming distributions and ancillary results. Section \ref{sec:functions} presents results concerning the role of the PICs in inferring one-bit functions of $X$ from an observation of $Y$ and in the  transformation of $p_X$ into $p_Y$ in certain symmetric settings. We argue that, in such settings, the PICs can be viewed as singular values (filter coefficients) in a linear transformation. In particular, results for binary channels with additive noise are derived using techniques inspired by Fourier analysis of Boolean functions. Furthermore, Section \ref{sec:functions} also introduces a conjecture that encompasses the one made by Kumar and Courtade in \cite{kumar_which_2013}. Finally, Section \ref{sec:estimators} provides further evidence for this conjecture by investigating the  Markov chain $B\rightarrow X\rightarrow Y \rightarrow \whB$ where $B$ and $\whB$ are binary random variables. Throughout this section we assume $X$ and $Y$ are discrete random variables defined over a finite support set.

%

\subsection{Conforming distributions}
\label{sec:conf}


In this section we shall focus on probability distributions that meet the following definition.

\begin{defn}
  A joint distribution $p_{X,Y}$ is said to be \textit{conforming} if the corresponding matrix $\bP$ satisfies $\bP=\bP^T$ and $\bP$ is positive-semidefinite. 
\end{defn}

Conforming distributions are particularly interesting since they are closely related to symmetric channels\footnote{We say that a channel is symmetric if $\Pygx=\Pygx^T$.}. In addition, if a joint distribution is conforming, then its eigenvalues are equal to (the square root of) its PICs when its marginal distributions are identical. We shall illustrate this relation in the following two lemmas and in Section \ref{sec:functions}.

\begin{remark}
If $X$ and $Y$ have a conforming joint distribution, then they have the same
marginal distribution.  Consequently, $\mathbf{D}\defined
\mathbf{D}_X=\mathbf{D}_Y$, and $\bP=\mathbf{D}^{1/2}\mathbf{U\Sigma U}^T
\mathbf{D}^{1/2}$ (cf. Definition \ref{def:Q} for notation). 
\end{remark}

\begin{lem}
  \label{lem:conf}
    If $\bP$ is conforming, then the corresponding conditional distribution matrix $\Pygx$ is positive semi-definite. Furthermore, for any symmetric channel $\Pygx=\Pygx^T$, there is an input distribution $\px$ (namely, the uniform distribution) such that the PICs of $\bP=\bD_X\Pygx$ correspond to the square of the eigenvalues of $\Pygx$. In this case, if $\Pygx$ is also positive-semidefinite, then the resulting $\bP$ is conforming. 
\end{lem}
\begin{proof}
  Let $\bP$ be conforming and $\calX=\calY=[m]$. Then $\Pygx = \mathbf{D}^{-1/2}\mathbf{U\Sigma U}^T \mathbf{D}^{1/2}=\left( \bD^{-1/2}\bU \right)\bSigma \left( \bD^{-1/2}\bU \right)^{-1}$. It follows that $\diag{\bSigma}$ are the eigenvalues of $\Pygx$, and, consequently, $\Pygx$ is positive semi-definite.

  Now let $\Pygx=\Pygx^T=\bU\bLambda\bU^T$. The entries of $\bLambda$ here are the eigenvalues of $\Pygx$ and not necessarily positive. Since $\Pygx$ is symmetric, it is also doubly stochastic, and for $X$ uniformly distributed $Y$ is also uniformly distributed. Thus, the resulting joint distribution matrix $\bP$ is symmetric, and $\bP=\bU\bLambda\bU^T/m$. It follows directly that the principal inertia components of $\bP$ are  the diagonal entries of $\bLambda^2$, and if $\Pygx$ is positive-semidefinite then $\bP$ is conforming.

\end{proof}

The $q$-ary symmetric channel, defined below, is of particular interest to some of the results derived in the following subsections.

\begin{defn}
  The $q$-ary symmetric channel with crossover probability $\epsilon\leq1-q^{-1}$, also
    denoted as  $(\epsilon,q)$-SC, is
    defined as the channel with input $X$ and output $Y$ where
    $\calX=\calY=[q]$ and
    \begin{align*}
        p_{Y|X}(y|x) = 
            \begin{cases}
                1-\epsilon& \mbox{if } x=y\\
                \displaystyle \frac{\epsilon}{q-1}& \mbox{if } x\neq y.
            \end{cases}
    \end{align*}
\end{defn}

In the rest of this section, we assume that $X$ and $Y$ have a conforming joint distribution matrix with $\calX=\calY=[q]$ and PICs $\lambda_k(X;Y)=\sigma_k^2$ for $k\in [d-1]$. The following  lemma shows that a conforming $\bP$ with uniform marginals can be transformed into the joint distribution of a $q$-ary symmetric channel with input distribution $p_X$ by setting $\sigma_1^2=\sigma_2^2=\dots=\sigma_{q-1}^2$, i.e. making all principal inertia components equal to the largest one.

\begin{lem}
  \label{lem:qary}
  Let $\bP$ be a conforming joint distribution matrix of $X$ and
  $Y$, with $\calX=\calY=[q]$,
  $\bP=\bD^{1/2}\mathbf{U\Sigma U}^T\bD^{1/2}$, where $\bD = \bD_X$ and
  $\mathbf{\Sigma}=\diag{1,\sigma_1,\dots,\sigma_d}$. For
  $\tilde{\mathbf{\Sigma}}=\diag{1,\sigma_1,\dots,\sigma_1}$, let $X$ and
  $\tilde{Y}$ have joint distribution $\tilde{\bP}=\bD^{1/2}\mathbf{U\widetilde{\Sigma}
U}^T \bD^{1/2}$.  Then,  $\tilde{Y}$ is output of a channel with input $X$ and probability transition matrix
\begin{equation}
\label{eq:btilde}
\bP_{\tilde{Y}|X} = \sigma_1 \eye + (1-\sigma_1) \ones \Px^T.
\end{equation}
In particular, if $X$ is uniform, $\tilde{Y}$ is the output of an $(\epsilon,q)$-SC with input $X$, where
  \begin{align}
    \label{eq:epsilon}
    \epsilon = \frac{(q-1)(1-\rho_m(X;Y))}{q}.
  \end{align}
\end{lem}
\begin{proof}
  The first column of $\mathbf{U}$ is $\Px^{1/2}$. Therefore
\begin{align}
    \tilde{\bP}&=\bD^{1/2}\mathbf{U\tilde{\Sigma} U}^T \bD^{1/2}\nonumber \\
    &=\sigma_1 \bD +(1-\sigma_1)\Px\Px^T. \label{eq:qarydist}
\end{align}    
By left multiplying $ \tilde{\bP}$ by $\bD^{-1}$, we obtain the channel transition matrix given in \eqref{eq:btilde}. 
\end{proof}



\begin{remark}
  For $X$, $Y$ and $\tilde{Y}$ given in the previous lemma, a natural question that arises is whether $Y$ is a degraded version of $\Yt$, i.e. $X\rightarrow \tilde{Y}\rightarrow Y$. Unfortunately, this is not true in general, since the matrix $\mathbf{U\widetilde{\Sigma}^{-1}\Sigma U^T}$ does not necessarily contain only positive entries, although it is doubly-stochastic. However, since the PICs of $X$ and $\tilde{Y}$ upper bound the PICs of $X$ and $Y$, it is natural to expect that, at least in some sense, $\tilde{Y}$ is more informative about $X$ than $Y$. This intuition is indeed correct for certain estimation problems where a one-bit function of $X$ is to be inferred from a single observation $Y$ or $\tilde{Y}$, and will be investigated in the next subsection. In addition, using the characterization of the PICs in Theorem \ref{thm:PIC_Charac}, it follows that \textit{any} function of $X$ can be inferred with smaller MMSE from $\tilde{Y}$ than from $Y$. Consequently, even if, for example $I(X;\tilde{Y})\leq I(X;Y)$, any function of $X$ can be estimated with smaller MMSE for $\tilde{Y}$ than from $Y$.
\end{remark}

\subsection{One-bit Functions and Channel Transformations}
\label{sec:functions}

Let $B\rightarrow X \rightarrow Y$, where $B$ is a binary random variable. When $X$ and $Y$ have a conforming probability distribution, the PICs of $X$ and $Y$ have a particularly interesting interpretation: they can be understood as the filter coefficients in a linear transformation from $p_{B|X}$ into $p_{B|Y}$, as we explain next. Consider the joint distribution of $B$ and $Y$, denoted here by $\bB$, given by 
\begin{align}
  \label{eq:OBF}
  \bB \defined [\bx~ ~1-\bx]^T\bP= [\bx~ ~1-\bx]^T\bP_{X|Y}\bD_Y=[\by~ ~ 1-\by]^T\bD_Y,
\end{align}
where $\bx\in \Reals^m$ and $\by\in \Reals^n$ are column-vectors with entries $x_i=p_{B|X}(0|i)$ and $y_j=p_{B|Y}(0|j)$. In particular, if $B$ is a deterministic function of $X$, $\bx\in \{0,1\}^m$.

If $\bP$ is conforming and $\calX=\calY=[m]$, then $\bP=\bD^{1/2}\bU\bSigma\bU^T\bD^{1/2}$, where $\bD=\bD_X=\bD_Y$. Assuming $\bD$ fixed, the joint distribution $\bB$ is entirely specified by the linear transformation of $\bx$ into $\by$. Denoting $\bT\defined \bU^T\bD^{1/2}$, this transformation is done in three steps:
\begin{enumerate}
  \item (Linear transform) $\widehat{\bx}\defined \bT\bx$,
  \item (Filter) $\widehat{\by}\defined \bSigma\widehat{\bx}$, where the diagonal of
    $\bSigma^2$ are the PICs of $X$ and $Y$, 
  \item (Inverse transform)  $\by= \bT^{-1}\widehat{\by}$.
\end{enumerate}
Note that $\widehat{x}_1=\widehat{y}_1=1-\EE{B}$ and  $\widehat{\by}=
\bT\by$. Consequently, the PICs of $X$ and $Y$ correspond to the singular values (or filter coefficients) of the linear transformation of  $p_{B|X}(0|\cdot)$
into $p_{B|Y}(0|\cdot)$.




A similar interpretation can be made for symmetric channels, where
$\Pygx=\Pygx^T=\bU\bLambda\bU^T$ and $\Pygx$ acts as the matrix of the linear
transformation of $\Px$ into
$\Py$.  Note that $\Py = \Pygx\Px$, and,
consequently, $\Px$ is transformed into $\Py$ in the same three steps as before:
\begin{enumerate}
  \item (Linear transform) $\widehat{\Px}= \bU^T\Px$,
  \item (Filter) $\widehat{\Py}\defined\bLambda\widehat{\Px}$, where the diagonal of $\bLambda^2$ is the PICs of $X$ and $Y$ in the particular case when $X$ is uniformly distributed (Lemma \ref{lem:conf}),
  \item (Inverse transform)  $\Py= \bU\widehat{\Py}$.

\end{enumerate} From this perspective, the vector $\mathbf{z}=\bU\bLambda\ones m^{-1/2}$ can be understood as a proxy for the noise effect of the channel. Note that $\sum_i z_i=1$. However, the entries of $\mathbf{z}$ are not necessarily positive, and $\mathbf{z}$ might not be a  probability distribution.

We now illustrate these ideas by investigating binary channels with additive noise in the next section, where $\bT$ will correspond to the well-known Walsh-Hadamard transform matrix.

\subsection{Example: Binary Additive Noise Channels}
\label{sec:binAdd}

In this example, let $\calX^n, \calY^n\subseteq \{-1,1\}^n$ be the support sets of $X^n$ and $Y^n$, respectively.  We define two sets of channels that maps $X^n$ to $Y^n$. In each set definition, we  assume the conditions for $p_{Y^n|X^n}$ to be a valid probability distribution (i.e. non-negativity and unit sum).

 \begin{defn}
    The set of \textit{parity-changing channels} of block-length $n$, denoted
    by $\calA_{n}$, is defined as:
    \begin{align}
      \calA_{n} \defined \left\{
        p_{Y^n|X^n} \mid \forall
        \calS\subseteq[n],~\exists c_{\calS} \in
        [-1,1]  \mbox{ s.t. }
        \EE{\chi_\calS(Y^n)|X^n}=c_{\calS}\chi_\calS(X^n)  \right\} \label{eq:PAchannels},
    \end{align}
    where $\chi_\calS(\cdot)$ is defined in \eqref{eq:defchi}.
    The set of all \textit{binary additive noise channels} is given by
        \begin{align}
          \mathcal{B}_{n} \defined \left\{
        p_{Y^n|X^n} \mid  \exists Z^n \mbox{ s.t. }
        Y^n=X^n\oplus Z^n, \mbox{ supp}(Z^n)\subseteq\{-1,1\}^n, Z^n\independent X^n \right\}.
        \label{eq:PBBchannels}
    \end{align}
 \end{defn}
 The definition of parity-changing channels is inspired by results from the literature on Fourier analysis of Boolean functions. For an overview of the topic we refer the reader to the survey \cite{odonnell_topics_2008}.  The set of binary additive noise channels, in turn, is widely used in the information theory literature. The following lemma shows that both characterizations are equivalent.
%
%
%

 \begin{lem}
   \label{thm:AB}
For $\calA_n$ and $\calB_n$ given in \eqref{eq:PAchannels} and \eqref{eq:PBBchannels}, respectively,  $\calA_n=\mathcal{B}_n$.
 \end{lem}
 \begin{proof}
 The proof is in Appendix \ref{app:proofsPIC_IT}.

  \end{proof}

  The previous theorem suggests that there is a correspondence between the
  coefficients $c_\calS$ in \eqref{eq:PAchannels} and the distribution of the
  additive noise $Z^n$ in the definition of $\mathcal{B}_n$. The next result
  shows that this is indeed the case and, when $X^n$ is uniformly distributed,
  the coefficients $c_\calS^2$ correspond to the PICs
  of $X^n$ and $Y^n$.

  \begin{thm}
    \label{thm:PIbinary}
    Let $p_{Y^n|X^n}\in \mathcal{B}_n$, and $X^n\sim p_{X^n}$. Then
    $\bP_{X^n,Y^n}=\bD_{X^n}\bbH_{2^n}\bLambda \bbH_{2^n}$, where $\bbH_l$ is
    the $l\times l$ normalized Hadamard matrix\footnote{We define the normalized Hadamard matrix $\bbH_{2^k}$ as
    $ \bbH_1\defined [1]$, 
    \begin{equation*}
    \bbH_2\defined \frac{1}{\sqrt{2}}
    \begin{bmatrix} 
    1&1\\
    1&-1
    \end{bmatrix},
    \end{equation*}
    and $\bbH_{2^{k}}\defined \bbH_2\otimes \bbH_{2^{k-1}}$.}
    (hence $\bbH_l^2=\eye$).
    Furthermore, for $Z^n\sim p_{Z^n}$,  $\diag{\bLambda} =
    2^{n/2}\bbH_{2^n}\mathbf{p}_{Z^n}$, and the diagonal entries of $\bLambda$
    are equal to $c_\calS$ in \eqref{eq:PAchannels}. Finally, if $X$ is
    uniformly distributed, then $c_\calS^2$ are the principal inertia components
    of $X^n$ and $Y^n$.
  \end{thm}
  \begin{proof}
    Let  $p_{Y^n|X^n}\in \calA_n$ be given. From Lemma \ref{thm:AB} and the
    definition of $\calA_n$, it follows that $\chi_\calS(Y^n)$ is a right
    eigenvector of $p_{Y^n|X^n}$ with corresponding eigenvalue $c_\calS$.  Since
    $\chi_\calS(Y^n)2^{-n/2}$ corresponds to a row of $\bbH_{2^n}$ for each
    $\calS$ (due to the Kronecker product construction of the Hadamard matrix) and $\bbH_{2^n}^2=\eye$, then
    $\bP_{X^n,Y^n}=\bD_{X^n}\bbH_{2^n}\bLambda \bbH_{2^n}$. Finally, note that
    $\mathbf{p}_Z^T=2^{-n/2}\ones^T\bLambda\bbH_{2^n}$. From Lemma
    \ref{lem:conf}, it follows that  $c_\calS^2$ are the PICs of $X^n$ and $Y^n$ if $X^n$ is uniformly distributed.
  \end{proof}

  \begin{remark}
    Theorem \ref{thm:PIbinary} suggests that one possible method for estimating the distribution of the additive binary noise $Z^n$ is to estimate  its effect on the  parity bits of $X^n$ and $Y^n$. In this case, we are estimating the coefficients $a_\calS$ of the Walsh-Hadamard transform of $p_{Z^n}$. This approach was  studied by Raginsky \etal in \cite{raginsky_recursive_2013} and in other learning literature (see \cite{odonnell_analysis_2014} and the references therein).
  \end{remark}

Theorem \ref{thm:PIbinary} illustrates the filtering role of the principal inertia components (discussed in Section \ref{sec:functions}) in binary additive noise channels. If $X^n$ is uniform, then the vector of conditional probabilities $\Px$ is transformed into the vector of \textit{a posteriori} probabilities $\Py$ by: (i) taking the Hadamard transform of $\Px$, (ii) filtering the transformed vector according to the coefficients $c_\calS$ (these coefficients have a one-to-one mapping to the entries of the vector resulting from the Hadamard transform of $\bp_Z$), and (iii) taking the inverse Hadamard transform to recover $\Py$. 

\subsection[The Information of a Boolean Function of the Input of a
 Channel]{Quantifying the Information of a Boolean Function of the Input of a
Noisy Channel}

We now investigate the connection between the PICs and $f$-information (cf. Eq. \eqref{eq:defn_finf}) in the context of one-bit functions of $X$. Recall from the discussion in the beginning of this section and, in particular, equation \eqref{eq:OBF},  that for a binary $B$ and $B\rightarrow X \rightarrow Y$, the distribution of $B$ and $Y$ is entirely specified by the transformation of $\xb$ into $\yb$, where $\xb$ and $\yb$ are vectors with entries equal to $p_{B|X}(0|\cdot)$ and $p_{B|Y}(0|\cdot)$, respectively.

For $\EE{B}=1-a$, the $f$-information between $B$ and $Y$ is given by (cf. \eqref{eq:defn_finf})
    \begin{align*}
            I_f(B;Y) =\ExpVal{}{ a f\left( \frac{p_B(0|Y)}{a}
            \right) +(1-a) f\left( \frac{1-p_B(0|Y)}{1-a}
            \right) }.
    \end{align*}

For $0\leq r,s \leq 1$, and since $f$ is smooth with $f(1)=0$, we can expand $f\left( \frac{r}{s} \right)$ around
1 as

    \begin{equation*}
        f\left( \frac{r}{s} \right) = \sum_{k=1}^\infty  \frac{
        f^{(k)}(1)}{k!}\left( \frac{r-s}{r} \right)^k.
    \end{equation*}
Denoting 
    \begin{align*}
        c_k(\alpha) &\defined
        \frac{1}{a^{k-1}}+\frac{(-1)^k}{(1-a)^{k-1}},
    \end{align*}
the $f$-information can then be expressed as
    \begin{align}
        I_f(B;Y)& = \sum_{k=2}^\infty  \frac{
         f^{(k)}(1)c_k(a)}{k!}\ExpVal{}{(p_B(0|Y)-a)^k}.
     \label{eq:momentDecomp}
    \end{align}

    Similarly to \cite[Chapter 4]{_information_2004}, for a fixed $\EE{B}=1-a$, maximizing the PICs of $X$ and $Y$ will always maximize the first term in the expansion \eqref{eq:momentDecomp}. To see why this is the case, observe that    
    \begin{align}
        \ExpVal{}{(p_{B|Y}(0|Y)-a)^2} &= (\by-a)^T\bD_Y(\by-a) \nonumber\\
                               &= \by^T\bD_Y\by -a^2 \nonumber\\
                               &=\bx^T\bD_X^{1/2}\bU\bSigma^2\bU^T\bD_X^{1/2}\bx-a^2.\label{eq:var}
     \end{align}
For a fixed $a$ and any $\bx$ such that $\bx^T\ones=a$, \eqref{eq:var} is non-decreasing in the diagonal  entries of $\bSigma^2$ which, in turn, are exactly the PICs of $X$ and $Y$. Equivalently, \eqref{eq:var} is non-decreasing in the $\chi^2$-divergence between $\pxy$ and $\px\py$.

However, we do note that increasing the PICs does not increase the $f$-information between $B$ and $Y$ in general. Indeed, for a fixed $\bU$, $\bV$ and marginal distributions of $X$ and $Y$, increasing the PICs might not even lead to a valid probability distribution matrix $\bP$.

Nevertheless, if $\bP$ is conforming and $X$ and $Y$ are uniformly distributed over $[q]$, as shown in Lemma \ref{lem:qary}, by increasing the PICs we can define a new random variable $\tilde{Y}$ that results from sending $X$ through a $(\epsilon,q)$-SC, where $\epsilon$ is given in \eqref{eq:epsilon}. In this case, the $f$-information between $B$ and $Y$ has a simple expression when $B$ is a function of $X$.

\begin{lem}
  \label{lem:qaryIf}
  Let $B\rightarrow X \rightarrow \Yt$, where $B=b(X)$ for some $b:[q]\rightarrow
  \{0,1\}$,  $\EE{B}=1-a$ where $aq$ is an integer, $X$ is uniformly distributed in $[q]$ and $\Yt$ is the
  result of passing $X$ through a $(\epsilon,q)$-SC with $\epsilon\leq (q-1)/q$. Then
  \begin{equation}
    \label{eq:genaralqary}
    I_f(B;\Yt)=a^2f\left( 1+\sigma_1 c
    \right)+2a(1-a)f\left(1-\sigma_1\right)+(1-a)^2f\left( 1+\sigma_1c^{-1} \right)
  \end{equation}
  where $\sigma_1=\rho_m(X;\Yt)=1-\epsilon q (q-1)^{-1}$ and $c\defined (1-a)a^{-1}$. In particular,
  for $f(x)=x\log x$, then $I_f(X;\Yt)=I(X;\Yt)$, and for $\sigma_1 = 1-2\delta$
    \begin{align}
        I(B;\Yt)&=h_b(a)-\alpha h_b\left( 2\delta(1-a) \right)-(1-a)h_b(2\delta
        a) \label{eq:qaryMI}\\
        &\leq 1-h_b(\delta) \label{eq:qaryMax}.
    \end{align}{}
  where $h_b(\cdot)$ is the binary entropy function, defined in \eqref{eq:defhb}.
\end{lem}

\begin{proof}
  Since $B$ is a deterministic function of $X$ and $aq$ is an integer, $\bx$ is a vector with $aq$ entries equal to 1 and $(1-a)q$ entries equal to 0. It follows from \eqref{eq:qarydist} that  
  \begin{align*}
    I_f(B;\Yt)=&\frac{1}{q}\sum_{i=1}^q af\left(
    \frac{(1-\sigma_1)a+x_i\sigma_1}{a} \right)+(1-a)f\left(
    \frac{1-(1-\sigma_1)a-x_i\sigma_i}{1-a} \right)\\
    =&a^2f\left( 1+\sigma_1 \frac{1-a}{a}
    \right)+2a(1-a)f\left(1-\sigma_1\right)+(1-a)^2f\left(
    1+\sigma_1\frac{a}{1-a} \right).
  \end{align*}
  Letting $f(x)=x\log x$, \eqref{eq:qaryMI} follows immediately. Since
  \eqref{eq:qaryMI} is concave in $a$ and symmetric around $a=1/2$, it is
  maximized at $a=1/2$, resulting in \eqref{eq:qaryMax}.
\end{proof}

\subsection{On the ``Most Informative Bit''}
We now return to channels with additive binary noise, analyzed in Section \ref{sec:binAdd}. Let $X^n$ be a uniformly distributed binary string of length $n$ ($\calX = \{-1,1\})$ and $Y^n$ be the result of passing $X^n$ through a memoryless binary symmetric channel with crossover probability $\delta\leq 1/2$. Kumar and Courtade conjectured  \cite{kumar_which_2013} that for all binary $B$ and $B\rightarrow X^n \rightarrow Y^n$ we have \begin{equation} I(B;Y^n)\leq 1-h_b(\delta).~~\mbox{(conjecture)} \label{eq:conjI} \end{equation} It is sufficient to consider $B$ a function of $X^n$, denoted by $B=b(X^n)$, $b:\{-1,1\}^n\rightarrow \{0,1\}$, and we make this assumption henceforth.

From the discussion in Section \ref{sec:binAdd}, for the memoryless binary symmetric channel $Y^n=X^n\oplus Z^n$, where $Z^n$ is an i.i.d. string with $\Pr\{Z_i=1\}=1-\delta$, and any $\calS\in [n]$,
\begin{align*}
  \EE{\chi_\calS(Y^n)|X^n}&=
  \chi_\calS(X^n)\left(\Pr\left\{\chi_\calS(Z^n)=1\right\}-\Pr\left\{\chi_\calS(Z^n)=-1\right\}\right)\\
                          &=
                          \chi_\calS(X^n)\left(2\Pr\left\{\chi_\calS(Z^n)=1\right\}-1\right)\\
                          &=\chi_\calS(X^n)(1-2\delta)^{|\calS|}.
\end{align*}
It follows directly that $c_\calS = (1-2\delta)^{|\calS|}$ for all $\calS\subseteq [n]$. Consequently, from Theorem \ref{thm:PIbinary}, the principal inertia components of $X^n$ and $Y^n$ are of the form $(1-2\delta)^{2|\calS|}$ for some $\calS\subseteq [n]$. Observe that the principal inertia components act, broadly speaking, as a low pass filter on the vector of conditional probabilities $\bx$ given in \eqref{eq:OBF}, since it attenuates the high order interaction terms in the Walsh-Hadamard transform of $\bx$.

Can the noise distribution be modified so that the principal inertia components act as an all-pass filter? More specifically, what happens when $\Yt^n=X^n\oplus W^n$, where $W^n$ is such that the principal inertia components between $X^n$ and $\Yt^n$ satisfy $\sigma_i=1-2\delta$? Then, from Lemma \ref{lem:qary}, $\Yt^n$ is the result of sending $X^n$ through a $(\epsilon,2^n)$-SC with $\epsilon=2\delta(1-2^{-n})$. Therefore, from \eqref{eq:qaryMax},
    \begin{equation*}
        I(B;\Yt^n)\leq 1-h_b(\delta).
    \end{equation*}

For any function $b:\{-1,1\}^n\rightarrow \{0,1\}$ such that $B=b(X^n)$, from standard results in Fourier analysis of Boolean functions  \cite[Prop.  1.1]{odonnell_topics_2008}, $b(X^n)$ can be expanded as \[b(X^n)=\sum_{\calS\subseteq [n]}\beta_\calS \chi_\calS(X^n) .\] The value of $B$ is uniquely determined by the action of $b$ on $\chi_\calS(X^n)$. Consequently, for a fixed function $b$, one could expect that $\Yt^n$ should be more informative about $B$ than $Y^n$, since the parity bits $\chi_\calS(X^n)$ are more reliably estimated from $\Yt^n$ than from $Y^n$. Indeed, the memoryless binary symmetric channel attenuates $\chi_\calS(X^n)$ exponentially  in $|\calS|$, acting (as argued previously) as a low-pass filter. In addition, if one could prove that for any fixed $b$ the inequality  $I(B;Y^n)\leq I(B;\Yt^n)$ holds, then \eqref{eq:conjI} would be proven true. This motivates the following conjecture.


\begin{conj}
  \label{conj}
    For all $b:\{-1,1\}^n\rightarrow \{0,1\}$ and $B=b(X^n)$ 
    \begin{align*}
        I(B;Y^n)\leq I(B;\Yt^n).
    \end{align*}
\end{conj}

We note that Conjecture \ref{conj} is false if $B$ is not a deterministic function of $X^n$. In the next section, we provide further evidence for this conjecture by investigating information metrics between $B$ and an estimate $\whB$ derived from $Y^n$.

\subsection{One-bit Estimators}
\label{sec:estimators}

Let $B\rightarrow X \rightarrow Y \rightarrow \widehat{B}$, where $B$ and $\whB$
are binary random variables with $\EE{B}=1-a$ and $\mathbb{E}[\whB]=1-b$. 
Again, we let $\xb\in \Reals^m$ and $\yb\in \Reals^n$ be the
column vectors with entries $x_i = p_{B|X}(0|i)$ and $y_j =
p_{\whB|Y}(0|j)$. The joint distribution matrix of $B$ and $\whB$ is given by
\begin{equation}
  \label{eq:Pbbh}
  \mathbf{P}_{B,\whB}=\left(
\begin{array}{cc}
 z  & a-z  \\
 b-z  & 1-a-b+z \\
\end{array}
\right),
\end{equation}
where $z=\xb^T\bP\yb=\Pr\{B=\whB=0 \}$. For fixed values of $a$ and $b$, the joint distribution
of $B$ and $\whB$ only depends on $z$.

Let $f:\calP_{2\times 2}\rightarrow \Reals$, and, with a slight abuse of
notation, we also denote $f$ as a function of the entries of the $2\times 2$ matrix as
$f(a,b,z)$. If $f$ is convex in $z$ for a fixed $a$ and $b$, then $f$ is
maximized at one of the extreme values of $z$. Examples of such functions $f$
include mutual information and expected error probability.  Therefore,
characterizing the maximum and minimum values of $z$ is equivalent to
characterizing the maximum value of $f$ over all possible mappings $X\rightarrow
B$ and $Y\rightarrow \whB$. This leads to the following definition. 
\begin{defn}
    For a fixed $\bP$ and given  $\EE{B}=1-a$ and $\mathbb{E}[\whB]=1-b$, the minimum and maximum
    values of $z$ over all possible mappings $X\rightarrow B$ and $Y\rightarrow
    \whB$ are defined as
    \begin{equation*}
      z^*_l(a,b,\bP) \defined \min_{\substack{\xb \in \calC^m(a,\bP^T)\\\yb \in
      \calC^n(b,\bP)}}
      \xb^T\bP\yb ~\mbox{ and }~
       z^*_u(a,b,\bP) \defined \max_{\substack{\xb \in \calC^m(a,\bP^T)\\\yb \in
       \calC^n(b,\bP)}}
       \xb^T\bP\yb,
    \end{equation*}
    respectively, and $\calC^{n}(a,\bP)$ is defined in \eqref{eq:Cdef}.
    
\end{defn}


The next lemma provides a simple upper-bound for $z^*_u(a,b,\bP)$ in terms of the largest principal inertia components or, equivalently, the maximal correlation between $X$ and $Y$.
\begin{lem}
  \label{lem:zupper}
  $z_u^*(a,b,\bP)\leq ab+\rho_m(X;Y)\sqrt{a(1-a)b(1-b)}$.
\end{lem}
\begin{proof}
The proof is in Appendix \ref{app:proofsPIC_IT}.
\end{proof}

\begin{remark}
  An analogous result was derived by Witsenhausen  \cite[Thm. 2]{witsenhausen_sequences_1975} for bounding the probability of agreement of a common bit derived from two correlated sources.
\end{remark}

We will focus in the rest of this section on functions and corresponding estimators that are (i) unbiased ($a=b$) and (ii) satisfy $z=\Pr\{\hat{B}=B=0\}\geq a^2$. The set of all such mappings is given by
\begin{equation*}  
    \mathcal{ H }(a,\bP)\defined \left\{ (\bx,\by)\mid \bx\in
    \calC^m(a,\bP^T),\by\in \calC^n(a,\bP),\bx^T\bP \by\geq a^2 \right\}.
\end{equation*}

The next results provide upper and lower bounds on $z$ for the mappings in $\calH(a,\bP)$.
\begin{lem}
  \label{lem:zbounds}
    Let $0\leq a\leq 1/2$ and $\bP$ be fixed. For any $(\bx,\by)\in  \calH(a,\bP)$
        \begin{equation}
          \label{eq:zbounds}
            a^2\leq z \leq a^2+\rho_m(X;Y)a(1-a),
        \end{equation}  
        where  $z=\bx^T\bP\by$.
\end{lem}
\begin{proof}
    The lower bound for $z$ follows directly from the definition of
    $\calH(a,\bP)$, and the upper bound follows from Lemma \ref{lem:zupper}. 
\end{proof}

The previous lemma allows us to provide an upper bound over the mappings in $\calH(a,\bP)$ for the $f$-information between $B$ and $\whB$ when $I_f$ is non-negative.
\begin{thm}
  \label{thm:Estimators}
  For any non-negative $I_f$  and fixed $a$ and $\bP$,
  \begin{equation} 
    \label{eq:unbiasedBound}
    \sup_{ (\bx,\by) \in \calH(a,\bP)} I_f(B;\hat{B})\leq  a^2f\left( 1+\sigma_1 c
    \right)+2a(1-a)f\left(1-\sigma_1\right)+(1-a)^2f\left( 1+\sigma_1c^{-1} \right)
  \end{equation}
  where here $\sigma_1=\rho_m(X;\Yt)$ and $c\defined (1-a)a^{-1}$. In particular, for
  $a=1/2$,
    \begin{equation} 
    \sup_{ (\bx,\by) \in \calH(1/2,\bP)} I_f(B;\hat{B})\leq
  \frac{1}{2}\left( f(1-\sigma_1)+f(1+\sigma_1) \right). 
  \end{equation}
\end{thm}
\begin{proof}
     Using the matrix form of the joint distribution between $B$ and $\whB$
     given in \eqref{eq:Pbbh}, for $\EE{B}=\EE{\whB}=1-a$, the $f$ information is given by
    \begin{align}
      \label{eq:Ifproof}
      I_f(B;\hat{B}) = a^2f\left( \frac{z}{a^2} \right)+ 2a(1-a)f\left(
      \frac{a-z}{a(1-a)} \right)
                        + (1-a)^2f\left( \frac{1-2a+z}{(1-a)^2} \right).
    \end{align}
    Consequently, \eqref{eq:Ifproof} is convex in $z$. For $(\bx,\by) \in \calH(a,\bP)$, it follows from Lemma \ref{lem:zbounds} that $z$ is restricted to the interval in \eqref{eq:zbounds}. Since  $I_f(B;\hat{B})$ is non-negative by assumption, $I_f(B;\hat{B})=0$ for $z=a^2$ and \eqref{eq:Ifproof} is convex in $z$, then $I_f(B;\hat{B})$  is non-decreasing in $z$ for $z$ in \eqref{eq:zbounds}. Substituting $z=a^2+\rho_m(X;Y)a(1-a)$ in \eqref{eq:Ifproof}, inequality \eqref{eq:unbiasedBound} follows.
  \end{proof}

\begin{remark}
  Note that the right-hand side of \eqref{eq:unbiasedBound} matches  the right-hand side of \eqref{eq:genaralqary}, and provides further evidence for Conjecture \ref{conj} by demonstrating that the conjecture holds for the specific case when $B\to X\to Y\to \hat{B}$ and $\EE{B}=\mathbb{E}[\hat{B}]$. Moreover, this result indicates that, for conforming probability distributions, the information between a binary function and its corresponding unbiased estimate is maximized when all the PICs have the same value.
\end{remark}

Following the same approach from Lemma \ref{lem:qaryIf}, we find the next 
bound for the mutual information between $B$ and $\whB$. 
\begin{cor}
  \label{cor:MIEstimators}
  For $0\leq a \leq 1$ and $\rho_m(X;Y)=1-2\delta$ 
  \begin{align*}
    \sup_{ (p_{B|X},p_{\whB|Y} )  \in \calH(a,\bP)} I(B;\hat{B})\leq
    1-h_b(\delta). 
  \end{align*}
\end{cor}

We  provide next a few application examples for the results derived in this section.

\begin{example}[Memoryless Binary Symmetric Channels with Uniform Inputs]
We turn our attention back to the setting considered in Section \ref{sec:binAdd}. Let $Y^n$ be the result of passing $X^n$ through a memoryless binary symmetric channel with crossover probability $\delta$, $X^n$ uniformly distributed, and $B\rightarrow X^n\rightarrow Y^n\rightarrow \whB$. Then $\rho_m(X^n;Y^n)=1-2\delta$ and, from
\eqref{eq:Fano}, when $\EE{B}=1/2$, 
\begin{equation*}
     \Pr\{B\neq \whB\}\geq \delta.
\end{equation*}
Consequently, inferring any unbiased one-bit function of the input of a binary symmetric channel is at least as hard (in terms of error probability) as inferring a single output from a single input.

Using the result from Corollary \ref{cor:MIEstimators}, it follows that when
$\EE{B}=\EE{\whB}=a$ and $\Pr\{B=\whB=0\}\geq a^2 $, then 
\begin{equation}
  \label{eq:Sudeep}
    I(B;\whB)\leq 1-h_b(\delta).
\end{equation}

\end{example}
\begin{remark}
%
    Anantharam \etal presented in  \cite{anantharam_hypercontractivity_2013} a computer aided proof that the upper bound \eqref{eq:Sudeep} holds  for any $B\rightarrow X^n\rightarrow Y^n\rightarrow \whB$. Nevertheless, we highlight that the methods introduced here allowed an analytical derivation of  \eqref{eq:Sudeep} for unbiased estimators.
 \end{remark}

\begin{example}[Lower Bounding the Estimation Error Probability]

For $z$ given in \eqref{eq:Pbbh}, the average estimation error probability is given by $\Pr\{B\neq \whB\}=a+b-2z$, which is a convex (linear) function of $z$. If  $a$ and $b$ are fixed, then the error probability  is minimized when $z$ is maximized. Therefore
\begin{equation*}
   \Pr\{B\neq \whB\}\geq a+b-2z_u^*(a,b). 
\end{equation*}
Using the bound from Lemma \ref{lem:zupper}, it follows that
\begin{equation}
  \label{eq:errorProb}
  \Pr\{B\neq \whB\}\geq a+b-2ab-2\rho_m(X;Y)\sqrt{a(1-a)b(1-b)}. 
\end{equation}
The bound \eqref{eq:errorProb} is exactly the bound derived by Witsenhausen in
\cite[Thm 2.]{witsenhausen_sequences_1975}. Furthermore, minimizing the right-hand side of  \eqref{eq:errorProb}
over $0\leq b \leq 1/2$, we arrive at
\begin{equation}
  \label{eq:Fano}
      \Pr\{B\neq \whB\}\geq\frac{1}{2}\left(
  1-\sqrt{1-4a(1-a)(1-\rho_m(X;Y)^2)} \right). 
  \end{equation}
  This result suggests that the PICs are particularly useful for deriving bounds on error probability. We explore this fact in the next section, and show that \eqref{eq:Fano} is a particular form of a more general bound derived in Theorem \ref{thm:Bound}. 

\end{example}

%% file: estimation-to-inference.tex
In this section we derive lower bounds on error-probability based on the PICs (cf. section \ref{chap:PICs}). Before presenting these bounds, we discuss the general approach used for deriving lower bounds, which can be extended to other measures of dependence. This approach is particularly useful for proving information-theoretic security and privacy guarantees.


Recall the central estimation-theoretic problem: Given an observation of a random variable $Y$, what can we learn about a correlated, hidden variable $X$? Such questions are relevant for different application areas. For example, in a symmetric-key encryption setup, $X$ can be the plaintext message, and  $Y$ the ciphertext and any additional side information available to an adversary.  If there is an encryption mechanism  in place that guarantees that the mutual information between an individual symbol of the plaintext $X$ and a cipheretext $Y$ is at most 0.01 bits \cite{inproc:allertonCrypt}, how well can  an adversary guess individual symbols of $X$? How does this result depend on the distribution of the plaintext source? Are there other information measures besides mutual information for deriving such bounds on estimation?

If the joint distribution between $X$ and $Y$ is known, the  probability of error of estimating $X$ given an observation of $Y$ can be calculated exactly.
However, in most practical settings, this joint distribution is unknown.
Nevertheless, it may be possible to estimate certain correlation (dependence) measure of
$X$ and $Y$  reliably, such as  maximal correlation, $\chi^2$ or
mutual information. In general, we will denote this measure as $\calI(X;Y)$.

Given  an upper bound $\theta$ on a certain dependence measure $\calI$, i.e. $\calI(X;Y)\leq \theta$, is it possible to determine a lower bound for the average error of estimating $X$ from $Y$ over all possible estimators? We answer this question in the affirmative. In particular, the problem of computing such a bound for a given distribution $\px$ and $\theta$ is equivalent to computing a distortion-rate function, presented    in Definition \ref{defn:distRate}.  When the estimation metric is error probability, we call the corresponding distortion-rate function the \textit{error-rate function}, denoted by $e_\calI(\px,\theta)$ and given in Definition \ref{defn:ei}. In the context of security and privacy, this bound  characterizes the best  estimation of the plaintext that a (computationally unbounded) adversary can make given an observation of the output of the system in terms of the statistic of the distribution of the input and output. This allows, for example, guarantees on correlation measures frequently used in security and privacy settings to be translated into bounds on the estimation error.  

Recall that $X$ and $Y$ are discrete random variables with  support $\calX=[m]$ and $\calY=[n]$, and, consequently, the joint pmf $\pxy$ can be displayed as the entries of a matrix $\bP\in \Reals^{m\times n}$, where $[\bP]_{i,j}=\pxy(i,j)$. The problem of determining the estimator  $\Xh$ of $X$ given an observation of $Y$ then reduces to finding  a row-stochastic matrix $\bP_{\Xh|Y}\in \Reals^{n\times m}$ that is the solution of
\begin{equation}
\label{eq:bestPe}
P_e(X|Y) = \min_{\bP_{\Xh|Y}}1-\Tr{\bP\times\bP_{\Xh|Y}}.
\end{equation}
Note that the previous minimization is a linear program, and by taking its dual the reader can verify that the optimal $\bP_{\Xh|Y}$ is the maximum a posteriori (MAP) estimator, as expected.

We highlight again that in  applications the joint distribution matrix $\bP$ may not be known exactly -- only a given dependence measure $\calI(\pxy)$  may be known.  Equation \eqref{eq:bestPe} hints that dependence measures that depend on the spectrum of $\bP$ may lead to  sharp lower bounds for error probability. Indeed, the trace of the product of two matrices is closely related to their spectra (cf. Von Neumman's trace inequality \cite[Thm. 7.4.1.1]{horn_matrix_2012}). This motivates the following question: Are there information measures that capture the spectrum of a joint distribution matrix $\bP$? This naturally leads to the consideration of  measures of dependence and lower bounds on estimation error  based on the PICs. These bounds are derived in Section \ref{sec:fanoinertia}, but we first provide an overview of our approach in Section \ref{sec:convexbounds}.


Owing to the nature of the joint distribution, it may be infeasible to estimate $X$ from $Y$ with small estimation error. It is, however, possible that a non-trivial function   $f(X)$ exists that is of interest to a learner and can be  estimated reliably from $Y$.   If $f$ is the identity function, this  reduces to the standard problem of estimating $X$ from $Y$.  Determining if such a function exists is relevant to several applications in learning, privacy, security and information theory. In particular, this setting is related to the information bottleneck method \cite{tishby_information_2000} and functional compression \cite{doshi_functional_2010}, where the goal is to compress $X$ into $Y$ such that $Y$ still preserves information about $f(X)$.

For most security applications, minimizing the average error of estimating a hidden variable $X$ from an observation of $Y$ is insufficient. As argued in \cite{calmon2014allerton}, cryptographic definitions of security, and in particular semantic security \cite{goldwasser_probabilistic_1984}, require that an adversary has negligible advantage in guessing any function of the input given an observation of the output. In light of this, we present bounds for the best possible average error achievable for estimating functions of  $X$ given an observation of $Y$.

Assuming that $\pxy$ is unknown, $\px$ is given and a bound  $\calI(X;Y)\leq \theta$ is known (where $\calI$ is not restricted to being mutual information), we present in Theorem \ref{thm:boundPeM} a method for adapting bounds for error probability into bounds for the average estimation error of functions of $X$ given $Y$. This method depends on a few technical assumptions on the dependence measure (stated in Definition \ref{defn:informationMeasure} and in Theorem \ref{thm:boundPeM}), foremost of which is the existence of a lower bound for the error-rate function that is Schur-concave\footnote{A function $f:\Reals^n\rightarrow \Reals$ is said to be \textit{Schur-concave} if for all $\bx,\by\in \Reals^n$ where $\bx$ is majorized by $\by$, then $f(\bx)\geq f(\by)$.} in $\px$ for a fixed $\theta$. Theorem \ref{thm:boundPeM} then states that, under these assumptions, for any deterministic, surjective function $f:\calX\rightarrow \{1,\dots,M\}$, we can obtain a lower bound for the average estimation error of $f$ by  computing $e_\calI(p_U,\theta)$, where $U$ is a random variable that is a function $X$.

Note that Schur-concavity is crucial for this result. In Theorem \ref{thm:schur}, we show that this condition is always satisfied when $\calI(X;Y)$ is concave in $\px$ for a fixed $\pygx$, convex in $\pygx$ for a fixed $\px$, and satisfies the DPI. This generalizes a result by Ahlswede  \cite{ahlswede_extremal_1990} on the extremal properties of rate-distortion functions.  Consequently, Fano's inequality can be  adapted in order to bound the average estimation error of functions, as shown in Corollary \ref{cor:PeMboundI}. By observing that a particular form of the bound stated in Theorem \ref{thm:Bound} is Schur-concave, we present in the next section a bound for the error probability of estimating functions in terms of  the maximal correlation, stated in Corollary \ref{cor:PeMboundrho}.

\subsection[A Convex Program for  Bounds on Estimation]{A Convex Program for Mapping Information Guarantees to Bounds on Estimation}
\label{sec:convexbounds}
Throughout the rest of the paper, we let $X$ and $Y$ be two random variables drawn from finite sets $\calX$ and $\calY$. We have the following definition.


\begin{defn}
\label{defn:informationMeasure}
We say that a function $\calI$ that maps any joint probability mass function  (pmf) to a non-negative real number is a \textit{dependence measure} (equivalently \textit{measure of dependence}) if  for any discrete random variables $W$, $X$, $Y$ and $Z$ (i) $\calI(\pxy)$ is convex in $\pygx$ for a fixed $\px$, (ii) $\calI$ satisfies the DPI, i.e. if $X\to Y\to Z$ then $\calI(p_{X,Z})\leq \calI(\pxy)$, and (iii) if $W$ is a one-to-one mapping of $Y$ and $Z$ is a one-to-one mapping of $X$, then $\calI(p_{W,Z})=\calI(\pxy)$ (invariance property). We overload the notation of $\calI$ and let $\calI(\pxy)=\calI(\px,\pygx)$ in order to make the dependence on the marginal distribution and the channel (transition probability) clear. Furthermore, we also denote $\calI(\pxy)=\calI(X;Y)$ when the distribution is clear from the context. Examples of dependence measures includes maximal correlation, defined in \eqref{eq:maxcorrdef}, and mutual information. 
\end{defn}

Now consider the standard estimation setup where a hidden variable $X$ should be estimated from an observed random variable $Y$. We assume that the joint distribution between $\pxy$ is not known, but the marginal distribution $\px$ is known, and that $\calI(\pxy)\leq \theta$ (e.g. security constraint) for a given dependence measure $\calI$. Since $\calI$ satisfies the DPI, for any estimate $\hat{X}$ of $X$ such that $X\to Y\to \Xh$ we have $\calI(X;\Xh)\leq \calI(X;Y)\leq \theta$. The problem of translating a bound on $\calI$ into a constraint on how well a hidden variable $X$ can (on average) be estimated from $Y$ given an error function $d:\calX\times\calX\to \Reals$  can be approximated by solving the optimization problem
\begin{align}
\inf_{p_{\hat{X}|X}}&~~\EE{d(X,\Xh)}\\
\sto&~~\calI(X;\Xh) \leq \theta.
\end{align}
This motivates the following definition.

\begin{defn}
\label{defn:distRate}
We denote the smallest (average) estimation error $D_{\calI,d}$ for a given dependence measure $\calI$ and estimation cost function $d:\calX\times\calX\to \Reals$ as
\begin{equation}
\label{eq:Ddefn}
D_{\calI,d}(p_X,\theta)\defined \inf_{p_{\hat{X}|X}} \left\{\EE{d(X,\Xh)} \middle|  \calI(p_X,p_{\Xh|X}) \leq \theta\right\},
\end{equation}
where the infimum is over all conditional distributions $p_{\hat{X}|X}$. 
\end{defn}

Observe that for any $p_{Y|X}$ that satisfies $ \calI(p_X,p_{Y|X}) \leq \theta$
\begin{align*}
D_{\calI,d}(p_X,\theta)\leq \inf_{p_{\hat{X}|Y}} \left\{\EE{d(X,\Xh)} \middle| X\to Y\to \Xh  \right\},
\end{align*}
since, by the assumption that $\calI$ satisfies the DPI, $\calI(X;\Xh)\leq \calI(X;Y)\leq \theta$.
When $\calI(X;Y)=I(X;Y)$, $D_{I,d}(\px,\theta)$ is the distortion-rate function \cite[pg. 306]{cover_elements_2006}. When the distortion function $d$ is the Hamming distortion, $D_{\calI,d}(p_X,\theta)$ gives the smallest probability of error for estimating $X$ given an observation $Y$ that satisfies $\calI(X;Y)\leq \theta$. This case will be of particular interest in this section, motivating the next definition.

\begin{defn}
\label{defn:ei}
Denoting the Hamming distortion metric as
\begin{equation*}
	d_H(x,y)\defined \begin{cases}
		0,&~x=y,\\
		1,&~\mbox{otherwise},
	\end{cases}
\end{equation*}
we define the \textit{error-rate function}\footnote{The term \textit{error-rate} function is used in the same sense as \textit{distortion-rate} function in rate distortion theory \cite[Chap. 10]{cover_elements_2006}. We adopt ``error'' instead of distortion here since we only consider Hamming distance as the distortion metric.}  for the dependence measure $\calI$ as \[e_{\calI}(p_X,\theta)\defined D_{\calI,d_H}(p_X,\theta).\] 
\end{defn}
The definition of error-rate function directly leads to the following simple lemma.

\begin{lem}
For  a given dependence measure $\calI$ and any fixed $\pxy$ such that $\calI(\pxy)\leq \theta$
\begin{equation*}
    P_e(X|Y)\geq e_\calI(\px,\theta).
\end{equation*}
\end{lem}
\begin{proof}
  Observe that $P_e(X|Y)=\min_{X\to Y\to\Xh} \EE{d_H(X,\Xh)}$, where the minimum is over all distributions $p_{\Xh|X}$ that satisfy the Markov constraint $X\to Y\to \Xh$. Since $\calI$ satisfies the DPI, then $\calI(X;\Xh)\leq \calI(X;Y)\leq \theta$, and the result follows from Definition \ref{defn:distRate}.
\end{proof}

The previous lemma shows that the characterization of $e_\calI(\px,\theta)$ for different measures of information $\calI$ is particularly relevant for applications in privacy and security, where $X$ is a variable that should remain hidden (e.g. plaintext) and $Y$ is an adversary's observation (e.g. ciphertext). Knowing $e_\calI$ allows us to translate an upper bound  $\calI(X;Y)\leq \theta$ into an estimation guarantee: regardless of an adversary's computational resources, given  only access to $Y$ he will not be able to estimate $X$ with an average error probability $P_e(X|Y)$ smaller than $e_\calI(\px,\theta)$. Therefore, by simply estimating $\theta$ and calculating $e_\calI(\px,\theta)$ we are able to evaluate the security threat incurred by an adversary that has access to $Y$.

\begin{example}[Error-rate function for mutual information.]
Using the expression for the rate-distortion function under Hamming
distortion for mutual information (\cite[(9.5.8)]{gallager_information_1968}), for $\calI(X;Y)=I(X;Y)$ and $\calX=[m]$, the error-rate function is given by $e_I(p_X,\theta)=d^*$, where $d^*$ is the solution of
    \begin{equation}
    \label{eq:rateDist}
      h_{b}(d^*)+d^*\log(m-1)=H(X)-\theta,
    \end{equation}
and $h_b(x)\defined -x\log x-(1-x)\log (1-x)$. Denoting $X\to Y\to \Xh$ and $p_e \defined P_e(X|Y)$, note that \eqref{eq:rateDist} implies Fano's inequality \cite[2.140]{cover_elements_2006}:
\begin{equation}
      h_b(p_e)+p_e\log(m-1)\geq H(X)-I(X;Y)=H(X|Y).
\end{equation}
\end{example}

\subsection{A Lower Bound for Error Probability Based on the PICs}
  \label{sec:fanoinertia}
Throughout the rest of the section, we assume without loss of generality that
$\px$ is sorted in decreasing order, i.e. $\px(1)\geq\px(2)\geq \dots\geq
\px(m)$.

\begin{defn}
  Let $\bLambda(\pxy)$ denote the vector of PICs of a joint
  distribution $\pxy$ sorted in decreasing order, i.e.
  $\bLambda(\pxy)=(\lambda_1(X;Y),\dots,\lambda_d(X;Y))$.  We denote $\bLambda(\pxy)\leq
  \blambdat\defined(\lambdat_1,\dots,\lambdat_d) $ if $\lambda_k(X;Y)\leq \lambdat_k$ for $k\in [d]$
  \begin{equation}
    \mathcal{R}(q,\blambdat)\triangleq\left\{\pxy \big| \px=q\mbox{ and } \bLambda(\pxy)\leq
  \blambdat   \right\}.
\end{equation}
  \end{defn}

In the next theorem we present a Fano-style bound for the estimation error
probability of $X$ that depends on the marginal distribution $\px$ and on the
principal inertias. 
\begin{thm}
  \label{thm:Bound}
 For $\blambda = (\lambda_1,\dots,\lambda_d)$ and fixed $\px$, let
  \begin{align}
    k^* \defined \max \left\{ k\in [m] ~\Big|~\px(k)\geq
     \sum_{i\in[m]}\px(i)^2\right\}. \label{eq:kstar}
   \end{align}
    In addition, let $\bp = (\px(1),\dots,\px(m))$ and $\blambda_{k^*} = \left(\lambda_1,\dots,\lambda_{k^*},\lambda_{k^*},\lambda_{k^*+1},\dots,\lambda_{m-1}\right)$ (where $\lambda_m\defined 0$ and $\blambda_{m}=(\lambda_1,\dots,\lambda_{m})$). Defining
\begin{equation*}
u(\px,\blambda) \defined \min_{0\leq\beta \leq \px(2)} \beta+\sqrt{\bp^T\blambda_{k^*}-\lambda_{k^*} \|\bp\|_2^2+\left\|\left[\bp-\beta \right]^+\right\|_2^2}~,
\end{equation*}
 then for any $(X,Y)\sim q_{X,Y}\in \mathcal{R}(\px,\blambda)$, 
\begin{equation}
   P_e(X|Y) \geq 1-  u(\px,\blambda).
   \label{eq:mainBound}
\end{equation}
\end{thm}

\begin{proof}
The proof of the theorem is presented in  Appendix \ref{app:proofsPIC_EP}. 
\end{proof}

\begin{remark}
  If  $\lambda_i=1 $ for all $1\leq i \leq d$,   \eqref{eq:mainBound} reduces to $P_e(X|Y) \geq 0$. Furthermore, if  $\lambda_i=0$ for all $1\leq i \leq d$,  \eqref{eq:mainBound} simplifies to $P_e(X|Y)\geq 1- \px(1)$.
\end{remark}

We now present a few direct but, as we shall show in the next section, useful corollaries of the result in Theorem
\ref{thm:Bound}. We note that a  bound with the same square-root order dependence on $\chi^2$-divergence as Eq. \eqref{eq:gunt} below has
appeared in the context of bounding the minmax decision risk in
\cite[Eq. (3.4)]{guntuboyina_minimax_2011}. However, the proof technique used in
\cite{guntuboyina_minimax_2011} does not seem to lead to the general bound presented in
Theorem \ref{thm:Bound}.

\begin{cor}
  If $X$ is uniformly distributed in $[m]$, then 
  \begin{equation}
    P_e(X|Y) \geq 1-\frac{1}{m}-\frac{\sqrt{(m-1)\chi^2(X;Y)}}{m}~.
    \label{eq:gunt}
  \end{equation}
  Furthermore, for $\rho_m(X;Y)=\sqrt{\lambda_1}$
    \begin{align*}
      P_e(X|Y) &\geq 1- \frac{1}{m} -\sqrt{\lambda_1}\left(1-\frac{1}{m}\right)\\
                &=1-\frac{1}{m}-\rho_m(X;Y)\left(1-\frac{1}{m}\right).
    \end{align*}
\end{cor}

\begin{cor}
  \label{cor:coolBounds}
    For any pair of variables $(X,Y)$ with marginal distribution in $X$ equal to $\px$ and maximal correlation (largest principal inertia)   $\rho_m(X;Y)^2= \lambda_1$, we have for all $\beta\geq 0$
    \begin{equation}
      P_e(X|Y) \geq 1-  \beta -
      \sqrt{\lambda_1\left(1-\sum_{i=1}^m\px(i)^2\right)+\sum_{i=1}^m\left(\left[\px(i)-\beta\right]^+\right)^2}~.
    \label{eq:coolBound1}
    \end{equation}
In particular, setting $\beta = \px(2)$,
     \begin{align}
      P_e(X|Y) &\geq  1-  \px(2) -
      \sqrt{\lambda_1\left(1-\sum_{i=1}^m\px(i)^2\right)+\left(\px(1)-\px(2)\right)^2}
      \label{eq:coolint} \\
      &\geq 1-  \px(1)-
      \rho_m(X;Y)\sqrt{\left(1-\sum_{i=1}^m\px(i)^2\right)}~,
      \label{eq:coolBound2}
    \end{align}
    where \eqref{eq:coolBound2} follows from \eqref{eq:coolint} being decreasing in $p_X(2)$.
\end{cor}

\begin{remark} 
  The bounds \eqref{eq:coolBound1} and  \eqref{eq:coolBound2}   are particularly helpful for showing how the error probability scales with the input distribution and the maximal correlation. For a given $\pxy$, recall that  \[\Adv(X|Y) \defined 1-\px(1)-P_e(X|Y),\] defined in \eqref{eq:advGuess}, is the advantage of correctly estimating $X$ from an observation of $Y$ over a random guess of $X$ when $Y$ is unknown.  Then, from equation \eqref{eq:coolBound2}
  \begin{align*}
        \Adv(X|Y) &\leq \rho_m(X;Y)\sqrt{\left(1-\sum_{i=1}^m\px(i)^2\right)} \\
        &\leq \rho_m(X;Y)= \sqrt{\lambda_1(X;Y)}.
  \end{align*}
  Therefore, the advantage of estimating $X$ from $Y$ decreases at least linearly with the maximal correlation between $X$ and $Y$.
\end{remark}

We present next results on the extremal properties of the error-rate function. This analysis will be particularly useful for determining how to bound the probability of error of estimating functions of a random variable.

\subsection[ Bounding the Estimation Error of Functions of a Hidden Random Variable]{Extremal Properties of the Error-Rate Function and Bounding the Estimation Error of Functions of a Hidden Random Variable}
\label{sec:BoundingFunctionsU}
%


Owing to convexity of $\calI(\px,p_{\Xh|X})$ in $p_{\Xh|X}$, it follows directly that $e_{\calI}(\px,\theta) $ is convex in $\theta$ for a fixed $\px$. 
%
%
%
We will now prove that, for a fixed $\theta$, $e_\calI(\px,\theta)$  is \textit{Schur-concave} in $\px$ if $\calI(\px,p_{\Xh|X})$ is concave in $p_X$ for a fixed $p_{\Xh|X}$. Ahlswede \cite[Theorem 2]{ahlswede_extremal_1990} proved this result for the particular case where $\calI(X;Y)=I(X;Y)$  by investigating the properties of the explicit characterization of the rate-distortion function under Hamming distortion. The proof presented here is  simpler and more general, and is based on a proof technique used by Ahlswede in \cite[Theorem 1]{ahlswede_extremal_1990}.

\begin{thm}
  \label{thm:schur}
If $\calI(\px,p_{\Xh|X})$ is concave in $\px$ for a fixed $p_{\Xh|X}$, then
$e_\calI(\px,\theta)$ is Schur-concave in $\px$ for a
fixed $\theta$.
\end{thm}

\begin{proof}
The proof is presented in Appendix \ref{app:proofsPIC_EP}.
\end{proof}

For a given integer $1\leq M\leq |\calX|$, we define
\begin{equation}
\label{eq:calFM}
  \calF_{M} \defined \left\{ f:\calX\rightarrow \calU~ \big|~f\mbox{ is
  surjective and }|\calU|\geq M
\right\}
\end{equation}
and
\begin{equation}
  P_{e,M}(X|Y) \defined \min_{f\in \calF_M} P_e(f(X)|Y).
  \label{eq:PeM}
\end{equation}
$P_{e,|\calX|}(X|Y)$ is simply the error probability of estimating $X$ from
$Y$, i.e.  $P_{e,|\calX|}(X|Y)=P_e(X|Y)$. The surjectivity condition in the definition of $\calF_M$ is mostly technical, and was added to (i) avoid the constant function being in $\calF_M$ (which would render the estimation error trivial) and (ii) enable the use of Schur-concavity results to derive bounds on estimation error. Note that, in the discrete setting considered here, by varying $M$ we span the set of all  functions of $X$, so there is no loss of generality. Nevertheless, there are practical settings where this condition naturally arises.  In classification problems, for example, the surjectivity condition would correspond to the number of classes used to classify $X$.

The next theorem shows that a lower bound for $P_{e,M}$ can be derived for any dependence measure $\calI$ as long as  $e_\calI(\px,\theta)$ or a lower bound for $e_\calI(\px,\theta)$ is Schur-concave in $p_X$.

\begin{thm}
  \label{thm:boundPeM}
  For a given $M\in [m]$ and $\px$ with $\calX=[m]$ and $p_X(1)\geq p_X(2)\geq\dots\geq p_X(m)$, let $U=g_M(X)$, where $g_M:\{1,\dots,m\}\rightarrow
  \{1,\dots,M\}$ is defined as
  \begin{equation*}
    g_M(x) \defined 
        \begin{cases}
            1& 1\leq x \leq m-M+1\\
            x-m+M & m-M+2\leq x \leq m~.
        \end{cases}
  \end{equation*}
  Let $p_U$ be the marginal distribution\footnote{The pmf of $U$ is
    $p_U(1)=\sum_{i=1}^{m-M+1} \px(i)$ and $p_U(k)=\px(m-M+k)$ for
  $k=2,\dots,M$.} of $U$. Assume that, for a given dependence measure
  $\calI$, there exists a function $L_{\calI}(\cdot,\cdot)$ such that for all
  distributions $q_X$ and any $\theta$, 
  $e_\calI(\qx,\theta)\geq
  L_\calI(\qx,\theta)$.  If $L_\calI(\px,\theta)$ is Schur-concave in $\px$, then for $X\sim \px$ and $\calI(X;Y)\leq \theta$,
  \begin{equation}
    P_{e,M}(X|Y)\geq L_\calI(p_U,\theta).
  \end{equation} 
In addition\footnote{We thank Dr. Nadia Fawaz (nadia.fawaz@gmail.com) for pointing out this extension.}, for any $S\to X\to Y$ such that $p_U$ majorizes $p_S$,
\begin{equation}
  P_e(S|Y)\geq L_\calI(p_U,\theta).
\end{equation}
\end{thm}
\begin{proof}
  The result follows from the following chain of inequalities:
  \begin{align*}
    P_{e,M}(X|Y) & \stackrel{(a)}{\geq}   \min_{f\in \calF_M,\ttheta}
    \left\{e_\calI\left(p_{f(X)},\ttheta \right): 
    \ttheta\leq \theta \right\}\\
    & \stackrel{(b)}{\geq} \min_{f\in \calF_M}
    \left\{e_\calI\left(p_{f(X)},\theta\right)\right\}\\
    &  \stackrel{(c)}{\geq}  \min_{f\in \calF_M}
    \left\{L_\calI\left(p_{f(X)},\theta\right)\right\}\\
    & \stackrel{(d)}{\geq}  L_\calI(p_U,\theta),
  \end{align*}
  where (a) follows from the DPI, (b) follows from  $e_\calI(\qx,\theta)$, being decreasing in $\theta$, $(c)$ follows from
  $e_\calI(\qx,\theta)\geq L_\calI(\qx,\theta)$ for all $\qx$, and $\theta$ and (d)
  follows from the Schur-concavity of the lower bound and by observing that  $p_U$
  majorizes $p_{f(X)}$ for every $f\in \calF_M$. In the case of $P_e(S|X)$, the same inequalities hold with $S$ playing the role of $f(X)$ in (a) and (b), and the last inequality also following from  Schur-concavity of $L_\calI(p_S,\theta)$ in $p_S$.
\end{proof}

\begin{remark}
The function $g_M(X)=U$ in Theorem  \ref{thm:boundPeM} is formed by adding the most likely symbols of $X$, and, consequently, $p_U$ majorizes any other distribution $p_{f(X)}$ for $f\in \calF_M$. The function $g_M$ can thus be regarded as the ``least uncertain'' function of $X$ in $\calF_M$ in the following sense: since R\'enyi entropy\footnote{The R\'enyi entropy of a discrete random variable $X$ is given by $H_\alpha(X)\defined \frac{1}{1-\alpha}\log\left(\sum_{x\in\calX} \px(x)^\alpha\right)$.} is Schur-concave, $H_\alpha(g_M(X))\leq H_\alpha(f(X))$ for all $f\in \calF_M$ and $\alpha\geq 0$.
\end{remark}

The following results illustrates how Theorem \ref{thm:boundPeM} can be
used for mutual information and maximal correlation.
\begin{cor}
  \label{cor:PeMboundI}
    Let $I(X;Y)\leq \theta$. Then
    \begin{equation*}
        P_{e,M}(X|Y)\geq d^* 
    \end{equation*}
    where $d^*$ is the solution of
    \begin{equation*}
      h_{b}(d^*)+d^*\log(m-1)=\min\{H(U)-\theta,0 \},
    \end{equation*}
    and $h_b(\cdot)$ is the binary entropy function.
\end{cor}
\begin{proof}
Let $R_I(\px,\delta)\defined \min_{p_{\Xh|X}}\{I(X;\hat{X})|\mathbb{E}[d_H(X,\hat{X})]\leq \delta \}$ be  the well known rate-distortion function under Hamming
distortion. Then $R_I(\px,\delta)$ satisfies (\cite[(9.5.8)]{gallager_information_1968})
$R_I(\px,\delta)\geq H(X)-h_{b}(d^*)-d^*\log(m-1)$. The result follows from
Theorem \ref{thm:schur}, since mutual information is concave in $\px$.
\end{proof}







\begin{cor}
    \label{cor:PeMboundrho}
    Let $\calJ_1(X;Y)=\rho_m(X;Y)\leq \theta$. Then  
    \begin{equation*} 
      P_{e,M}(X|Y) \geq 1-p_U(1)- \theta\sqrt{\left(1-\sum_{i=1}^M p_U(i)^2\right)}~,
   \end{equation*}
   where $P_{e,M}(X|Y)$ is defined in \eqref{eq:PeM} and $U$ is defined as in Theorem \ref{thm:boundPeM}.
\end{cor}
\begin{proof}
  The proof follows directly from Theorems \ref{thm:convex}, \ref{lem:dataProc}
  and Corollary \ref{cor:coolBounds}, by noting that \eqref{eq:coolBound2}
  is Schur-concave in $\px$.
  \end{proof}

The previous result leads to the next theorem, which states that the probability of guessing \textit{any} function of a hidden variable $X$ from an observation $Y$ is upper bounded by the maximal correlation of $X$ and $Y$.

\begin{thm}
\label{thm:AdvPeM}
  Let $p_X$ be fixed, $|\calX| < \infty$ and $\calF_M$ be given in \eqref{eq:calFM}. Define (cf. \eqref{eq:advGuess})
  \begin{equation*}
    \Adv_M(X|Y)\defined \max\left\{ 1-\max_{k\in [M]} p_{f(X)}(k)-P_e(f(X)|Y) \suchthat f\in \calF_M\right\}.
  \end{equation*}
  Then
  \begin{equation}
    \Adv_M(X|Y)\leq \rho_m(X;Y)\sqrt{1-\frac{1}{M}}\leq \rho_m(X;Y).
  \end{equation}
\end{thm}

\begin{proof}
  For $f\in \calF_M$
    \begin{align*}
      \Adv(f(X)|Y) &\leq \rho_m(f(X);Y) \sqrt{1-\sum_{i\in[M]} p_{f(X)}(i)^2}\\
                  &\leq  \rho_m(X;Y) \sqrt{1-\sum_{i\in[M]} p_{f(X)}(i)^2}\\
                  &\leq \rho_m(X;Y) \sqrt{1-\frac{1}{M}},
    \end{align*}
    where the first inequality follows from \eqref{eq:coolBound2} and the definition \eqref{eq:advGuess}, the second inequality follows by combining Theorem  \ref{lem:dataProc} (DPI for the PICs) and the fact that $\lambda_1(f(X);X)\leq 1$, which leads to $\rho_m(f(X);Y)\leq \rho_m(X;Y)$, and the last inequality follows from the fact that $\sum_{i\in[M]} p_{f(X)}(i)^2$ is minimized when $p_{f(X)}$ is uniform. The result follows by maximizing over all $f\in \calF_M$.
\end{proof}

The results presented in this section  demonstrate that the PICs are a  useful information measure that can shed light on  fundamental limits of estimation. In particular, Theorem \ref{thm:AdvPeM} connects the largest PIC, namely the maximal correlation, with the probability of correctly guessing \textit{any} function of a hidden, discrete random variable. The PICs also provide a  characterization of the functions of a hidden variable that can (or cannot) be estimated with small mean-squared error (Theorem \ref{thm:PIC_Charac}).  In the next section, we explore applications of the PICs to privacy and security.

%% file: PIC_Priv.tex
\section[Applications of the PICs to Security and Privacy]{Applications of the PICs to Security and Privacy}
\label{chap:PIC_Priv}

In this section, we present a few applications of the principal inertia components to problems in security and privacy. We adopt the privacy against statistical inference framework presented in \cite{du2012privacy}. This setup, called the \textit{Privacy Funnel}, was introduced in \cite{Makhdoumi2014funnel}. Consider two  communicating parties, namely Alice and Bob. Alice's goal is to disclose to Bob information about a set of measurement points, represented by the random variable $X$. Alice discloses this information in order to receive some utility from Bob. Simultaneously, Alice wishes to limit the amount of information revealed about a private random variable $S$ that is dependent on $X$. For example, $X$ may represent Alice's movie ratings, released to Bob in order to receive movie recommendations, whereas $S$ may represent Alice's political preference or yearly income.  Bob will try to extract the maximum amount of information about $S$ from the data disclosed by Alice.

Instead of revealing $X$ directly to Bob, Alice releases a new random variable, denoted by $Y$. This random variable is produced from $X$ through a random mapping $p_{Y|X}$, called the \textit{privacy-assuring mapping}. We assume that $p_{S,X}$ is fixed and known by both Alice and Bob, and $S\rightarrow X \rightarrow Y$. Alice's goal is to find a mapping $p_{Y|X}$ that minimizes $I(S;Y)$, while guaranteeing that the information disclosed about $X$ is above a certain threshold $t$, i.e. $I(X;Y)\geq t$. We refer to the quantity $I(S;Y)$ as the \textit{disclosed private information}, and $I(X;Y)$ as the \textit{disclosed useful information}. As discussed in Section \ref{sec:generalPICs}, when $I(S;Y)=0$,  we say that \textit{perfect privacy} is achieved, i.e. $Y$ does not reveal any information about $S$.  We consider here the non-interactive, one-shot regime, where Alice discloses information once, and no additional information is released. We also assume that  Bob knows the privacy-assuring mapping $p_{Y|X}$ chosen by Alice, and no side information is available to Bob about $S$ besides $Y$.

\subsection{The Privacy Funnel}
\label{sec:funnel}

We define next the privacy funnel function, which captures the smallest amount of disclosed private information for a given threshold on the amount of disclosed useful information. We then characterize properties of the privacy funnel function in the rest of this section. 

\begin{defn}
  For $0\leq t \leq H(X)$ and a joint distribution $p_{S,X}$ over $\calS\times \calX$, we
  define the \textit{privacy funnel function} $G_I(t,p_{S,X})$ as
  \begin{equation}
    G_I(t,p_{S,X})\defined \inf\left\{ I(S;Y)\middle| I(X;Y)\geq t, S\rightarrow
    X\rightarrow Y \right\},
    \label{eq:defGI}
  \end{equation}
  where the infimum is over all mappings $p_{Y|X}$ such that $\calY$ is finite.
  For a fixed $p_{S,X}$ and $t\geq 0$, the set of  pairs $\{\left(t,
    G_I(t,p_{S,X})\right)\}$ is called the
  \textit{privacy region} of $p_{S,X}$.
\end{defn}
We now enunciate a few useful properties of $G_I(t,p_{S,X})$ and the privacy region.

\begin{lem}
  \label{lem:increase}
   \begin{align}
     G_I(t,p_{S,X})=\min_{p_{Y|X}}\left\{ I(S;Y)\middle|  I(X;Y)\geq t, S\to X\to
     Y,~|\calY|\leq |\calX|+1\right\}.
   \end{align}
In addition, for a fixed $p_{S,X}$, the mapping $t\mapsto \frac{G_I(t,p_{S,X})}{t}$ is  non-decreasing, and $G_I(t,p_{S,X})$ is convex in $t$.
\end{lem}
\begin{proof}
The proof is in Appendix \ref{app:PICpriv}.
\end{proof}

\begin{lem}
  \label{lem:region}
  For $0\leq t \leq H(X)$,
    \begin{equation}
      \max\{t-H(X|S),0\}\leq G_I(t,p_{S,X})\leq \frac{tI(X;S)}{H(X)}.
      \label{eq:Bounded_Region}
    \end{equation}
\end{lem}
\begin{proof}
Observe that $G_I(H(X),p_{S,X})=I(X;S)$, since $I(X;Y)=H(X)$ implies that
$p_{Y|X}$ is a one-to-one mapping of $X$. The upper bound then
follows directly from \eqref{eq:Gupperb}. 

Clearly $G_I(t,p_{S,X})\geq 0$. In addition, for any $p_{Y|X}$,
\begin{align*}
    I(S;Y)&=I(X;Y)-I(X;Y|S)\\
          &\geq I(X;Y) - H(X|S)\\
          &\geq t - H(X|S),
\end{align*}
proving the lower bound. 
\end{proof}

\begin{figure}[tb]
  \centering
  \psfrag{A}[r][c]{\footnotesize $I(S;X)$}
  \psfrag{B}[cc]{\footnotesize $H(X|S)$}
  \psfrag{C}[cc]{\footnotesize $H(X)$}
  \psfrag{D}[bc]{ }
  \psfrag{X}[b][B]{$t$}
  \psfrag{Y}[c][c]{$G_I(t,p_{S,X})$}
  \includegraphics[width=0.5\textwidth]{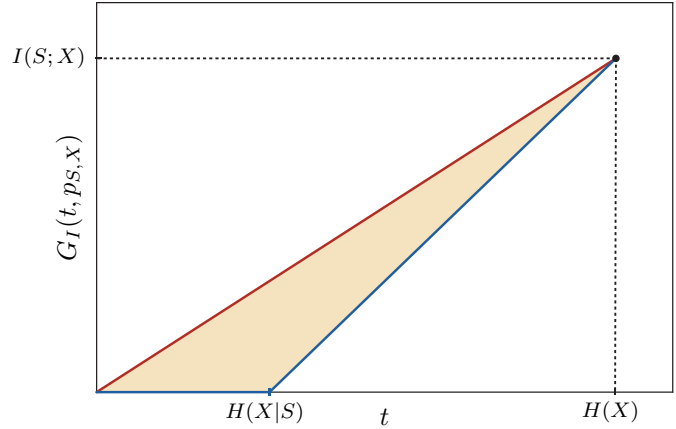}
  \caption[The Privacy Funnel.]{For a fixed $p_{S,X}$, the privacy region is contained within the
shaded area. The red and the blue lines correspond, respectively, to the upper
and lower bounds presented in Lemma \ref{lem:region}.}
  \label{fig:funnel}
\end{figure}

Figure \ref{fig:funnel} illustrates the bounds from Lemma \ref{lem:region}. The
privacy region is contained withing the shaded area. The next two examples
illustrate that both the upper bound (red line) and the
lower bound (blue line) of the privacy region can be achieved for particular
instances of $p_{S,X}$.

\begin{example}
   Let $X=(S,W)$, where $W\independent S$. Then by setting $Y=W$, we have
   $I(S;Y)=0$ and $I(X;Y)=H(W)=H(X|S)$. Consequently, from Lemmas
   \ref{lem:increase} and \ref{lem:region}, $G_I(t,p_{S,X})=0$ for
   $t\in[0,H(X|S)]$. By letting $Y=W$ w.p. $\lambda$ and $Y=(S,W)$ w.p.
   $1-\lambda$ for $\lambda\in [0,1]$, the lower-bound $G_I(t,p_{S,X})=t-H(X|S)$ can be
   achieved for $H(X|S)=H(W)\leq t\leq H(X)$. Consequently, the lower bound in
   \eqref{eq:Bounded_Region} is sharp.
\end{example}

\begin{example}
    Now let $X=f(S)$. Then $I(X;S)=H(X)$ and
    \begin{align*}
        I(S;Y)=I(X;Y)-I(X;Y|S)=I(X;Y).
    \end{align*}
    Consequently, $G_I(t,p_{S,X})=t$, and the upper bound in
    \eqref{eq:Bounded_Region} is sharp.
\end{example}   

\subsection{The Optimal Privacy-Utility Coefficient and the Smallest PIC}
\label{sec:PIC}

We now study the smallest possible ratio between disclosed private and
useful information, defined next.
\begin{defn}
  The \textit{optimal privacy-utility coefficient} for a given distribution
  $p_{S,X}$ is given by
  \begin{equation}
    \vs \defined \inf_{p_{Y|X}} \frac{I(S;Y)}{I(X;Y)}.
  \end{equation}
\end{defn}
It follows directly from Lemma \ref{lem:increase} that
\begin{equation}
    \vs = \lim_{t\to 0}  \frac{G_I(t,p_{S,X})}{t}.
\end{equation}

We show in Section \ref{sec:perfect} that the value of $\vs$ is  related to the
smallest PIC of $p_{S,X}$ (i.e. the smallest eigenvalue
of the spectrum of the conditional expectation operator, defined below).
We also prove that $\vs=0$ is a necessary and sufficient condition for
achieving perfect privacy while disclosing a non-trivial amount of useful
information. Before introducing these results, we present an alternative
characterization of $\vs$ (Lemma \ref{thm:var}), and introduce a measure based on the smallest PIC  (Definition \ref{defn:PIC}) and an auxiliary  result
(Lemma \ref{lem:delta}).
\begin{remark}
    The proofs of Lemma \ref{thm:var} and Lemma \ref{thm:Ineq1} in this section are closely related
    to \cite{anantharam_maximal_2013}. We
    acknowledge that their proof techniques inspired some of the results presented here.
\end{remark}

The next result provides a characterization of the optimal privacy-utility coefficient.

\begin{lem}
  \label{thm:var}
  Let $q_S$ denote the distribution of $S$ when $p_{S|X}$ is fixed and $X\sim
  q_X$. Then 
  \begin{equation}
    \vs =  \inf_{q_X\neq p_X} \frac{D(q_S||p_S)}{D(q_X||p_X)}.
    \label{eq:divergence}
  \end{equation}
\end{lem}
\begin{proof}
The proof is in Appendix \ref{app:PICpriv}.
\end{proof}



The smallest PIC is of particular interest
for privacy, and upper bounds the value of $\vs$.  In particular, we will be interested in the coefficient $\delta(p_{S,X})$, defined bellow

\begin{defn}
  \label{defn:PIC}
  Let  $d\defined \min\{|\calS|,|\calX|\}-1$,  and $\lambda_d(S;X)$ the smallest PIC of $p_{S,X}$. We define
  \begin{equation}
    \dps\defined 
        \begin{cases} 
          \lambda_d(S;X)&\mbox{ if } |\calX|\leq |\calS|,\\
          0&\mbox{ otherwise.}
          \end{cases}
  \end{equation}
  \end{defn}
The following lemma provides a useful characterization of $\dps$, related to the interpretation of the PICs as the spectrum of the conditional expectation operator given in Theorem \ref{thm:PIC_Charac}. This result is a direct consequence of Theorem \ref{thm:PIC_Charac}, and we present a self-contained proof in Appendix \ref{app:PICpriv}.
\begin{lem}
  \label{lem:delta}
For a given $p_{S,X}$,
  \begin{align} 
    \dps = \min\left\{ \|\EE{f(X)|S}\|_2^2\middle| ~ f:\calX\to \Reals,
    \EE{f(X)}=0, \|f(X)\|_2=1 \right\}.
\end{align}    
\end{lem}

\subsection{Information Disclosure with Perfect Privacy}
\label{sec:perfect}

If $\vs=0$, then it may be possible to disclose some information about $X$
without revealing any information about $S$. However, since $G_I(0,p_{X,S})=0$,
it is not immediately clear that $\vs=0$ implies that there exists $t$ strictly
bounded away from 0 such $G_I(t,p_{X,S})=0$.  This would represent the ideal
privacy setting, since, from Lemma \ref{lem:increase}, there would exist a
privacy-assuring mapping that allows the disclosure of some non-negligible
amount of useful information while guaranteeing $I(S;Y)=0$. This, in turn, would mean that
perfect privacy is achievable with non-negligible utility \textit{regardless of
the specific privacy metric used}, since $S$ and $Y$ would be independent.  

In this section, we prove that if the optimal privacy-utility coefficient is 0,
then there indeed exists a privacy-assuring mapping that allows the disclosure
of a non-trivial amount of useful information while guaranteeing perfect
privacy. We also show that the value of $\dps$ is closely related to $\vs$.
This relationship is analogous to the one between the hypercontractivity
coefficient $s^*$, defined in \cite{ahlswede_spreading_1976} and \cite{kamath2012non}, and the maximal
correlation $\rho_m$. In particular, as shown in the next two lemmas, $\vs
\leq \dps$ and $\vs=0 \iff \dps = 0$.  
\begin{lem}
  \label{thm:Ineq1}
  For any $p_{S,X}$ with finite support $\calS\times \calX$,
    \begin{equation}
        \vs \leq \dps.
    \end{equation}
    and 
    \begin{equation}
\inf_{p_X} \vs =
\inf_{p_X} \dps. 
    \end{equation}
\end{lem}
\begin{proof}
The proof is in Appendix \ref{app:PICpriv}.
\end{proof}
The next theorem proves that $\dps$ can serve as a proxy for perfect privacy,
since the optimal privacy-utility coefficient is 0 if and only if $\dps$ is also
0.

\begin{lem}
  \label{thm:iff}
  Let $p_{S,X}$ be such that $H(X)>0$ and $\calS$ and $\calX$ are finite.
  Then\footnote{If $S$ is binary, then \eqref{eq:iff} implies
that perfect privacy is achievable iff $S$ and $X$ are independent (since
$\dps=\rho_m(S;X)^2$), recovering
\cite[Thm. 2]{asoodeh_notes_2014}.} 
    \begin{equation}
      \label{eq:iff}
        \vs = 0 \iff \dps = 0.
    \end{equation}
\end{lem}
\begin{proof}
The proof can be found in Appendix \ref{app:PICpriv}.
\end{proof}

We are now ready to prove that a non-trivial amount of useful information can be disclosed
without revealing any private information if and only if  $\vs=0$ (or equivalently,
$\dps=0$). This result follows naturally from Theorem \ref{thm:iff}, since
$\vs=0$ implies that $\dps=0$, which means that the matrix $\bQ$ and,
consequently, $\bP_{S|X}$, is either not full rank or has more columns than rows
(i.e. $|\calX|>|\calS|$). This, in turn, can be exploited in order to find a
mapping $p_{Y|X}$ such that $Y$ reveals some information about $X$, but no
information about $S$. This argument is made precise in the next theorem.

\begin{remark}
When  $\bP_{S|X}$ is not full rank or has more columns than rows, then
$S$ and $X$ are weakly independent. As shown in \cite[Thm.
4]{berger_multiterminal_1989} and \cite{asoodeh_notes_2014}, this implies that a privacy-assuring mapping that
achieves perfect privacy while disclosing a non-trivial amount of useful
information can be found. Theorem \ref{thm:ntrivial} recovers this result in terms of the
smallest PIC, and Corollary \ref{cor:explicit} provides
an  estimate of the amount of useful information that can be revealed
with perfect privacy. \end{remark}

\begin{thm}
  \label{thm:ntrivial}
  For a given $p_{S,X}$, there exists a privacy-assuring mapping $p_{Y|X}$
  such that $S\to X\to Y$, $I(X;Y)>0$ and $I(S;Y)=0$ if and only if $\dps=0$ (equivalently
  $\vs=0$). In particular,
  \begin{equation}
    \exists t_0>0 :G_I(t_0,p_{S,X})=0 \iff \dps=0.
    \label{eq:ntrivial}
  \end{equation}
\end{thm}
\begin{proof}
   The direct part of the theorem follows directly from the definition of $\vs$
   and Lemma \ref{thm:iff}. Assume that $\dps=0$. Then, from Lemma
   \ref{lem:delta}, there exists  $f:\calX\to \Reals$ such that
   $\|f(X)\|_2=1$, $\EE{f(X)=0}$, and $\|\EE{f(X)|S}\|_2=0$. Consequently,
   $\EE{f(X)|S=s}=0$ for all $s\in \calS$. 

   Fix $\calY=[2]$, and, for $\epsilon>0$ and $\epsilon$ appropriately small,
   \begin{equation*}
     p_{Y|X}(y|x)= 
       \begin{cases}
         \frac{1}{2}-\epsilon f(x),&y=1,\\
         \frac{1}{2}+\epsilon f(x),&y=2.  
       \end{cases}
   \end{equation*}
   Note that it is sufficient to choose
   $\epsilon=(2\max_{x\in\calX}|f(X)|)^{-1}$, so $\epsilon$ is strictly bounded
   away from 0. In addition, $p_Y(1)=1/2$. Therefore, 
   \begin{equation}
     \label{eq:t0}
   I(X;Y)=1-\sum_{x\in\calX} p_X(x)h_b\left(
     \frac{1}{2}+\epsilon f(x) \right)>0.
   \end{equation}
     Since $S\to X\to Y$,
     \begin{align*}
       p_{Y|S}(y|s)&=\sum_{x\in \calX} p_{Y|X}(y|x)p_{X|S}(x|s)\\
                   &= \sum_{x\in \calX} \left(\frac{1}{2}+(-1)^y\epsilon
                   f(x)\right) p_{X|S}(x|s)\\
                   &= 1/2 + (-1)^y\epsilon \EE{f(X)|S=s} \\
                   &=1/2,
     \end{align*}
     and, consequently, $S$ and $Y$ are independent. Then $I(S;Y)=0$, and the
     result follows.
\end{proof}

The previous result proves that if either $|\calX|>|\calS|$ or the smallest
principal inertia component  of $p_{S,X}$ is 0 (i.e. $\dps=0$), then it is
possible to achieve perfect privacy while disclosing some useful information.
In particular, the value of $t_0$ in \eqref{eq:Hess} is
lower-bounded by the expression in \eqref{eq:t0}.  We  note that this result
would not necessarily hold if $\calS$ and $\calX$ are not  finite sets.

Since $I(S;Y)=0$ implies that $S$ and $Y$ are independent, Theorem \ref{thm:ntrivial} holds not only for mutual information, but also for \textit{any} dependence measure $\calI$, defined in  Definition \ref{defn:informationMeasure}, that satisfies $\calI(X;Y)=0$ if and only if $X$ and $Y$ are independent. 
This leads to the following result.

\begin{cor}
  Let $p_{S,X}$ be given, and $\calI$ be a non-negative dependence measure (e.g. total variation or
 maximal correlation, cf. Definition \ref{defn:informationMeasure}) such that for any two random variable $A$ and $B$, 
$\calI(A;B)=0 \iff A\independent B$. Then there exists $p_{Y|X}$ such that
 $S\to X\to Y$, $\calI(X;Y)>0$ and $\calI(S;Y)=0$ if and only if $\dps=0$ .
 \end{cor}
 \begin{proof}
    This is a direct consequence of Theorem  \ref{thm:ntrivial}, since $\calI(X;Y)>0 \iff I(X;Y)>0$ and $\calI(S;Y)=0 \iff I(S;Y)=0$.
 \end{proof}

 \begin{remark}
As long as privacy is measured in terms of
statistical dependence (with perfect privacy implying statistical independence)
and some utility can be derived when $Y$ is not independent of $X$, then
$\dps$ fully characterizes when perfect privacy is achievable with non-trivial
utility. 
\end{remark}


We present next an explicit lower bound for the largest amount of useful
information that can be disclosed  while guaranteeing perfect privacy. The
result follows directly from the construction used in the proof of Theorem
\ref{thm:ntrivial}, and is presented in Appendix \ref{app:PICpriv}.

\begin{cor}
  \label{cor:explicit}
  For fixed $p_{S,X}$, let 
  \begin{align*}
    \calF_0\defined \left\{ f:\calX\to
    \Reals\middle|\EE{f(X)}=0,~\|f(X)\|_2= 1,~\|\EE{f(X)|S}\|_2=0 \right\}\cup w_0,
\end{align*}
  where $w_0$ is the trivial
  function that maps $\calX$ to $\{0\}$. Then $G_I(t,p_{S,X})=0$ for $t\in[0,t^*]$,
  where
  \begin{equation}
    \label{eq:tstar}
    t^*\geq1- \max_{f\in\calF_0} \EE{h_b \left(\frac{1}{2}+  \frac{f(X)}{2
    \|f\|_\infty}
  \right)}.
  \end{equation}
  Furthermore, the lower bound for $t^*$ is sharp when $\dps=0$, i.e. there
  exists a $p_{S,X}$ such that $t^*>0$ and $G_I(t,p_{S,X})=0$ if and
  only if $t\in [0,t^*]$.
\end{cor}


The previous bound for $t^*$ can be loose, especially if  $|\calX|$ is large.
In addition,  the right-hand side of \eqref{eq:tstar} can be made arbitrarily
small by decreasing $\min_{x\in \calX} p_X(x)$.  Nevertheless, \eqref{eq:tstar}
is an explicit bound on the amount of useful information that can be disclosed
with perfect privacy.

When $S^n=(S_1,\dots,S_n)$ and $X^n=(X_1,\dots,X_n)$, where
$(S_i,X_i)\sim p_{S,X}$ are i.i.d. random variables, the next proposition states
that $\delta(p_{S^n,X^n})=\dps^n$. Consequently, as long as $\dps<1$, it is
possible to disclose a non-trivial amount of useful information while disclosing
an arbitrarily small amount of private information by making $n$ sufficiently
large. Loosely speaking, this is similar to hiding a needle in a haystack: As the number of available samples of $S$ and $X$ increases, we can use the additional randomness to better hide the private variables $S_i$.
\begin{prop}
   Let $S^n=(S_1,\dots,S_n)$ and $X^n=(X_1,\dots,X_n)$, where
$(S_i,X_i)\sim p_{S,X}$ are i.i.d. random variables. Then
\begin{equation}
    v^*(p_{S^n,X^n})\leq \delta(p_{S^n,X^n})=\dps^n.
\end{equation}
\end{prop}
\begin{proof}
    The result is a direct consequence of the tensorization property of the principal inertia components, presented in Lemma \ref{lem:tensor}.
    \end{proof}

%% file: finalremarks.tex
\section{Final Remarks}
\label{sec:finalremarks}

The PICs are powerful information-theoretic metrics that provide both (i) a measure of dependence between two random variables $X$ and $Y$, and (ii) a complete characterization of which functions of $X$ can be reliably estimated (in terms of mean-squared error) given an observation of $Y$. As shown here, the PICs play can be used for deriving  bounds on one-bit functions of a channel input given a channel output. Furthermore, in privacy applications, we proved that perfect privacy can be achieved if and only if the smallest PIC is zero.
The PICs were also used to derive bounds on estimation error probability. In particular, we demonstrated that the largest PIC (equivalently, the maximal correlation $\rho_m(X;Y)$) plays a key role in estimation: \[\Adv(f(X)|Y)\leq \rho_m(X;Y), \]i.e. the advantage over a random guess of estimating any function of $X$ given $Y$ is at most $\rho_m(X;Y)$.

Information theoretic security and privacy applications provide fertile ground for the use of PICs, specially when privacy is measured in terms of how well an adversary can estimate a secret (private) variable. The principal functions (cf. Definition \ref{def:PIC_4}) provide a basis  for the  finite-variance functions of a random variable, and the PICs measure the MMSE of estimating each of these functions. Consequently, the PICs provide a  characterization of which functions of $X$ can or cannot be inferred reliably (in terms of MMSE) from an observation of $Y$. This property can be used in privacy applications: For example, in order to quantify how well an adversary can estimate a private function $S=f(X)$ given a disclosed variable $Y$, it is sufficient to express $f(X)$ in terms of the principal functions of $p_{X,Y}$. The adversary's ability of correctly estimating $f(X)$ is then entirely determined by the PICs of $p_{X,Y}$.

More precisely, for $f:\calX\rightarrow \Reals$, the mean squared-error $f(X)$ given $Y$ can be expressed as
\begin{align}
  \mmse(f(X)|Y) &= \EE{f(X)^2-\EE{f(X)|Y}^2} \nonumber\\
  &= \normEuc{f(X)}^2\left( 1-\frac{\normEuc{\EE{f(X)|Y}}^2}{
    \normEuc{f(X)}^2} \right), \label{eq:MMSE}
\end{align}
Consequently, the MMSE depends on the spectrum of the conditional expectation operator $(T_Yf)(y)\defined \EE{f(X)|Y=y}$ which, in turn, is composed by the principal inertia components (cf. Theorem \ref{thm:PIC_Charac}). When   $\EE{f(X)}=0$ and $\EE{f(X)^2}=1$, one can determine functions $f_1,f_2,\dots$ as in Theorem \ref{thm:PIC_Charac}, such that $f_i$ is given by
\begin{align*}
  f_i = \argmax \left\{ \normEuc{\EE{f(X)|Y}}^2\mid \right. &
    f:\calX\rightarrow\Reals,~\EE{f(X)}=0,~\EE{f(X)^2}=1,\\
    &\left.\EE{f(X)f_j(X)}=0 \mbox{ for } 1\leq
  j\leq i-1 \right\}.
\end{align*}
Then $$\|\EE{f_i(X)|Y}\|_2^2=\lambda_i(X;Y).$$ 

It follows directly that, for any non-trivial function $f:\calX\rightarrow
\Reals$ with $\EE{f(X)}=0$,
\begin{align}
    \mmse(f(X)|Y) \geq  \normEuc{f(X)}^2\left(1-\rho_m(X;Y)^2\right),
\end{align}
with equality if $f(X)=cf_1(X)$, where $c=\normEuc{f(X)}$. Therefore, for a
fixed  variance $c$, $cf_1(X)$ is the function of $X$ that can be most reliably
estimated (in terms of mean-squared error) from all possible mappings
$\calX\rightarrow \Reals$. Furthermore,
\begin{align}
\label{eq:mmsedecomp}
  \mmse(f(X)|Y) =  \normEuc{f(X)}^2\left(1-\sum_{i} c_i^2\lambda_i(X;Y)\right),
\end{align}
where $c_i\defined \EE{f(X)f_i(X)}/\normEuc{f(X)}$ and $\sum_i{c_i^2}=1$. Consequently,
functions that are closely ``aligned'' with $f_i$ for small $i$ cannot be inferred with small mean squared-error.

In privacy applications with estimation constraints, this result sheds light on the nature of the fundamental tradeoff between privacy and utility. If $X$ and $Y$ correspond, respectively, to the input and output of a privacy-assuring mapping, then the PICs and corresponding principal functions of $p_{X,Y}$ will determine which functions (features) of $X$ remain private. If, for example, the principal functions corresponding to small PICs also span functions of $X$ that should be reliably estimated from $Y$ for utility purposes, then the privacy-assuring mapping $p_{Y|X}$ will provide an unfavorable tradeoff between privacy an utility. 

As another example, assume that we wish to design a privacy-assuring mapping where the secret $S=(h_1(X),\dots,h_t(X))$ is composed by a certain set of functions (features) $h_1,\dots,h_t$ of $X$ that are supposed to remain private. The privacy-assuring mapping $\pygx$ should then assure that the principal functions that span the subspace formed by $(h_1(X),\dots,h_t(X))$ must have small PICs.  These examples, together with the results presented here, motivate the future use of PICs to  drive the design of privacy-assuring mappings that achieve a favorable tradeoff between privacy and utility.

%% file: appPIC.tex

\section{Proofs from Section \ref{chap:PICs}}
\label{app:proofsPIC}

\subsection*{Lemma \ref{lem:DPI_MMSE}}
\begin{proof}
Let $f\in \calL_2(p_X)$, $\EE{f(X)}=0$ and $g\in \calL_2(Z)$, $\EE{g(Z)}=0$, $\|g(Z)\|_2=1$. Then
\begin{align*}
  \EE{f(X)g(Z)} &=\EE{ \EE{f(X)g(Z)|Y} }\\
                &\stackrel{\mathrm{(a)}}{=}\EE{ \EE{f(X)|Y}\EE{g(Z)|Y} }\\
                &\stackrel{\mathrm{(b)}}{\leq} \| \EE{f(X)|Y}\|_2 \| \EE{g(Z)|Y} \|_2\\
                &\stackrel{\mathrm{(c)}}{\leq} \sqrt{\lambda_1(Z;Y)}\| \EE{f(X)|Y}\|_2,
\end{align*}
where (a) follows from  the assumption that $X\to Y\to Z$, (b) follows from the Cauchy-Schwarz inequality, and (c) follows from characterization (3) in Theorem \ref{thm:PIC_Charac}. By choosing $g(z)=\EE{f(X)|Z=z}/\|\EE{f(X)|Z}\|_2$ and using the last inequality, we have
\begin{align*}
\EE{f(X)g(Z)}  = \EE{\EE{f(X)|Z}g(Z)}= \|\EE{f(X)|Z}\|_2\leq \sqrt{\lambda_1(Z;Y)}\| \EE{f(X)|Y}\|_2.
\end{align*}
Squaring both sides, we arrive at \eqref{eq:MMSE_DPI}.
\end{proof}

\section{Proofs from Section \ref{chap:PIC_IT}}
\label{app:proofsPIC_IT}

\subsection*{Lemma \ref{thm:AB}}
 \begin{proof}
     Let $Y^n = X^n \oplus Z^n$ for some
     $Z^n$ distributed over $\{-1,1\}^n$ and independent of $X^n$. Thus
        \begin{align*}
          \EE{\chi_\calS(Y^n)|X^n}&=\EE{\chi_\calS(Z^n\oplus X^n) \mid X^n}\\
          &=\EE{\chi_\calS(X^n)\chi_\calS(Z^n) \mid X^n}\\
          &=\chi_\calS(X^n) \EE{\chi_\calS(Z^n)},
        \end{align*}
     where the last equality follows from the assumption that $X^n\independent Z^n$. By
     letting $c_\calS =  \EE{\chi_\calS(Z^n)}$, it follows that $p_{Y^n|X^n}\in
     \calA_n$ and, consequently, $\mathcal{B}_n\subseteq \calA_n$.

     Now let $y_n$ be fixed and $\delta_{y^n}:\{-1,1\}^n\rightarrow \{0,1\}$ be given by 
        \begin{align*}
          \delta_{y^n}(x^n)=
            \begin{cases}
                1,& x^n=y^n,\\
                0,& \mbox{otherwise.}
            \end{cases}
        \end{align*}
        Since the function $\delta_{y^n}$ has Boolean inputs, it  can be expressed
        in terms of its Fourier expansion  \cite[Prop. 1.1]{odonnell_topics_2008} as
       \begin{equation*}
         \delta_{y^n}(x^n) = \sum_{\calS\subseteq [n]} \widehat{d}_\calS
         \chi_\calS(x^n)
       \end{equation*}
       for some set of coefficients $\widehat{d}_\calS\in \Reals$, $\calS\subseteq[n]$. Now let $p_{Y^n|X^n}\in \calA_n$. Observe that $p_{Y^n|X^n}(y^n|x^n) =
       \EE{\delta_{y^n}(Y^n)\mid X^n=x^n}$ and, for $z^n\in \{-1,1\}^n$,
     \begin{align*}
       p_{Y^n|X^n}(y^n\oplus z^n|x^n\oplus z^n) &= \EE{\delta_{y^n\oplus z^n}(Y^n)\mid
     X^n=x^n\oplus z^n}\\
     & =\EE{\delta_{y^n}(Y^n\oplus z^n)\mid X^n=x^n\oplus z^n}\\
     & = \EE{\sum_{\calS\subseteq [n]} \widehat{d}_\calS
         \chi_\calS(Y^n\oplus z^n)\mid X^n=x^n\oplus z^n}\\
     &=  \EE{\sum_{\calS\subseteq [n]} \widehat{d}_\calS
         \chi_\calS(Y^n)\chi_\calS(z^n)\mid X^n=x^n\oplus z^n}\\
         &\stackrel{(a)}{=}  \sum_{\calS\subseteq [n]} c_\calS \widehat{d}_\calS
         \chi_\calS(x^n\oplus z^n)\chi_\calS(z^n)\\
     &= \sum_{\calS\subseteq [n]} c_\calS \widehat{d}_\calS
         \chi_\calS(x^n)\\
         &\stackrel{(b)}{=} \EE{\sum_{\calS\subseteq [n]} \widehat{d}_\calS
         \chi_\calS(Y^n)|X^n=x^n}\\
         &= \EE{\delta_{y^n}(Y^n)\mid X^n=x^n}\\
         &=p_{Y^n|X^n}(y^n|x^n).
     \end{align*}
    Equalities $(a)$ and $(b)$ follow from the definition of $\calA_n$.
    By defining the distribution of $Z^n$ as
    $p_{Z^n}(z^n)\defined p_{Y^n|X^n}(z^n|\ones^n)$, where $\ones^n$ is the vector with
    all entries equal to 1, it follows that $Z^n=X^n\oplus Y^n$,
    $Z^n\independent X^n$ and
    $p_{Y^n|X^n}\subseteq \mathcal{B}_n$.

  \end{proof}

\subsection*{Lemma \ref{lem:zupper}}
\begin{proof}
    Let $\xb \in \calC^m(a,\bP^T)$ and $\yb \in \calC^n(b,\bP)$. Then, for $\bP$
     decomposed as $\bP=\Dx^{1/2}\bQ\Dy^{1/2}$ where $\bQ$ given in \eqref{eq:Qdefn} and denoting
    $\bSigma^{-}=\diag{0,\sigma_1,\dots,\sigma_d}$,
    \begin{align} 
      \xb^T\bP\yb & = ab + \xb^T\bD_X^{1/2}\bU\bSigma^{-}\bV^T\bD_Y^{1/2}\yb
      \nonumber\\
      & = ab+\hat{\xb}^T\bSigma^{-}\hat{\yb}, \label{eq:zExpand}            
    \end{align}
    where $\hat{\xb}\defined \bU^T\bD_X^{1/2}\xb$ and $\hat{\yb}\defined
    \bV^T\bD_Y^{1/2}\yb$. Since  $\hat{x}_1=\normEuc{\hat{\xb}}=a$ and
    $\hat{y}_1=\normEuc{\hat{\yb}}=b$, then
    \begin{align*}
      \hat{\xb}^T\bSigma^{-}\hat{\yb} &= \sum_{i=2}^{d+1}\sigma_{i-1}
      \hat{\xb}_i\hat{\yb}_i\\
      &\leq
      \sigma_1\sqrt{\left(\normEuc{\hat{\xb}}^2-\hat{\xb}_1^2
      \right)\left(\normEuc{\hat{\yb}}^2-\hat{\yb}_1^2\right)}\\
      &= \sigma_1\sqrt{(a-a^2)(b-b^2)}.
    \end{align*}
    The result  follows by noting that $\sigma_1=\rho_m(X;Y)$.
\end{proof}

\section{Proofs from Section \ref{sec:boundEP}}
\label{app:proofsPIC_EP}

\subsection*{Theorem \ref{thm:Bound}}
Consider the matrix $\bQ=\bU\bSigma\bV^T$ given in \eqref{eq:Qdefn}, and define
\begin{equation*}
\At \defined \Dx^{1/2}\bU,~\Bt \defined \Dy^{1/2}\bV.
\end{equation*}
Then 
\begin{equation}
 \bP = \At \bSigma \Bt^T,
 \label{eq:comactDecompPxy}
\end{equation} 
where $\At^T \Dx^{-1} \At=\Bt^T \Dy^{-1} \Bt = \bI$. 

It follows directly from Theorem \ref{thm:PIC_Charac} that $\At$, $\Bt$ and $\bSigma$ have the form
\begin{align}
  \At = \left[ \Px~~\bA \right],~ \Bt = \left[ \Py~\bB \right],~
  \bSigma  =
  \diag{1,\sqrt{\lambda_1},\dots,\sqrt{\lambda_d}}, \label{eq:defAB}
\end{align}
and, consequently, the joint distribution can be written as 
\begin{equation} 
  \pxy(x,y) = \px(x)\py(y)+\sum_{k=1}^d
  \sqrt{\lambda_k}b_{y,k}a_{x,k},
\end{equation}
where $a_{x,k}$ and $b_{y,k}$ are the entries of $\bA$ and $\bB$  in
\eqref{eq:defAB}, respectively.


%

Theorem  \ref{thm:Bound} follows directly from the next two lemmas.

\begin{lem}
  \label{lem:fanobound1}
Let the marginal distribution $\Px$ and the PICs
$\blambda=(\lambda_1,\dots,\lambda_d)$ be
given, where   $d=m-1$. Then
for any $\pxy\in \mathcal{R}(\Px,\blambda)$, $0\leq \alpha \leq 1$ and $0\leq \beta \leq \px(2)$
\begin{equation*} 
  P_e(X|Y) \geq 1-\beta -
  \sqrt{f_0(\alpha,\Px,\blambda)+\sum_{i=1}^m\left(\left[\px(i)-\beta\right]^+\right)^2},
\end{equation*}
where
\begin{align}
  f_0(\alpha,\Px,\blambda)=&  \sum_{i=2}^{d+1}
  \px(i)(\lambda_{i-1}+c_{i}-c_{i-1})\nonumber\\
  &+\px(1)(c_1 + \alpha) - \alpha\Px^T\Px~,
  \label{eq:f0_def}
\end{align}
and $c_i = \left[ \lambda_{i}-\alpha \right]^+$ for $i=1,\dots,d$ and
$c_{d+1}=0$.
\end{lem}
    
\begin{proof}
Let $X$ and $Y$ have a joint distribution matrix $\bP$ with marginal $\px$ and
principal inertias individually bounded by $\blambda =
(\lambda_1,\dots,\lambda_d)$. We assume without loss of generality that
$d=m-1$, where $|\calX|=m$. This can always be achieved by adding inertia components equal to
0.
  
Consider $X\rightarrow Y\rightarrow \hat{X}$, where $\hat{X}$ is the estimate of $X$
from $Y$.  The mapping from $Y$ to $\hat{X}$ can be described without loss of
generality by a
$|\calY|\times|\calX|$ row stochastic matrix, denoted by $\Emat$, where the
$(i,j)$-th entry is the probability $p_{\hat{X}|Y}(j|i)$. The probability of
correct estimation $P_c$ is then
\begin{equation*} 
	P_c = 1-P_e(X|Y)=\Tr{\bP_{X,\Xh}},
\end{equation*}
where $\bP_{X,\Xh} \defined \bP \mathbf{F}$.

The matrix $\Pxxp$ can be decomposed according to \eqref{eq:comactDecompPxy},
resulting in 
\begin{align}
  P_c  = \Tr{\Dx^{1/2}\bU \Sigmat \bV^T \bD_{\Xh}^{1/2}}= \Tr{\Sigmat  \bV^T \bD_{\Xh}^{1/2} \Dx^{1/2}\bU}, \label{eq:PxxpDecomp}
\end{align}
where  
\begin{align*}
  \bU &= \left[ \Px^{1/2}~~\bu_2~\cdots~\bu_m \right],\\
  \bV &= \left[ \Pxp^{1/2}~~\bv_2~\cdots~\bv_m  \right],\\
  \Sigmat &= \diag{1,\sqrt{\lambdat_1},\dots,\sqrt{\lambdat_d}},\\
  \bD_{\Xh} & = \diag{\Pxp},
\end{align*}
and $\Ut$ and $\Vt$ are orthogonal matrices. The probability of correct
detection can be written as
\begin{align*}
  P_c &= \Px^T\Pxp + \sum_{k=2}^m\sum_{i=1}^m
  \left(\lambdat_{k-1}\px(i)\pxp(i)\right)^{1/2}u_{k,i}v_{k,i} \nonumber\\
    &= \Px^T\Pxp + \sum_{k=2}^m\sum_{i=1}^m
    \lambdat_{k-1}^{1/2} \tu_{k,i}\tv_{k,i}
\end{align*}
where $u_{k,i}=[\bu_k]_i$, $v_{k,i}=[\bv_k]_i$, $\tu_{k,i}=\sqrt{\px(i)}u_{k,i}$ and
$\tv_{k,i}= \sqrt{\pxp(i)}v_{k,i}$. Applying the Cauchy-Schwarz inequality twice, we
obtain
\begin{align}
  P_c &\leq \Px^T\Pxp + \sum_{i=1}^m  \left(\sum_{k=2}^m \tv_{k,i}^2
  \right)^{1/2}\left(\sum_{k=2}^m  \lambdat_{k-1} \tu_{k,i}^2
  \right)^{1/2}\nonumber\\
  & =  \Px^T\Pxp +  \sum_{i=1}^m \left(\pxp(i)(1-\pxp(i))\sum_{k=2}^m
  \lambdat_{k-1} \tu_{k,i}^2
  \right)^{1/2}\nonumber\\
  &\leq  \Px^T\Pxp +\left(1-\sum_{i=1}^m \pxp(i)^2 \right)^{1/2}\left(
  \sum_{i=1}^m \sum_{k=2}^m  \lambdat_{k-1} \tu_{k,i}^2 \label{eq:LeftBound}
  \right)^{1/2}.
\end{align}
Let $\olU=[\bu_2 \cdots \bu_m]$ and $\Sigmat =
\diag{\lambdat_1,\dots,\lambdat_d}$. Then
\begin{align}
  \sum_{i=1}^m \sum_{k=2}^m  \lambdat_{k-1} \tu_{k,i}^2 &= \Tr{\bSigma\olU^T
  \Dx\olU}\nonumber \\
  &\leq \sum_{k=1}^d\sigma_k \lambdat_k, \nonumber \\
  &\leq \sum_{k=1}^d\sigma_k \lambda_k. \label{eq:prodVonNeum}
\end{align}
where $\sigma_k$ are the singular values of $\olU^T \Dx\olU$. The first inequality follows from the
application of Von-Neumman's trace inequality \cite[Thm. 7.4.1.1]{horn_matrix_2012} and the fact that $\olU^T \Dx\olU$ is symmetric and positive
semi-definite. The second inequality follows by observing that the PICs satisfy the DPI and, therefore, $\lambdat_k\leq \lambda_k$.

We will now find an upper bound for \eqref{eq:prodVonNeum} by bounding the
eigenvalues $\sigma_k$. First, note that $\olU
~\olU^T =\eye-\Px^{1/2}\left(\Px^{1/2}\right)^T$ and consequently
\begin{align}
  \sum_{k=1}^d \sigma_k& = \Tr{\olU^T \Dx\olU} \nonumber \\
  & = \Tr{\Dx\left(\eye-\Px^{1/2}\left(\Px^{1/2}\right)^T\right)} \nonumber\\
  & = 1-\sum_{i=1}^m\px(i)^2~.\label{eq:sumbound}
\end{align}
Second, note that $\olU^T \Dx\olU$ is a principal submatrix of $\bU^T \Dx \bU$,
formed by removing the first row and columns of $\bU^T \Dx \bU$. It then follows
from Cauchy's interlacing theorem \cite[Theorem 4.3.17]{horn_matrix_2012} that
\begin{equation}
  \px(m)\leq \sigma_{m-1} \leq \px(m-1)\leq \dots\leq \px(2)\leq
  \sigma_1 \leq \px(1). \label{eq:cauchyLace}
\end{equation}

Combining  \eqref{eq:sumbound} and \eqref{eq:cauchyLace}, an
upper bound for \eqref{eq:prodVonNeum} can be found by solving the following
linear program
\begin{align}
  \max_{s_i}~~& \sum_{i=1}^d \lambda_i s_i \label{eq:boundLP}\\
  \mbox{subject to}~~&\sum_{i=1}^d s_i =1-\Px^T\Px, \nonumber\\
                        & \px(i+1)\leq s_i \leq
                        \px(i),~i=1,\dots,d~.\nonumber
\end{align}

Let $\delta_i \defined \px(i)-\px(i+1)$ and $\gamma_i \defined \lambda_i\px(i+1)$. The dual of \eqref{eq:boundLP} is
\begin{align}
  \min_{y_i,\mu}~~& \mu\left(\px(1)-\Px^T\Px\right)+
  \sum_{i=1}^{m-1}\delta_i y_i+\gamma_i  \label{eq:DualboundLP}\\
  \mbox{subject to}~~&y_i\geq \left[ \lambda_i-\mu\right]^+,~i=1,\dots,d~.\nonumber
\end{align}
For any given value of $\mu$, the optimal values of the dual variables $y_i$
in  \eqref{eq:DualboundLP} are
\begin{equation*}
y_i =  \left[ \lambda_i-\mu\right]^+=c_i,~i=1,\dots,d~.
\end{equation*}
Therefore the linear program \eqref{eq:DualboundLP} is equivalent to
\begin{equation}
  \min_{\mu}   f_0(\mu,\Px,\blambda), \label{eq:DualImproved}
\end{equation}
where $f_0(\mu,\Px,\blambda)$ is defined in the statement of the lemma.

Denote the solution of \eqref{eq:boundLP} by $f_P^*(\Px,\blambda)$ and of
\eqref{eq:DualboundLP} by $f_D^*(\Px,\blambda)$. It follows that
\eqref{eq:prodVonNeum} can be bounded 
\begin{align}
\sum_{k=1}^d\sigma_k \lambda_k &\leq f_P^*(\Px,\blambda)  =f_D^*(\Px,\blambda) \leq f_0(\alpha,\Px,\blambda)~\forall~\alpha
                                \in \Reals. \label{eq:boundSumProd}
\end{align}
We may consider $0\leq\alpha \leq 1$ in
\eqref{eq:boundSumProd} without loss of generality.

Using \eqref{eq:boundSumProd} to bound \eqref{eq:LeftBound}, we find
\begin{equation}
  P_c\leq \Px^T\Pxp + \left[f_0(\alpha,\Px,\blambda)\left(1-\sum_{i=1}^m
  \pxp(i)^2 \right)\right]^{1/2} \label{eq:postLPBound}
\end{equation}
The previous bound can be maximized over all possible output distributions
$\pxp$ by solving:
\begin{align}
  \max_{x_i}~~& \left[f_0(\alpha,\Px,\blambda)\left(1-\sum_{i=1}^m
  x_i^2\right)\right]^{1/2} +\sum_{i=1}^m \px(i)x_i  \label{eq:NLP} \\
  \mbox{subject to}~~&\sum_{i=1}^m x_i = 1,\nonumber\\
                        & x_i \geq 0, i=1,\dots,m~.\nonumber
\end{align}
The dual function of \eqref{eq:NLP} over the constraint $\sum_{i=1}^m x_i = 1$ is
\begin{align}
  L(\beta) &=\max_{x_i\geq 0}~~\beta+ \left[f_0(\alpha,\Px,\blambda)\left(1-\sum_{i=1}^m
  x_i^2\right)\right]^{1/2} \nonumber \\
  &\hspace{.5in} +\sum_{i=1}^m (\px(i)-\beta)x_i \nonumber \\
&=\beta + \sqrt{f_0(\alpha,\Px,\blambda) +\sum_{i=1}^m \left(
\left[\px(i)-\beta\right]^+\right)^2}. 
\end{align}
Since $L(\beta)$ is an upper bound of \eqref{eq:NLP} for any $\beta$ and, therefore,
is also an upper bound of  \eqref{eq:postLPBound}, then
\begin{equation}
P_c \leq \beta + \sqrt{f_0(\alpha,\Px,\blambda) +\sum_{i=1}^m \left(
\left[\px(i)-\beta\right]^+\right)^2}. \label{eq:Pc_anyBeta}
\end{equation}
Note that we can consider $0\leq \beta\leq \px(2)$ in \eqref{eq:Pc_anyBeta},
since $L(\beta)$ is increasing for $\beta>\px(2)$. Taking $P_e(X|Y) = 1-P_c$, the result follows.
\end{proof}

The next result tightens the bound introduced in Lemma \ref{lem:fanobound1} by
optimizing over all values of $\alpha$.

\begin{lem}
  \label{lem:f0tight}
  Let $f^*_0(\Px,\blambda)\defined\min_{\alpha} f_0(\alpha,\Px,\blambda)$ and
  $k^*$ be defined as in \eqref{eq:kstar}. Then 
  \begin{align}
    f^*_0(\Px,\blambda)=&  \sum_{i=1}^{k^*}\lambda_i \px(i)+
    \sum_{i=k^* + 1}^{m}\lambda_{i-1} \px(i)  - \lambda_{k^*} \Px^T\Px~, 
  \end{align}
  where $\lambda_m=0$.
\end{lem}

\begin{proof}
  Let $\Px$ and $\blambda$ be fixed, and $\lambda_k \leq \alpha \leq
  \lambda_{k-1}$, where we define for ease of notation $\lambda_0\defined 1$ and $\lambda_m \defined 0$ (recall that the PICs correspond to $\lambda_1,\dots,\lambda_{m-1}$).  Then $c_i =\lambda_i-\alpha$
  for $1\leq i\leq k-1$ and $c_i = 0$ for $k\leq i \leq d$ in \eqref{eq:f0_def}.
  Therefore
  \begin{align}
    f_0(\alpha,\Px,\blambda)=&\sum_{i=1}^{k-1} \lambda_i \px(i) +
    \alpha \px(k) 
    + \sum_{i=k+1}^m \lambda_{i-1}\px(i)-\alpha\Px^T\Px. \label{eq:f0_fixaplha}
  \end{align}

Note that \eqref{eq:f0_fixaplha} is convex in $\alpha$, and is decreasing
when $\px(k)-\Px^T\Px\leq 0$ and increasing when   $\px(k)-\Px^T\Px\geq 0$.
Therefore,  $f_0(\alpha,\Px,\blambda)$ is minimized when $\alpha=\lambda_{k}$ such
that $\px(k)\geq \Px^T\Px$ and $\px(k-1)\leq \Px^T\Px$. If $\px(k)-\Px^T\Px\geq
0$ for all $k$ (i.e. $\px$ is uniform), then we can take $\alpha=0$. Theorem \ref{thm:Bound} follows directly.
\end{proof}

\subsection*{Theorem \ref{thm:schur}}
\begin{proof}
Consider two probability distributions $\px$ and $\qx$ defined over $\calX=\{1,\dots,m\}$, and assume that $\px$ majorizes $\qx$, i.e. $\sum_{i=1}^k \qx(i)\leq \sum_{i=1}^k \px(i)$ for $1\leq k \leq m$. Therefore $q_X$ is a convex combination of permutations of $\px$ \cite{marshall_inequalities:_2011}, and can be written as $q_X  = \sum_{i=1}^l a_i  (\px\circ \pi_i)$ for some $l\geq 1$, where $a_i\geq 0$, $\sum a_i = 1$ and $\pi_i$ is a permutation of $\px$, i.e. $ \px\circ \pi_i = p_{\pi_i(X)}$. Hence, for a fixed $\pxhgx$:
\begin{align}
\calI(\qx,\pxhgx)&=\calI\left(  \sum_{i=1}^l a_i (\px\circ \pi_i) ,\pxhgx\right) \nonumber\\
             &\geq \sum_{i=1}^l a_i \calI(\px\circ \pi_i,\pxhgx),\nonumber \\
             & =  \sum_{i=1}^l a_i \calI( \px,\pi_i\circ \pxhgx), \label{eq:Iconvex}
\end{align}
where the inequality follows from the concavity assumption and the last equality from $\calI(X;\Xh)$
being invariant to one-to-one mappings of $X$ and $\Xh$. Consequently, from the definition of error-rate function in   Defn. \ref{defn:ei}, 
\begin{align*}
  e_\calI(\qx,\theta)&= \inf_{\pxhgx}\left\{ \sum_{x,x'\in [m]} d_H(x,x')q_X(x)
                        \pxhgx(x'|x) \middle| \calI( q_X,\pxhgx)\leq \theta \right\} \\
                    &\stackrel{\mathrm{(a)}}{\geq} \inf_{\pxhgx}\left\{ \sum_{i\in [l]}a_i\sum_{x,x'\in [m]} d_H(\pi_i(x),x')p_X(x)\pxhgx(x'|\pi_i(x)) \middle|  \sum_{i\in [l]} a_i \calI( p_X,\pi_i\circ \pxhgx)\leq \theta \right\} \\
                    &\stackrel{\mathrm{(b)}}{=} \inf_{\pxhgx}\left\{ \sum_{i\in [l]}a_i\sum_{x,x'\in [m]} d_H(\pi_i(x),\pi_i(x'))p_X(x)\pxhgx(\pi_i(x')|\pi_i(x)) \right.\\
                    &\hspace{2in}\left. \middle|  \sum_{i\in [l]} a_i \calI( p_X,\pi_i\circ \pxhgx\circ \pi_i)\leq \theta \right\} \\
                    &\stackrel{\mathrm{(c)}}{\geq} \inf_{\pxhgx^1,\dots,\pxhgx^l}\left\{ \sum_{i\in [l]}a_i\sum_{x,x'\in [m]} d_H(x,x')p_X(x)\pxhgx^i(x|x') \middle|  \sum_{i\in [l]} a_i \calI( p_X,\pxhgx^i)\leq \theta \right\} \\
                    & \stackrel{\mathrm{(d)}}{=} \inf_{\theta_1,\dots,\theta_l\geq 0} \left\{ \sum_{i=1}^l a_i e_\calI(\px,\theta_i) \middle| \sum_{i=1}^l a_i\theta_i = \theta \right\}\\
  &\stackrel{\mathrm{(e)}}{\geq} \inf_{\theta_1,\dots,\theta_l\geq 0} \left\{e_\calI
\left(\px, \sum a_i \theta_i \right) \middle| \sum_{i=1}^l a_i\theta_i = \theta
\right\}\\
    &= e_\calI\left(\px, \theta \right),
\end{align*}
where inequality (a) follows from \eqref{eq:Iconvex},  (b) follows from the fact that the infimum is taken over all mapping $\pxhgx$ and that $\calI(X;\Xh)$ is invariant to one-to-one mappings of $X$ and $\Xh$,  (c) follows by allowing a mapping $\pxhgx^i$ to be independently minimized for each $i$ (as opposed to minimizing the same mapping $\pxhgx$ for all $i$),  (d) is obtained by noting that the optimal choice of $\pxhgx^i$ is the one that minimizes the Hamming distortion $d_H$ for a given upperbound on $\calI(\px,\pxhgx^i)$, and (e) follows from the convexity of $e_I(\px,\theta)$ in $\theta$. 
 Since this holds for any $\qx$ that is majorized by $\px$, $ e_\calI(\px,\theta)$ is Schur-concave. 
\end{proof}

\section{Proofs from Section \ref{chap:PIC_Priv}}
\label{app:PICpriv}
\subsection*{Lemma \ref{lem:increase}}
\begin{proof}
  Let $p_{S,X}$ and $p_{Y|X}$ be given, with $S\rightarrow X \rightarrow Y$. Denote by $\bw_i$ the vector in the $|\calX|$-simplex with entries
  $p_{X|Y}(\cdot|i)$. Furthermore, let $a_i \defined H(X)-H(X|Y=i)$, and
  $b_i\defined H(S)-H(S|Y=i)$. Therefore
    \begin{equation}
    \label{eq:points}
     \sum_{i=1}^{|\calY|} p_Y(i)\left[\bw_i,a_i,b_i\right] =
     \left[\Px,I(X;Y),I(S;Y)\right].
    \end{equation}
 Since $\bw_i$ belongs to the $|\calX|$-simplex, the vector $\left[\bw_i,a_i,b_i\right]$ is taken from a connected, compact $|\calX|+1$ dimensional space. Then, from Fenchel-Eggleston strengthening of Carath\'eodory's theorem \cite[Theorem 18, pg. 35]{eggleston_convexity_2009},  the point $\left[\Px,I(X;Y),I(S;Y)\right]$ can also be achieved by at most $|\calX|+1$ non-zero values of $p_Y(i)$. It follows directly that it is sufficient to consider $|\calY|\leq |\calX|+1$ for the mappings that approach the infimum  $G_I(t,p_{S,X})$ in \eqref{eq:defGI}. The set of all mappings $p_{Y|X}$ for $|\calY|\leq |\calX|+1$ is compact, and both $p_{Y|X}\to I(S;Y)$ and $p_{Y|X}\to I(X;Y)$  are continuous and bounded when $S$, $X$ and $Y$ have finite support. Consequently, the infimum in \eqref{eq:defGI} is attainable.

 For $0<t\leq H(X)$ and $p_{S,X}$ fixed, let $G_I(t,p_{S,X})=\alpha$. From the discussion above, there exists $p_{Y|X}$ that achieves $I(S;Y)= \alpha$ for
  $I(X;Y)\geq t$. Now consider $p_{\Yt|X}$ where $\tilde{\calY}=[|\calY|+1]$ and,
  for $0< \lambda \leq 1$,
  \begin{equation*}
    p_{\Yt|X}(y|x)=(1-\lambda)\ones_{\{y=|\calY|+1\}}+\lambda\ones_{\{y\neq|\calY|+1\}}
    p_{Y|X}(y|x).
  \end{equation*}
  Note that  $\Yt$ can be understood as an erased version of $Y$, with the erasure symbol being
  $|\calY|+1$. It follows  (cf. \cite[Sec. 7.1.5]{cover_elements_2006}) that $I(S;\Yt)=\lambda I(S;Y)= \lambda
  \alpha$ and  $I(X;\Yt)=\lambda I(X;Y)\geq \lambda t$. We have thus explicitly constructed a new mapping $p_{\Yt|X}$ that satisfies $S\to X\to \tilde{Y}$ and achieves $I(S;\tilde{Y})=\lambda \alpha$ and $I(X;\tilde{Y})\geq \lambda t$. Therefore, from the definition of $G_I$ in \eqref{eq:defGI}, $G_I(\lambda t,p_{S,X})\leq \lambda \alpha=\lambda I(S;Y)$. Consequently, 
  \begin{equation}
  \label{eq:Gupperb}
    \frac{G_I(\lambda t,p_{S,X})}{\lambda t}\leq \frac{\lambda I(S;Y)}{\lambda t}
    = \frac{G_I(t,p_{S,X} )}{t}.
  \end{equation}
  Since this holds for any $0<\lambda\leq 1$, then $\frac{G_I(t,p_{S,X} )}{t}$ is non-decreasing in $t$. Finally, for a fixed $p_{S,X}$, the set of points $(\bw_i,a_i,b_i)\in \Reals^{|\calX|+2}$ that satisfies \eqref{eq:points} is convex, and thus, for a fixed $\bp_X$, it's lower-boundary, which corresponds to the graph of $(t,G_I(t,p_{S,X}))$, is convex.
\end{proof}

\subsection*{Lemma \ref{thm:var}}
\begin{proof}
     For fixed $p_{Y|X}$ and $p_{S,X}$, and assuming $I(X;Y)>0$,
    \begin{align*}
      \frac{I(S;Y)}{I(X;Y)}& = \frac{\sum_{y\in \calY} p_Y(y)D(p_{S|Y=y}||p_S)
       }{\sum_{y\in \calY} p_Y(y)D(p_{X|Y=y}||p_X) } \\
       &\geq \min_{\substack{y\in \calY:\\D(p_{X|Y=y}||p_X) >0}}  \frac{D(p_{S|Y=y}||p_S)
       }{D(p_{X|Y=y}||p_X) }\\
       &\geq \inf_{q_X\neq p_X} \frac{D(q_S||p_S)}{D(q_X||p_X)}.
    \end{align*}

   Now let $d^*$ be the infimum in the right-hand side of
   \eqref{eq:divergence}, and $q_X$ satisfy \[\frac{D(q_Y||p_Y)}{D(q_X||p_X)}=d^*
   +\delta,\] where $\delta>0$. For $\epsilon>0$ and sufficiently small, let
   $p_{Y|X}$ be such that $\calY=[2]$, $p_Y(1)=\epsilon$, $p_{X|Y}(x|1)=q_X(x)$
   and \[p_{X|Y}(x|2)
   =\frac{1}{1-\epsilon}p_X(x) - \frac{\epsilon}{1-\epsilon} q_X(x).\]  Since for
   any distribution $r_X$ with support $\calX$ we have $D\left( (1-\epsilon)p_X+\epsilon
   r_X||p_X \right)=o(\epsilon)$, we find 
   \begin{align*}
        I(S;Y)&= \epsilon D(p_{S|Y=1}||p_S) + (1-\epsilon) D(p_{S|Y=0}||p_S)  \\
              &=  \epsilon D(q_S||p_S)  +o(\epsilon),        
   \end{align*}
   and equivalently, $I(X;Y)= \epsilon D(q_X||p_X)  +o(\epsilon)$. Consequently,
   \begin{equation*}
   \frac{I(S;Y)}{I(X;Y)}=\frac{\epsilon D(q_S||p_S)  +o(\epsilon) }{\epsilon
   D(q_X||p_X) + o(\epsilon) } \to d^*+\delta,
 \end{equation*}
   where the limit is taken as $\epsilon\to 0$. Since this holds for any
   $\delta>0$, then $\vs\leq d^*$, proving the result. 
\end{proof}

\subsection*{Lemma \ref{lem:delta}}
\begin{proof}
  Let $f:\calX\to \Reals$, $\EE{f(X)}=0$ and
  $\|f(X)\|^2_2$=1, and $\fb\in \Reals^{|\calX|}$ be a vector with entries $f_i
  = f(i)$ for $i\in \calX$. Observe that 
    \begin{align*}
      \|\EE{f(X)|S}\|_2^2&=\sum_{s\in \calS} p_S(s)\EE{f(X)|S=s}^2\\
                         &= \fb^T \bP_{X|S}^T \bD_S \bP_{X|S} \fb^T\\
                         &= \fb^T \bD_X^{1/2} \bQ^T \bQ \bD_X^{1/2}\fb\\
                         &\geq \dps,
    \end{align*}
    where the last inequality follows by noting that $\bx\defined \fb^T
    \bD_X^{1/2} $ satisfies $\|\bx\|_2=1$ and that $\dps$ is the smallest
    eigenvalue of the positive semi-definite matrix $\bQ^T\bQ$, where $\bQ$ was defined in Definition \ref{def:PIC_4} as $\bQ\defined \bD_S^{-1/2} \bP_{X,S}\Dx^{-1/2}$. 
\end{proof}

\subsection*{Lemma \ref{thm:Ineq1}}
\begin{proof}
  Let $p_{S|X}$ be fixed, and define \[g_\lambda(p_X)\defined H(S)-\lambda
  H(X),\] where $H(S)$ and
    $H(X)$ are the entropy of $S$ and $X$, respectively, when $(S,X)\sim
    p_{S|X}p_X$. For $0<\epsilon\ll 1$, let \[p_\epsilon(i)\defined p_X(i)(1+\epsilon
    f(i))\] be a perturbed version of $p_X$, where $\EE{f(X)}=0$ and, w.l.o.g.,
    $\|f(X)\|_2=1$. The second derivative of $g_\lambda(p_\epsilon)$ at
    $\epsilon=0$ is\footnote{This was observed in \cite{anantharam_maximal_2013} and 
    \cite{kamath2012non}, and follows directly from $-\frac{\partial^2}{\partial
    \epsilon^2}a(1+b\epsilon)\log_2 a(1+b\epsilon)=-b^2a\log_2(e)$.}
    \begin{align}
       \frac{\partial^2 g_\lambda(p_\epsilon) }{ \partial
       \epsilon^2}\bigg|_{\epsilon=0}&=\log_2(e)\left( -\|\EE{f(X)|S} \|_2^2+\lambda\|f(X)\|_2^2
       \right) \nonumber \\
                   &=\log_2(e)\left( -\|\EE{f(X)|S} \|_2^2+\lambda\right).
                   \label{eq:Hess}
    \end{align}
    Thus, from Lemma \ref{lem:delta}, if $\lambda \leq \dps$ then for any
    sufficiently small perturbation of $p_X$, \eqref{eq:Hess} will be
    non-positive. Conversely, if $\lambda>\dps$, then we can find a perturbation
    $f(X)$ such that \eqref{eq:Hess} is positive. Therefore,
    $g_{\lambda}(p_X)$ has a negative semi-definite Hessian if and only if
    $0\leq \lambda\leq \dps$. 
    
    For any $S\to X\to Y$, we have $I(S;Y)/I(X;Y)\geq \vs$, and, consequently,
    for  $0\leq \lambda^\dag\leq \vs$, 
    \begin{align*}
        g_{\lambda^\dag}(p_X) \geq H(S|Y)-\lambda^\dag H(X|Y),
      \end{align*}
    and  $g_{\lambda^\dag}(p_X)$ touches the upper-concave envelope of
    $g_{\lambda^\dag}$ at $p_X$. Since a function must be concave at the points where it matches its concave envelope, $g_{\lambda^\dag}$ has a negative
    semi-definite Hessian at $p_X$ and, from \eqref{eq:Hess}, $\lambda^\dag \leq
    \dps$. Since this holds for any $0\leq \lambda^\dag\leq \vs$, we find $\vs \leq
    \dps$. 

      For a fixed $p_{S|X}$, the function $g_{\lambda}(p_X)$ is concave when $\lambda=0$ and convex when
$\lambda=1$. Consequently, the maximum $\lambda$
for which $g_{\lambda}(p_X)$ has a negative Hessian at $p_X$ is $\dps$.
Furthermore, Lemma \ref{thm:var} implies that a value $\lambda_1$ for which
$g_{\lambda}(p_X)$ touches its lower concave envelope at $p_X$ for all $\lambda_1 \ge \lambda$ is
$\vs$.  Therefore, both $\inf_{p_X} \vs $ and $\inf_{p_X} \dps$ equal the
maximum value of $\lambda$ such that the function $g_{\lambda}(p_X)$ is concave
at all values of $p_X$.
Therefore, we established that for a given $p_{S|X}$,  \[\inf_{p_X} \vs =
\inf_{p_X} \dps. \]
  \end{proof}

\subsection*{Lemma \ref{thm:iff}}
\begin{proof}
    Theorem \ref{thm:Ineq1} immediately gives $\dps=0\Rightarrow \vs=0$. Let $\vs=0$. Then,
    since $D(q_X||p_X)\leq -\min_{i\in \calX} \log_2 p_X(i)$ and $\calX$ is
    finite,  Lemma \ref{thm:var} implies that for any $\epsilon> 0$ there exists
    $q_X$ and $0<\delta\leq  -\min_{i\in \calX} \log_2 p_X(i)$ such
    that \[D(q_X||p_X)\geq \delta>0\] and \[D(q_S||p_S)<\epsilon.\] We can then
    construct a sequence $q_X^1,q_X^2,q_X^3,\dots$ such that $q_X^i\neq p_X$,
    $D(q_S^k||p_S)\leq \epsilon_k$ and \[\lim_{k \to \infty}\epsilon_k=0.\] Let
    $\bq_S^k$ be a vector whose entries are $q_S^k(\cdot)$. Then, from Pinsker's
    inequality,
    \begin{align}
      \label{eq:epk}
      \epsilon_k \geq \frac{1}{2} \|\bq_S^k-\bp_S\|_1^2\geq \frac{1}{2}
      \|\bq_S^k-\bp_S\|_2^2.
    \end{align}
    Defining $\bx^k =\bq_X^k-\bp_X$, observe that $0<\|\bx^k\|_2^2\leq 2$ and,
    from \eqref{eq:epk}, $\|\bP_{S|X} \bx^k\|_2 \leq \sqrt{2\epsilon_k}$.
    Hence,
    \begin{equation}
      \label{eq:limit}
      \lim_{k\to \infty} \frac{\|\bP_{S|X} \bx^k\|_2^2  }{\|\bx^k\|_2^2}=0.
    \end{equation}
    In addition, denoting $s_m\defined \min_{s\in \calS}p_S(s)$ and $x_M\defined
    \min_{x\in \calX}p_X(x)$,  for each $k$ we have
    \begin{align}
        \frac{\|\bP_{S|X} \bx^k\|_2^2  }{\|\bx^k\|_2^2} &\geq
        \min_{\|\by\|_2^2> 0}  \frac{\|\bP_{S|X} \by\|_2^2
      }{\|\by\|_2^2} \nonumber \\
      &= \min_{\|\by\|_2^2> 0}  \frac{\|\bP_{S,X}\bD_X^{-1/2} \by\|_2^2
    }{\|\bD_X^{1/2}\by\|_2^2} \label{eq:deriv1}\\
    &\geq  \min_{\|\by\|_2^2> 0}  \frac{s_m\|\bD_S^{-1/2}\bP_{S,X}\bD_X^{-1/2} \by\|_2^2
    }{x_M\|\by\|_2^2} \label{eq:deriv2}\\
    &= \frac{s_m}{x_M} \min_{\|\by\|_2^2> 0}  \frac{\|\bQ \by\|_2^2
    }{\|\by\|_2^2}  \label{eq:deriv3} \\
    &= \frac{s_m\dps}{x_M} \label{eq:deriv4}.
    \end{align}
    In the derivation above, \eqref{eq:deriv1} follows from $\bD_X$ being invertible (by
    definition), \eqref{eq:deriv2} is a direct consequence of
    $\|\bD_S^{-1/2}\by\|_2^2\leq s_m^{-1}\|\by\|_2^2$ and
    $\|\bD_X^{1/2}\by\|_2^2\leq x_M\|\by\|_2^2$ for any $\by$,
    and \eqref{eq:deriv3} and \eqref{eq:deriv4}  follow from the definition of
    $\bQ$ and  $\dps$, respectively. Combining \eqref{eq:deriv4}  with
    \eqref{eq:limit}, it follows that $\dps=0$, proving the desired result.
  \end{proof}

\subsection*{Corollary \ref{cor:explicit}}
\begin{proof}
 If $\dps=0$, then the lower bound for $t^*$ follows  from the construction used in \eqref{eq:t0} and, more specifically, by (i) maximizing the right-hand side of \eqref{eq:t0} across all functions in $\calF_0$ and (ii) observing that the maximum value of $\epsilon$ such that $p_{Y|X}$ is non-negative is $\epsilon=1/2\|f\|_\infty$. If
    $\dps>0$, then $\calF_0$ is singular (i.e. $\calF_0=\{w_0\}$), and the lower bound \eqref{eq:tstar}
    reduces to the trivial bound $t^*\geq 0$.

    In order to prove that the lower bound is sharp, consider $S$ being an unbiased bit,
    drawn from $\{1,2\}$,
    and $X$ the result of sending $S$ through an erasure channel with erasure
    probability $1/2$ and $\calX=\{1,2,3\}$, with $3$ playing the role of the
    erasure symbol. Let 
    \begin{equation*}
        f(x)\defined
        \begin{cases}
            1, &x\in\{1,2\},\\
            -1 &x=3.
        \end{cases}
    \end{equation*}
    Then $f\in \calF_0$, $h_b \left(\frac{1}{2}+  \frac{f(x)}{2 \|f\|_\infty}
    \right)=0$ for $x\in\calX$ and $t^*=1$. But, from Lemma \ref{lem:region}, $t^*\leq
    H(X|S)=1$. The result follows.
\end{proof}